\PassOptionsToPackage{svgnames}{xcolor}
\documentclass[11pt,reqno]{article}
\usepackage{amsmath,hyperref,amssymb, amsthm}
\usepackage{authblk}
\usepackage[all]{xy}
\usepackage{graphicx,url}
\usepackage{color}
\usepackage{slashed}
\usepackage{tikz}
\usepackage{lipsum}
\usepackage{caption}
\usepackage[normalem]{ulem}

\usepackage[style=ieee]{biblatex}
\newcommand{\be}{\begin{eqnarray}}
\newcommand{\ee}{\end{eqnarray}}
\newcommand{\eeq}{\end{equation}}
\newcommand{\beq}{\begin{equation}}

\allowdisplaybreaks \numberwithin{equation}{section}

\DeclareSymbolFont{AMSa}{U}{msa}{m}{n}
\DeclareSymbolFont{AMSb}{U}{msb}{m}{n}
\DeclareMathSymbol{\fieldR}{\mathalpha}{AMSb}{"52}

\renewcommand{\Im}{\imag}
\DeclareMathOperator{\imag}{Im}

\newcommand{\CA}{{\cal A}}
\newcommand{\CB}{{\cal B}}

\newcommand{\CH}{\mathcal{H}}
\newcommand{\Hh}{\mathcal{H}}

\newcommand{\CI}{{\cal I}}
\newcommand{\CL}{{\cal L}}
\newcommand{\CN}{\mathcal{N}}
\newcommand{\CM}{\mathcal{M}}

\newcommand{\calC}{\mathcal{C}}
\newcommand{\calR}{\mathcal{R}}
\newcommand{\CT}{\mathcal{T}}

\newcommand{\dual}{\star}
\newcommand{\nL}{\mathsf{L}}

\DeclareMathOperator{\Tr}{Tr}

\newcommand{\NN}{\mathbb{N}}
\newcommand{\ZZ}{\mathbb{Z}}
\newcommand{\RR}{\mathbb{R}}
\newcommand{\CC}{\mathbb{C}}
\newcommand{\QQ}{\mathbb{Q}}

\newcommand{\cTop}{\mathsf{Top}}

\newcommand{\one}{\mathsf{1}}
\newcommand{\ii}{x}
\newcommand{\jj}{y}
\newcommand{\kk}{z}
\newcommand{\mi}{\underline{x}}
\newcommand{\mj}{\underline{y}}
\newcommand{\mk}{\underline{z}}
\newcommand{\mone}{\underline{\mathsf{1}}}

\newcommand{\rrangle}{\rangle\!\rangle}
\newcommand{\llangle}{\langle\!\langle}

\newtheorem{theorem}{Theorem}
\newtheorem{claim}[theorem]{Claim}
\newtheorem{conjecture}[theorem]{Conjecture}

\newtheorem{proposition}[theorem]{Proposition}

\DeclareMathOperator{\ch}{ch}

\usepackage{tcolorbox}
\tcbuselibrary{skins,breakable}
\usetikzlibrary{shadings,shadows}

    {\endtcolorbox}

    {\endtcolorbox}

    {\endtcolorbox}

\title{Topological defects in K3 sigma models}
\author[1]{Roberta Angius\thanks{roberta.angius@csic.es}}
\author[2]{Stefano Giaccari\thanks{stefano.giaccari@pd.infn.it}}
\author[3]{Roberto Volpato\thanks{volpato@pd.infn.it}}
{\small \affil[1]{\small  Instituto de F\'{\i}sica Te\'orica
IFT-UAM/CSIC,  C/ Nicol\'as Cabrera 13-15, Campus de Cantoblanco, 28049 Madrid, Spain}
\affil[2,3]{\small Dipartimento di Fisica e Astronomia `Galileo Galilei', Universit\`a di Padova \& INFN, sez. di Padova, Via Marzolo 8, 35131, Padova, Italy}}

\begin{document}

\maketitle
\abstract{
    We consider the topological defect lines commuting with the spectral flow and the $\CN=(4,4)$ superconformal symmetry in two dimensional non-linear sigma models on K3. By studying their fusion with boundary states, we derive a number of general results for the category of such defects. We argue that while for certain K3 models infinitely many simple defects, and even a continuum, can occur, at generic points in the moduli space the category is actually trivial, i.e. it is generated by the identity defect. Furthermore, we show that if a K3 model is at the attractor point for some BPS configuration of D-branes, then all topological defects have integral quantum dimension. We also conjecture that a continuum of topological defects arises if and only if the K3 model is a (possibly generalized) orbifold of a torus model. Finally, we test our general results in a couple of examples, where we provide a partial classification of the topological defects.
}
\pagebreak
\tableofcontents

\section{Introduction}

Defects in a quantum field theory (QFT) or in statistical lattice models and quantum spin systems are broadly defined as  inhomogeneities  localized on submanifolds of positive codimension, and  appear in a variety of physical settings, with the codimension $1$ case playing the distinctive role of interfaces between different theories. A more specific notion is the one of topological defects, which are assumed to be invariant under continuous deformations as long as the deformations do not move the inhomogeneities past other defects or field operator insertions.  Historically, in particular in the context of $2d$ QFT,  topological defect lines (TDLs) have been studied in Conformal Field Theories (CFTs), where conformal invariance provides a very restrictive framework, even more so for Rational CFTs (RCFTs),  because of their connection to Boundary Conformal Field Theories (BCFTs), twisted boundary conditions, and orbifolds \cite{Cardy:1986gw,Fuchs:2002cm, Fuchs:2003id, Fuchs:2004dz, Fuchs:2004xi,Oshikawa:1996dj, Oshikawa:1996ww, Petkova:2000ip, Verlinde:1988sn, Zuber:1986ng}. However most recent developments stem from the realization  that, in QFT with generic dimension $d$, topological defects provide a natural generalization of the notion of group symmetry, with the group elements corresponding to topological defects of codimension $1$ across which the value of field operators jumps by the respective symmetry actions. This picture has naturally led to considering both ``higher-form symmetries'', realized by topological defects of higher codimension ( see e.g. \cite{Gaiotto:2014kfa}), and  ``non-invertible symmetries'', corresponding to topological defects  which cannot be fused with another topological defect to produce the trivial defect.  In two dimensions, the group structure is in the latter case generalized by the one of fusion categories \cite{Bhardwaj:2017xup, Chang_2019, Thorngren:2019iar, Thorngren:2021yso}. The classical example of fusion category symmetry is the Ising one, defined by three simple TDLs: the trivial defect, the $\ZZ_2$ Ising symmetry and the  ``Kramers-Wannier'' duality defect in the Ising model \cite{Frohlich:2009gb, Frohlich:2006ch, Frohlich:2004ef, Kramers:1941kn}, which relates spin correlators in the Ising model at some inverse temperature to disorder/twist correlators at the Kramers-Wannier dual inverse temperature.  
TDLs have been widely studied in RCFTs \cite{Cardy:1986gw,Fuchs:2002cm, Fuchs:2003id, Fuchs:2004dz, Fuchs:2004xi,Oshikawa:1996dj, Oshikawa:1996ww, Petkova:2000ip, Verlinde:1988sn, Zuber:1986ng}, gapped boundaries of $(2+1)D$ topological field theories, and anyon chains, and found to encode nontrivial topological information about a theory, such as constraints on the operator spectrum of CFTs and renormalization group flows in gauge theories \cite{Komargodski:2020mxz}.  Non-invertible topological defects can also be considered in higher-dimensional theories where they also provide relevant information about the dynamics, so that they are by now considered a tool of paramount importance in unraveling the non-perturbative aspects of Quantum Field Theory.

One of the main hurdles in advancing our knowledge is the fairly limited number of theories where the fusion category of defects or at least a subset thereof are known. This is true even in the simple framework of 2d CFTs.
While in rational CFTs some general techniques have been developed, the much broader realm of non rational theories is largely unexplored. In fact, topological interfaces, which in general separate possibly different theories,  have been studied in a very limited number of examples, in particular for the free boson compactified on a circle and on an orbifold thereof  \cite{Bachas:2007td,Chang:2020imq, Fuchs:2007tx} and for $d$-dimensional torus models, where they are assumed to preserve a $\hat{u}(1)^{2d}$ current algebra \cite{Bachas:2012bj}.

In this article, we consider topological defects in two-dimensional superconformal field theories (SCFT) arising as supersymmetric non-linear sigma models with target space a K3 surface. More precisely, we will focus on the defects that preserve the full $\CN=(4,4)$ superconformal symmetry at central charge $c=\tilde c=6$, and that are invariant under the spectral flow transformations that relate the different (NS-NS, NS-R, R-NS, R-R) sectors of the theory. Non-linear sigma models on K3 (or K3 models, for short), provide the simplest examples of Calabi-Yau compactifications in type II string theory. A generic K3 model is not a rational CFT, and it cannot be solved exactly. Nevertheless, due to the large amount of space-time and worldsheet supersymmetries, many general results about these models are known, such as the geometry of the moduli space, the elliptic genus (which is the same for every K3 model), the spectrum of short $\CN=(4,4)$ representations, and even the finite groups of symmetries at each point in the moduli space \cite{Aspinwall:1996mn,Gaberdiel:2011fg,Nahm:1999ps}. 
For these reasons, K3 models represent the ideal framework to understand topological defects in a non-rational CFT, besides the torus models examples. An interesting analysis of the topological defects in some non-rational K3 models, using a different approach, recently appeared in \cite{Cordova:2023qei}; we will comment about the relationship with this article in the conclusions (see point 4 in  section \ref{s:conclusions}).

Because K3 models are not rational with respect to the $\CN=(4,4)$ superconformal algebra, one can expect infinitely many distinct simple defects. In particular, we will see some examples of K3 models where a continuum of simple non-invertible topological defects arise, a phenomenon that
has already been observed in orbifolds of torus models \cite{Chang:2020imq,Fuchs:2007tx,Thorngren:2021yso}. This implies that, strictly speaking, we are putting ourselves outside of the mathematical framework of fusion categories, at least in its most restrictive definitions. Nevertheless, we will assume that the basic properties of fusion categories still hold for the defects we consider. In particular, we will require that the fusion of any two defects is a superposition of finitely many simple defects. We denote by $\cTop_\calC$ the fusion category (in a broad sense) of topological defects in a K3 model $\calC$  preserving the $\CN=(4,4)$ superconformal algebra and the spectral flow. We will only focus on certain properties of this category, in particular on its fusion ring, on the behaviour of defects when moved past local operators, and on the fusion with boundary states. We will mostly ignore the detailed properties of the  fusion matrices.
The main constraints on the topological defects in $\cTop_\calC$ come from considering the fusion with boundary states representing 1/2 BPS D-branes, i.e. preserving half of the space-time supersymmetry of type II superstring compactified on $\calC$. The idea of studying defects in the presence of boundaries is a very natural one, and has been explored in a number of articles \cite{Fuchs:2001qc,Brunner:2007ur,Fredenhagen:2009tn,Kojita:2016jwe,Konechny:2019wff,Konechny:2020jym,Fukusumi:2021zme,Collier:2021ngi,Choi:2023xjw}. The 1/2 BPS D-branes we are interested in are always charged with respect to the $U(1)^{24}$ gauge group of R-R ground fields of the theory. One can argue that fusion with defects in $\cTop_\calC$ preserves the set of such boundary states. This means that each defect $\CL\in \cTop_\calC$ can be associated with a $\ZZ$-linear map (endomorphism) $\nL\in {\rm End}(\Gamma^{4,20}_{R-R})$ on the even unimodular lattice $\Gamma^{4,20}_{R-R}$ of D-brane R-R charges. The endomorphism $\nL$ is further constrained by the requirement that the defect cannot mix R-R ground fields belonging to different representations of the $\CN=(4,4)$ algebra. The map $\cTop_\calC \to {\rm End}(\Gamma^{4,20}_{R-R})$ gives rise to a ring homomorphism from the fusion ring of $\cTop_\calC$ the ring ${\rm End}(\Gamma^{4,20}_{R-R})$ (or rather the subring of endomorphisms satisfying suitable properties), so that properties of the former ring can be deduced by studying the latter. One can argue that the property of a defect $\CL$ to be preserved by some marginal deformation of the model depends only on $\nL$. 

Using this simple idea, we are able to derive several properties of topological defects in generic K3 models (see in particular section \ref{s:topdefsK3}).  We argue that $\cTop_{\calC}$ is trivial (i.e. the only simple defect is the identity) in most K3 models $\calC$, except a subset with null measure in the moduli space (Claim \ref{th:generic}). While the map $\cTop_{\calC} \to {\rm End}(\Gamma^{4,20}_{R-R})$ is not injective, we can show that an endomorphism $\nL$ proportional to the identity can only be associated to a superposition of $n$ copies of the identity defect. Given a point $\calC$ in the moduli space of K3 models, we will spell out some necessary conditions for $\cTop_\calC$ to be an integral category, i.e. such that the quantum dimensions of all topological defects are integral (Claim \ref{th:qdim}). In particular, this condition is satisfied for the points in the moduli space  of K3 models that are attractor points for some 1/2 BPS configuration of D-branes \cite{Andrianopoli:1998qg,Dijkgraaf:1998gf,Ferrara:1995ih,Moore:1998pn}. We also derive some weaker constraints on the quantum dimensions that are valid everywhere in the moduli space. 

While we know some examples of K3 models $\calC$ where $\cTop_\calC$ contains a continuum of topological defects, Claim \ref{th:generic} implies that this cannot be the generic situation. It is natural to ask for a characterization of K3 models where such a continuum exists. In section \ref{s:conjecture}, we conjecture that this only happens for (generalised) orbifolds of torus models.

The main limitations in our approach comes from the fact that the map $\cTop_{\calC} \to {\rm End}(\Gamma^{4,20}_{R-R})$ is, in general, neither injective, nor surjective. In particular, whenever $\cTop_{\calC}$ admits a continuum of defect $\CL_\theta$, parametrised by some real parameter(s) $\theta$, all such defects $\CL_\theta$ are necessarily mapped to the same endomorphism $\nL$. As for surjectivity, while we are able to put some constraints on the endomorphisms $\nL \in {\rm End}(\Gamma^{4,20}_{R-R})$ that arise from a defect $\CL$, we cannot determine precisely what the image of the map $\cTop_{\calC} \to {\rm End}(\Gamma^{4,20}_{R-R})$ is.

\medskip

The article is structured as follows. In section \ref{s:generaldefects}, we review some basic facts about topological defects in two dimensional CFTs, and fix the notation that we will use in the rest of the paper. Section \ref{s:genK3models} is the core of the article: after reviewing the main properties of K3 models, we describe and prove the main results of our work in sections \ref{s:topdefsK3} and \ref{s:Dbranes}. We stress that the proofs are on a physics level of rigour, as they are based on various assumptions about K3 models and boundary states that are not mathematically rigorous. This is why we prefer to call such statements `Claim' rather than `Theorem'.  In section \ref{s:torusOrbs} we focus on K3 models that can be described as torus orbifolds. We show that, in general, they admit a continuum of topological defects in $\cTop$. In section \ref{s:conjecture}, we conjecture that the converse might be true: generalised torus orbifolds are actually the \emph{only} K3 models for which such a continuum exists. In sections \ref{s:Z28M20} and \ref{s:onetosix}, we describe some topological defects in $\cTop_\calC$ in a couple of interesting K3 models. While in none of these two models we were able to determine precisely the category $\cTop_\calC$, the examples are useful both to confirm some of the general arguments of section \ref{s:topdefsK3}, and were used in the proof of some of the claims. Finally, in section \ref{s:conclusions} we describe some avenues for future investigation. Various technical details of our calculations are relegated in the appendices.
 
\section{Generalities on topological defects in 2D CFT}\label{s:generaldefects}
In this section we give a simple and concise review about defects in two dimensional CFTs. We refer to \cite{Carqueville:2023jhb, Chang_2019,Frohlich:2009gb} for more detailed information.\\
Usually, when we talk about a generic CFT, we characterize it by specifying the whole set of \textit{local operators} and their corresponding OPEs.  However, it is well known that in many cases further extended objects associated with non-local operators can also exist in the theory. Such objects encode additional properties of the quantum field theory that are not visible at the level of the spectrum, and they are easily understood in the language of the \textit{defects}.   \\
In a generic QFT defined over a $d$-dimensional spacetime $\mathcal{M}_d$, such extended objects can be described through operators $\hat{D}_a \left( \mathcal{M}^{d-q} \right)$  supported on $(d-q)-$dimensional submanifolds of $\mathcal{M}_d$, with $q<d$. These operators are called \textit{topological} in the sense that small deformations of their support manifold, which do not cross other operators of the theory, do not affect the physical observables.  \\
In the limit where the support manifolds of two distinct defects $D_a$ and $D_b$ overlap, the generalized OPE between the corresponding operators defines a \textit{fusion algebra} among defects of the form:
\begin{equation}
D_a \left( \mathcal{M}^{d-q} \right) \times D_b \left( \mathcal{M}^{d-q} \right) = \sum_c N^c_{ab} D_c \left( \mathcal{M}^{d-q} \right).
\label{fusion_alg}
\end{equation}
\noindent
The simplest example of defects we can encounter in a QFT are the \textit{invertible defects}, which encode information about the standard and higher-form \textit{global symmetries} owned by the theory. More specifically, let $G$ be the group of $p$-form global symmetries in our QFT, then we can associate to each element $g \in G$ a $(d-p-1)-$dimensional topological defect $D_g$ such that the induced fusion algebra \eqref{fusion_alg} satisfies the same group multiplication law as $G$:
\begin{equation}
D_g \left( \mathcal{M}^{d-p-1} \right) \times D_{g'} \left( \mathcal{M}^{d-p-1} \right) =  D_{g''} \left( \mathcal{M}^{d-p-1} \right), \quad \quad \quad g''=gg'.
\end{equation}
The name \textit{invertible} for this class of defects comes from the fact that for each of them, i.e. $D_g$, there exists a second defect $D_{g^{-1}}$ such that their fusion produces the trivial defect $D_e$ associated with the identity operator:
\begin{equation}
D_g \times D_{g^{-1}} = D_{g^{-1}} \times D_g = D_{e}, \quad \quad \quad D_e \mapsto \mathbb{I}.
\end{equation}

\noindent
The above definitions apply to any QFT of generic dimension $d$, where topological defects supported on submanifolds of different codimensions may be present. In the rest of the discussion we will focus our attention on 2D QFT, where topological defects have support on oriented $1-$dimensional manifolds (\textit{lines}) of the 2d spacetime. For this reason, we will refer to them as \textit{Topological Defect Lines} (TDLs) and denote them with the notation $\mathcal{L}$.  \\
In particular, if $G$ is the symmetry group of our CFT, we can associate to each element $g \in G$ an invertible TDL $\mathcal{L}_g$.
\begin{figure}[h!]
\centering
\includegraphics[scale=0.25]{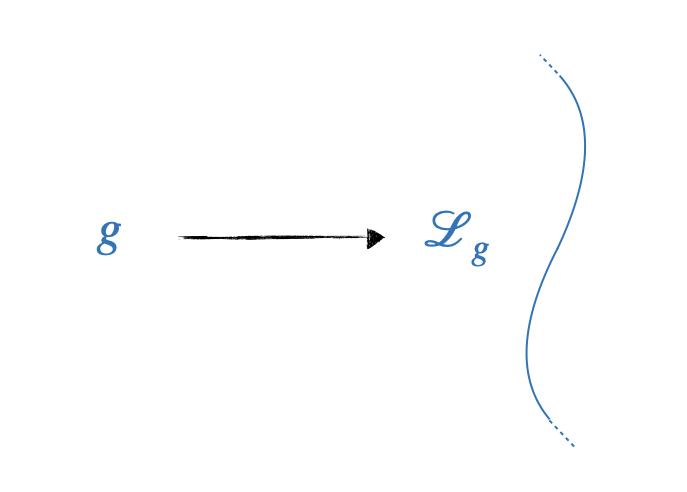}
\caption{\small Invertible TDL in 2d associated to the element $g \in G$.}
\end{figure}

\noindent
 By definition of global symmetry, the elements of the group G define a non-trivial action on the bulk operators:
 \begin{equation}
 g: \quad \mathcal{O}_i (x_i) \quad \mapsto \quad \rho (g) \cdot \mathcal{O}_i (x_i)
 \label{action_g_bulk}
 \end{equation}
 while they leave the correlators invariant:
\begin{equation}
\forall g \in G \quad \longrightarrow \quad \langle \prod_i \mathcal{O}_i (x_i) \rangle = \langle \prod_i \left( \rho(g) \cdot \mathcal{O}_i (x_i) \right) \rangle.
\end{equation}
We can represent the action \eqref{action_g_bulk} in the language of topological defects through the loop contraction of the TDL $\mathcal{L}_g$ encircling the bulk operator $\mathcal{O}_{\psi} (x_i)$ as depicted on the left side of figure \ref{fig2}.  
\begin{figure}[h!]
\centering
\includegraphics[scale=0.2]{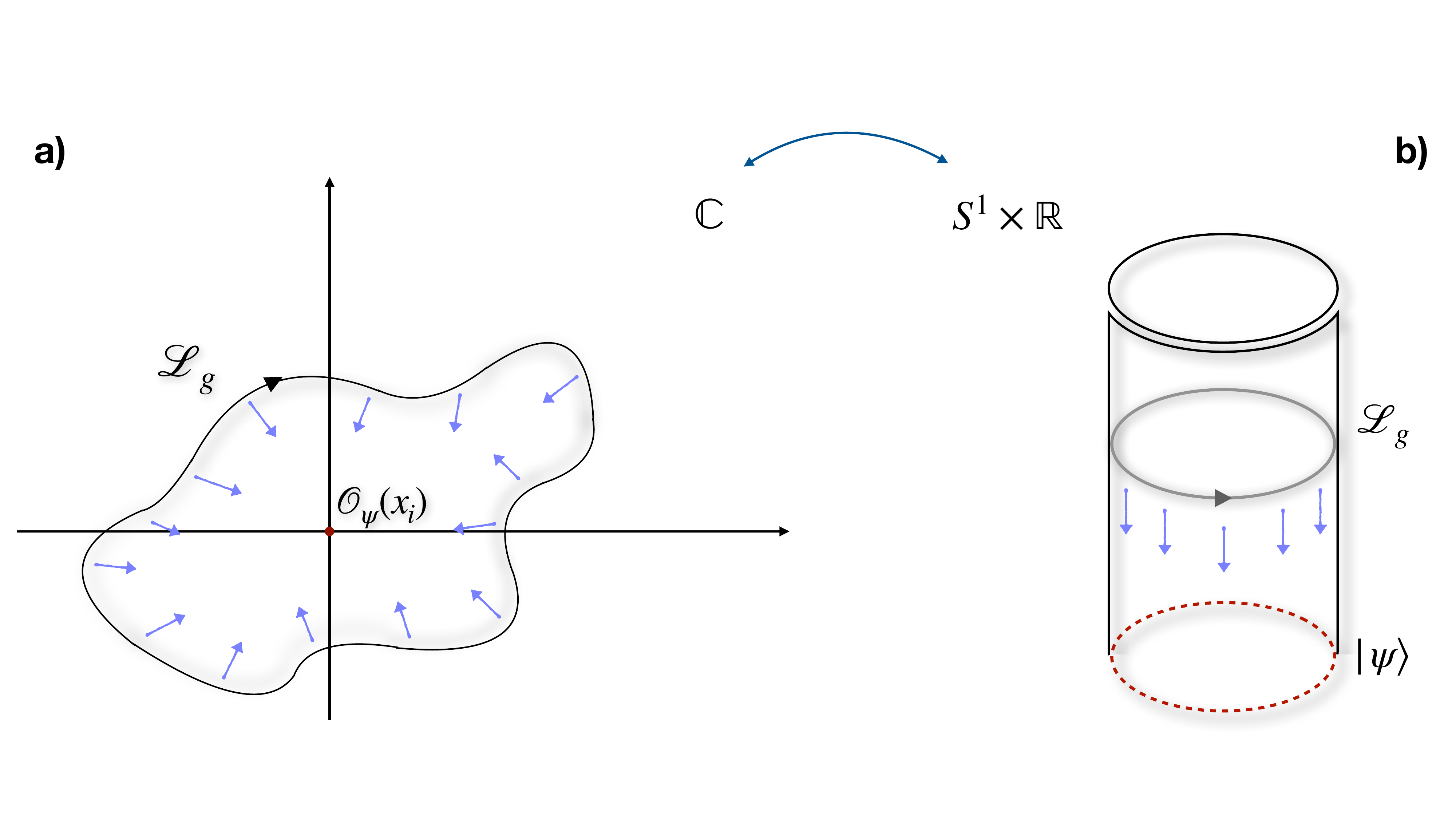}
\caption{ \small \textbf{(a)} Action of $\mathcal{L}_g$ on the bulk operator $\mathcal{O}_{\psi} (x_i)$.\\
\textbf{(b)} Action of $\CL_g$ on the asymptotic state $\vert \psi \rangle$.}
\label{fig2}
\end{figure}
This mechanism is also valid for general TDLs that are not associated with a global symmetry group. From now on we will use the hat notation $\hat{\mathcal{L}}$ to denote the extended operator supported on the line $\mathcal{L}$.\\

\noindent
Invertible defects associated with the elements of a continuous global symmetry group $G$ offer the simplest explicit construction of an extended operator supported on an oriented line. In this case the Noether's theorem provides us a set of conserved currents $J_{(r)}$, such that:
\begin{equation}
    \langle d \ast  J_{(r)} (x) ... \rangle =0,
    \label{conservation_current}
\end{equation}
where the dots denote any operator insertion away from the point $x$. Exponentiating the integral of the Noether's currents on the support line $\CL$: 
\begin{equation}
    \hat{\CL}_g = e^{i \alpha^{(r)} \int_{\CL} \ast J_{(r)}}.
\end{equation}
we get the extended operator $\hat{\CL}_g$ associated with the element $g \in G$ specified by the group parameters $\left\lbrace \alpha^{(r)} \right\rbrace$.
The operator $\hat{\mathcal{L}}_g$ is topological due to the conservation law in \eqref{conservation_current}.\\ 
Similarly, we can define extended operators associated with elements of discrete symmetry groups satisfying the same above proprieties.

 Beyond these objects, we can equip our theory with other point like operators where TDLs can terminate or join. In the first case we can associate to each TDL $\mathcal{L}$ the space $\mathcal{H}_{\mathcal{L}}$ of possible point like operators on which the line $\mathcal{L}$ can end. If $\mathcal{H}_{\mathcal{L}} \neq \emptyset$, then the line $\mathcal{L}$ is said \textit{endable}, and the point like operators of $\mathcal{H}_{\mathcal{L}}$ are called \textit{defect operators}.

\subsection{Defining properties of Topological Defect Lines }

Let us now focus on the case of topological defects $\CL(\gamma)$ supported on lines $\gamma$ in unitary Euclidean two dimensional CFTs.  In the same line of the Introduction, we can think about \textit{topological defects} as a generalization of \textit{global symmetries}. As we spell out below, the set of such defects is equipped with a composition law \eqref{fusion_alg} that is generally non-invertible. This means that one cannot define on this set the standard group structure, as for ordinary symmetries, but rather a fusion category.

In this subsection we will summarize the basic properties that we need in order to specify this structure. \\
As stated above, topological defect lines (TDLs) $\CL$ represent the fundamental objects of the structure. They are supported on oriented lines $\gamma$ on the worldsheet. On the set of distinct defects $\{\CL\}$ is defined an involution $\CL \mapsto \CL^\dual$, shown in figure \ref{fig_3}(a), that corresponds to inverting the orientation of the support. A defect $\CL$ is unoriented if $\CL^\dual =\CL$.

Correlation functions with the insertions of topological defects are invariant under deformation of the support line, as long as the line is not moved past another operator insertion. We will say that a defect $\CL$ and a local operator $\phi$ are `transparent' to each other, or that $\CL$ preserves $\phi$, if the support $\CL$ can be moved past the support of $\phi$ without changing any correlation function. In general, the holomorphic and anti-holomorphic stress tensors $T(z)$ and $\tilde T(\bar z)$ are always preserved by any topological defects. In the following sections, we will consider topological defects that preserve all supercurrents generating the $\CN=(4,4)$ superconformal algebra of a K3 sigma model.

The set of TDLs in a 2d QFT always includes the \textit{Identity} defect $\mathcal{I}$, typically represented by the dotted lines as depicted in figure \ref{fig_3} (b). The insertion or removal of the identity defect does not change any correlation function; equivalently, $\CI$ is transparent to all local operators of the theory.  \\

\begin{figure}[h!]
 \centering
 \includegraphics[scale=0.15]{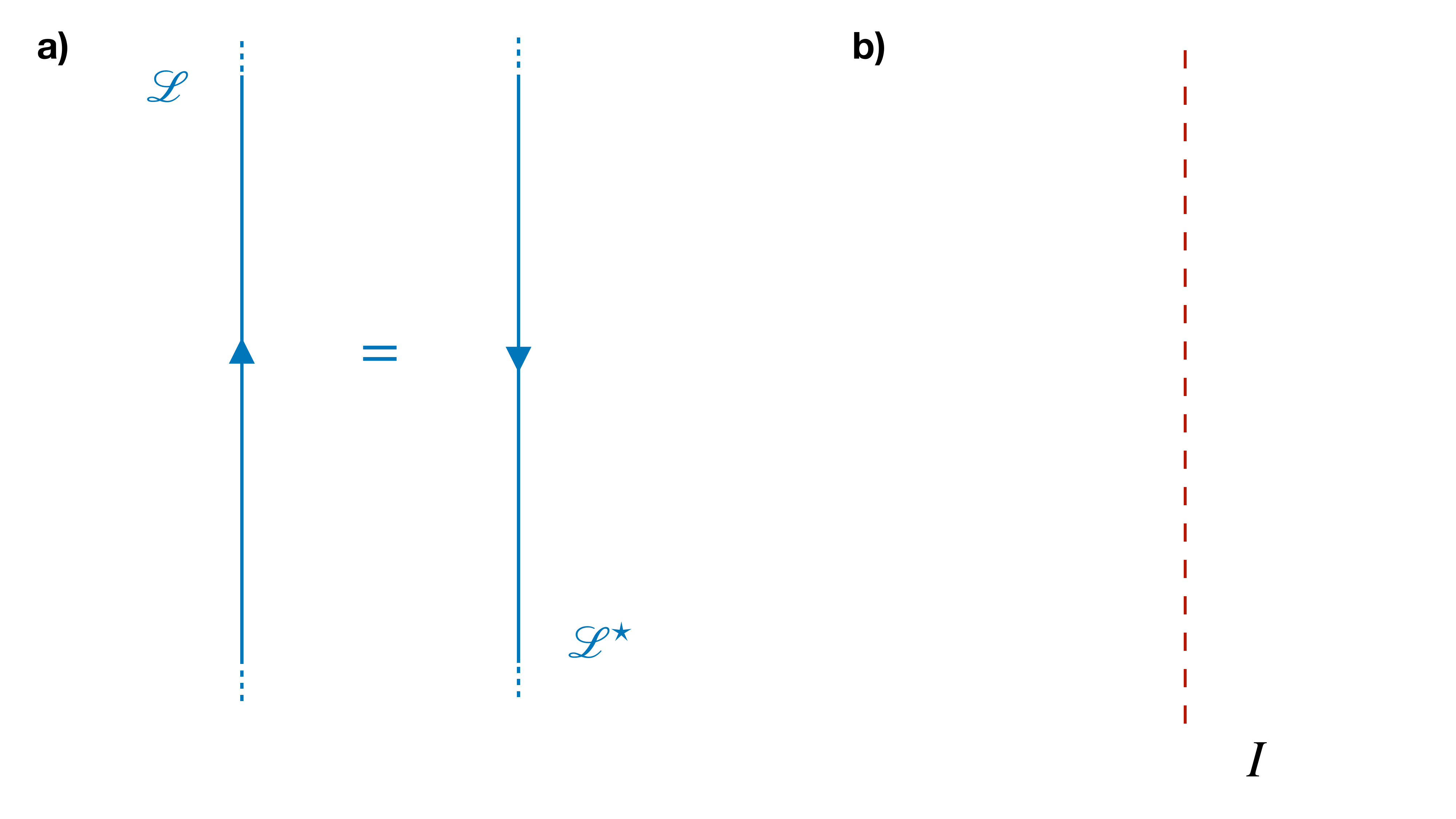}

\caption{ \small \textbf{(a)} Involution map $\mathcal{L} \mapsto \CL^\dual$.
\textbf{(b)} Identity operator in the trivial line.}
\label{fig_3}
\end{figure}

\noindent

Let us consider the two-dimensional CFT on the cylinder $S^1\times \RR$. A TDL $\CL$ inserted along the circle $S^1$ defines a linear operator $\hat \CL:\Hh\to \Hh$ on the Hilbert space $\Hh$ of states on $S^1$ (see figure \ref{fig2}.b). Equivalently, we can consider a closed line $\CL$ encircling the insertion point of a local operator $\phi(z,\bar z)$, corresponding to a state $\phi \in \Hh$. By shrinking the circle around the point $z$ we obtain a new local operator $(\hat \CL\phi)(z,\bar z)$.\footnote{To be precise, the definition of $\hat\CL$ on the sphere might differ from the definition of the cylinder by a phase, see section 2.4 in \cite{Chang_2019}. In this case, we reserve the notation $\hat \CL$ for the operator defined on the cylinder.}

On the other hand, inserting the defect line $\mathcal{L}$ along the Euclidean time direction $\RR$ of $S^1\times \RR$, correspond to modifying the space of states $\Hh$. We denote by $\mathcal{H}_{\mathcal{L}}$ the new Hilbert space of states on the circle $S^1$ that are `twisted' by the defect $\CL$. For example, if $\mathcal{L}$ is an invertible defect associated with a symmetry $g \in G$, then $\mathcal{H}_{\mathcal{L}}$ is simply the \textit{g-twisted sector} of the theory.  \\
The state/operator correspondence defined by the conformal mapping from the cylinder $S^1\times \RR$ to the plane $\CC\setminus \{0\}$, allows us to identify every state $\vert \psi \rangle \in \mathcal{H}_{\mathcal{L}}$ on the cylinder with a non-local defect operator $\mathcal{O}_{\psi}$ on the plane, i.e. an operator  attached to an outgoing defect $\CL$, as shown in figure \ref{fig4}. 
\begin{figure}[h!]
\centering
\includegraphics[scale=0.2]{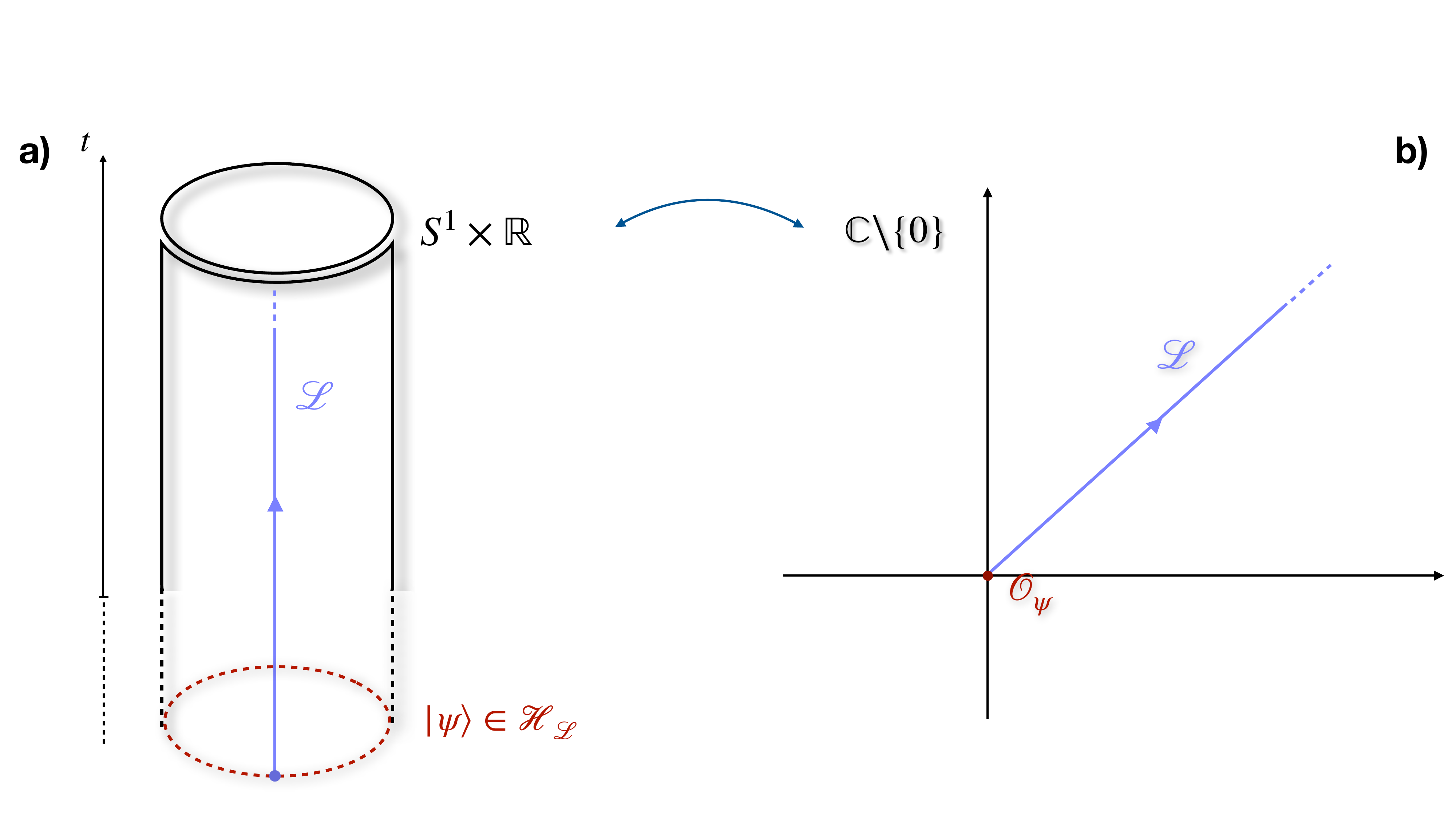}
\caption{ \small An \textit{endable} TDL in the cylinder (figure a) produces a non-local operator in the complex plane (figure b).}
\label{fig4}
\end{figure}
Using this identification, we will refer to $\mathcal{H}_{\mathcal{L}}$ indifferently as the space of $\CL$-twisted states or as the space of $\CL$ defect operators. When $\CL=\CI$ we recover the ordinary space $\Hh_{\CI}=\Hh$ of local point-like operators of the CFT. 

If the defect $\CL$ is transparent for a certain set of local (anti-)holomorphic operators, then the space $\Hh_{\CL}$ is a representation of the corresponding (anti-)chiral algebra. In particular, $\Hh_{\CL}$ is always a representations of the holomorphic and antiholomorphic Virasoro algebra.  More generally, the space $   \Hh_{\CL}$ is a $\CL$-twisted representation of the algebra of local operators on the cylinder. This means that for every $\phi\in \Hh$, $\phi(z,\bar z)$ defines a linear operator on the Hilbert space $\Hh_\CL$ that obeys the OPE relations with other local operators. By $\CL$-twisted, we mean that correlation functions on $S^1\times \RR$ with the insertion of a point operator $\phi(z,\bar z)$ and a defect line $\CL$ along $\RR$ are not quite periodic as $z$ moves around $S^1$, but have some non-trivial discontinuities at the support of the defect that depend on $\CL$. A defect $\CL$ is \emph{simple} if the space $\Hh_{\CL}$ is irreducible as a twisted representation of the algebra of local operators.

When the CFT is defined on a torus $S^1\times S^1$, the modular S-transformation exchanges the insertion of a line $\CL$ along the `space' circle with the insertion along the Euclidean `time' circle. This establishes a relation between the linear operator $\hat\CL$ on $\Hh$ and the twisted space $\Hh_\CL$. 

The $2$-point correlation functions on the sphere with two defect operators $\phi(z,\bar z),\psi(w,\bar w)$ connected by a defect line $\CL$ define a natural non-degenerate bilinear pairing $(\phi,\psi)$ between $\Hh_\CL$ and $\Hh_{\CL^\dual}$. The bilinear pairing is related to the hermitian product on the Hilbert space $\Hh_{\CL}$ by a anti-linear involution $\iota:\Hh_{\CL}\to \Hh_{\CL^\dual}$, such that $\langle \phi_1|\phi_2\rangle=(\iota(\phi_1),\phi_2)$.

\noindent

The set of TDLs is endowed with the algebraic structure of a generally non-commutative (semi-)ring defined by the two operations of \textit{direct sum} (or superposition) ($+$) and \textit{fusion} ($\triangleright$). \\
The \textit{direct sum} allows to associate to each pair of topological defects $\mathcal{L}_a$ and $\mathcal{L}_b$ a third defect $\mathcal{L}_a + \mathcal{L}_b$ such that:
\begin{equation}
    \mathcal{H}_{\mathcal{L}_a + \mathcal{L}_b} = \mathcal{H}_{\mathcal{L}_a} \oplus \mathcal{H}_{\mathcal{L}_b}.
\end{equation}
This first operation is associative and commutative. Unitary CFTs are expected to be \emph{semi-simple}, i.e. every defect $\CL$ can be written as a superposition of \emph{simple} defects
\be \CL=\sum_{\text{simple }\CL_i} n_i\CL_i\ ,
\ee for some non-negative multiplicities $n_i\in \ZZ_{\ge 0}$. Correspondingly, each reducible $\Hh_\CL$ can be decomposed into a direct sum of irreducible components $\Hh_\CL=\oplus_i n_i \Hh_{\CL_i}$.

If there are no additional intermediate insertions between two TDLs $\mathcal{L}_1$ and $\mathcal{L}_2$, we can define the \textit{fusion} deforming one defect into the other, as shown in fig. \ref{fig5}. The resulting object is again a topological defect line denoted as:
\begin{equation}
\mathcal{L}_1 \triangleright  \mathcal{L}_2 = \mathcal{L}_1 \mathcal{L}_2.  
\end{equation}
Fusion defines a notion of tensor product between defect operator spaces
\be \Hh_{\CL_1}\otimes \Hh_{\CL_2}=\Hh_{\mathcal{L}_1 \mathcal{L}_2}
\ee
This second operation is associative but it is in general not commutative. The fusion of two simple defects $\CL_i$ and $\CL_j$  is not necessarily simple, so that one has a decomposition
\be \CL_i\CL_j=\sum_{\text{simple }\CL_k} N_{ij}^k \CL_k\ ,
\ee for some fusion coefficients $N_{ij}^k\in \ZZ_{\ge 0}$. Formally, this means that the set of defects has the structure of a fusion ring, with the simple defects playing the role of a distinguished basis.

Sets of defects that also satisfy commutativity for the fusion form commutative rings.\\
\begin{figure}[h!]
\centering
\includegraphics[scale=0.25]{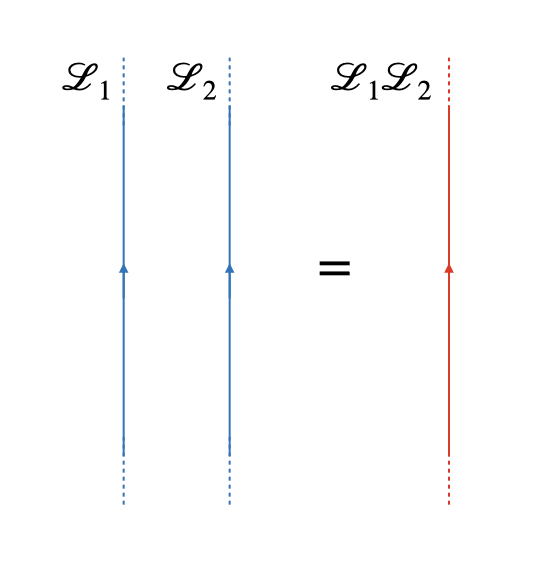}
\caption{Fusion of two TDLs.}
\label{fig5}
\end{figure}

\noindent
The \textit{trivial line} $\mathcal{I}$ is the \textit{neutral element} under fusion
\begin{equation}
    \mathcal{L} \mathcal{I} = \mathcal{I} \mathcal{L} = \mathcal{L}.
\end{equation}

 In analogy with the construction of the spaces $\Hh_\CL$, we can consider $k$ parallel defect lines $\CL_1,\ldots,\CL_k$ inserted in the time-like direction on the cylinder. The corresponding Hilbert space $\CH_{\CL_1,\ldots,\CL_k}$ of states on $S^1$ is identified, via the state/operator correspondence, with the space of  \textit{k-junction} operators. As a convention, we choose that all the involved lines are outgoing from the junction. Once again, $\CH_{\CL_1,\ldots,\CL_k}$ are genuine representations of the chiral and anti-chiral algebras preserved by the defects, and a suitable twisted representation of the algebra of local operators.

By moving the parallel lines $\CL_1,\ldots,\CL_k$ on the cylinder $S^1\times \RR$ very close to each other, we get the identification
\be \Hh_{\CL_1,\ldots,\CL_k}\cong \Hh_{\CL_1\cdots\CL_k}\cong \Hh_{\CL_1}\otimes \cdots\otimes \Hh_{\CL_k}\ ,
\ee between the $k$-junction space $\Hh_{\CL_1,\ldots,\CL_k}$ and the space of defect operators of the fusion $\CL_1\cdots\CL_k$.
   \\
 The subspace $V_{\CL_1,\ldots,\CL_k}\subset \Hh_{\CL_1,\ldots,\CL_k}$ of states with conformal weights $(0,0)$ correspond to junction operators that are themselves topological, i.e. such that the junction point can be moved without changing a correlation function (as long as the insertion point is not moved past the support of some other operator). We can restrict ourselves to consider the correlation functions where all the $k$-junctions with $k>1$ are topological, because all the other junction operators can be obtained by suitable OPE with local operators. 

For a simple defect $\CL$, the only topological $\CL$-twisted operator is the vacuum operator for $\CL=\CI$, so that
\be \dim V_{\CL}=\begin{cases}
    1 & \text{for $\CL=\CI$}\\
    0 & \text{for simple $\CL\neq\CI$}\ .
\end{cases}
\ee
For $k>1$, a topological operator $u\in V_{\CL_1,\ldots,\CL_{k}}$ can also be interpreted as a linear map $u:\Hh_{\CL^\dual_k }\to \Hh_{\CL_1,\ldots,\CL_{k-1}}$ that is a homomorphism of twisted representations of the algebra of local operators. In particular, $V_{\CL,\CL^\dual}={\rm Hom}(\Hh_{\CL},\Hh_{\CL})$ has always dimension at least $1$, because it contains (multiples) of the identity map $\mathbf{1}_\CL:\Hh_\CL\to \Hh_\CL$. 
 A defect $\CL$ is simple if and only if $\dim V_{\CL,\CL^\dual}=1$. Furthermore, for two simple defects $\CL$ and $\CL'$, $$ \dim V_{\CL',\CL^\dual}=\begin{cases}
     1 & \text{if }\CL=\CL'\\
     0 & \text{otherwise}
 \end{cases} .$$

 As for topological $3$-junctions, one can prove that if $\CL_i$, $\CL_j$, and $\CL_k$ are simple defects, then the dimension of topological junction operators is exactly the fusion coefficient
 \be \dim V_{\CL_i,\CL_i,\CL_k^\dual}=N_{ij}^k\ .
 \ee  This fits with the idea that $V_{\CL_i,\CL_i,\CL_k^\dual}$ is the space ${\rm Hom}(\Hh_{\CL_k},\Hh_{\CL_i\CL_j})$ of morphisms from $\Hh_{\CL_k}$  to $\Hh_{\CL_i\CL_j}=\oplus_{l} N_{ij}^l\Hh_{\CL_l}$. Notice that for $\CL$ simple, one can think of the identity $2$-junction $1\in V_{\CL,\CL^\dual}$ as a $3$-junction with the identity defect. As a consequence, $\CL$ is simple if and only if $\dim V_{\CL,\CL^\dual,\CI}=1$, i.e. if and only if $\CI$ appears with multiplicity $1$ in the fusion of $\CL$ with its dual
 \be \CL\CL^\dual=\CI+\ldots\ .
 \ee Here, $\ldots$ denotes a sum with non-negative multiplicities over simple defects distinct from the identity.

At this point it is important to emphasize that the set of topological defect lines equipped by fusion multiplication in general does not form a group. The reason is the absence of an inverse element associated to each TDL. Only a subclass of all possible TDLs admit an inverse under fusion. They are the \textit{invertible TDLs}, and they form a group with respect to this operation. \\
The remaining TDLs are called \textit{non-invertible} and, equipped with the fusion multiplication and the direct sum, they form a more complicated algebraic structure named \textit{Fusion Ring}. \\
A fundamental quantity that we can associate with each TDL $\CL$ is the \textit{quantum dimension} $\langle \CL \rangle\equiv \langle \CL\rangle_{S^1\times \RR} $, defined as the vacuum expectation value of a defect $\CL$ wrapping the circle $S^1$ on the cylinder:
\begin{equation}
    \langle \CL \rangle\equiv \langle \CL\rangle_{S^1\times \RR} := \langle 0 \vert \hat{\CL} \vert 0 \rangle.
\end{equation}
Using the modular invariance properties of the partition function with the defect $\CL$ inserted, it is easy to prove that for unitary theories with a unique vacuum the quantum dimension is bounded from below:
\begin{equation}
    \langle \CL \rangle \geq 0.
\end{equation}
Such constraint is more restrictive when we consider a unitary, compact CFT, where the condition becomes:
\begin{equation}
    \langle \CL \rangle \geq 1.
\end{equation}
Notice that the quantum dimension is also the absolute value of the vacuum expectation value $\langle \CL\rangle_{S^2}$ on the sphere, defined by considering an loop encircling only the vacuum on $S^2$ $$|\langle \CL\rangle_{S^2}|=\langle \CL \rangle\ ,$$ but in general the phase might be different.
The quantum dimensions provide a $1$-dimensional representation of the fusion ring, so that, in particular,
\be \langle \CL_i\rangle\langle \CL_j\rangle=\sum_k N_{ij}^k \langle \CL_k\rangle\ ,
\ee for any simple $\CL_i,\CL_j,\CL_k$. Together with the condition $\langle \CL \rangle \geq 1$, this means that for every simple $\CL_i$, $\CL_j$ of finite quantum dimension, there are only finitely many non-zero fusion coefficients $N_{ij}^k$. In this article, we only consider defects with finite quantum dimension; see \cite{Chang:2020imq} for a discussion about more general possibilities.

\begin{figure}[h!]
    \centering
    \includegraphics[scale=0.17]{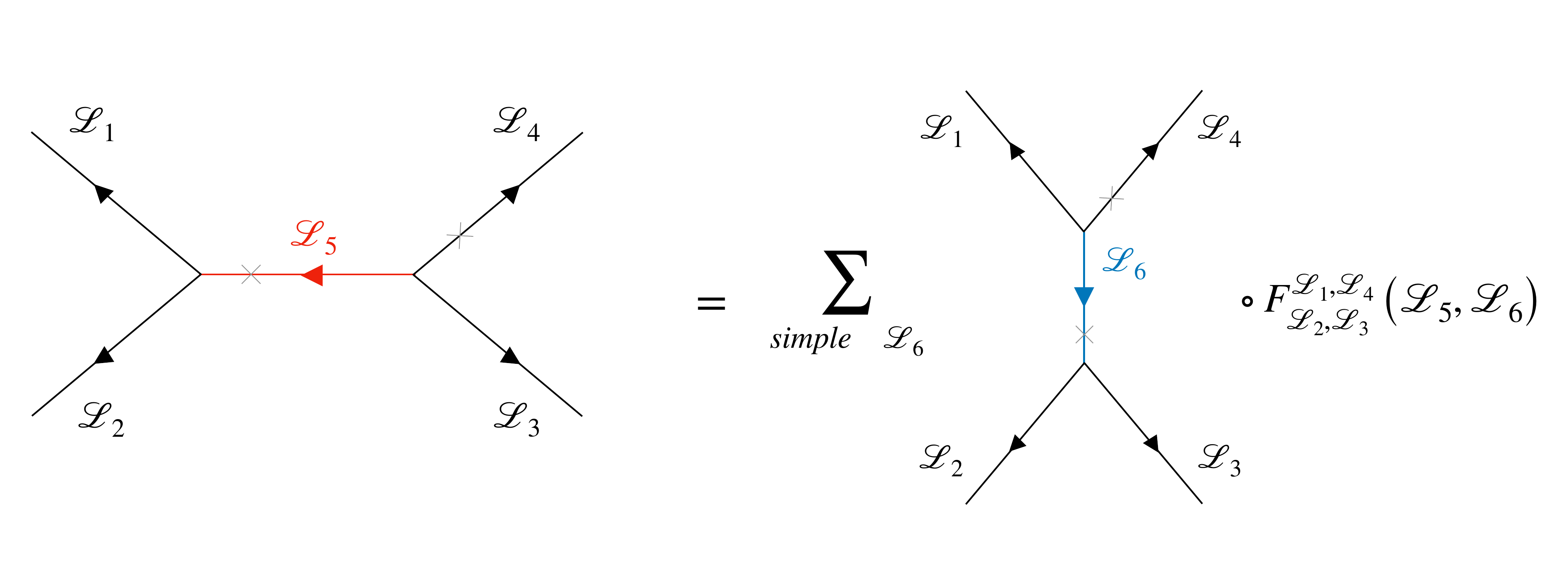}
    \caption{ \small Equivalent configurations under the action of the fusion matrices $F^{\CL_1 , \CL_3}_{\CL_2 , \CL_4}$.}
    \label{fig_6}
\end{figure}

The topology of a network of defects can be modified using the fusion rules shown in fig. \ref{fig_6}, where the fusion matrices $F_{\CL_2,\CL_3}^{\CL_1,\CL_4}(\CL_5,\CL_6)$ map the topological junctions $V_{\CL_1 ,\CL_2, \CL_5^\dual} \otimes V_{\CL_5 ,\CL_3 , \CL_4}$ in $V_{\CL_1 , \CL_6 , \CL_4 } \otimes V_{\CL_2 , \CL_3 , \CL_6^\dual}$. In particular, 
$F^{\CL,\CL^\dual}_{{\CL^\dual}, \CL}(\CI,\CI)$ maps ${\bf 1}\otimes {\bf 1}$ to $\frac{1}{\langle \hat \CL\rangle_{S^2}}{\bf 1}\otimes {\bf 1}$, see fig.\ref{fig_7}.

\begin{figure}[h!]
    \centering
    \includegraphics[scale=0.15]{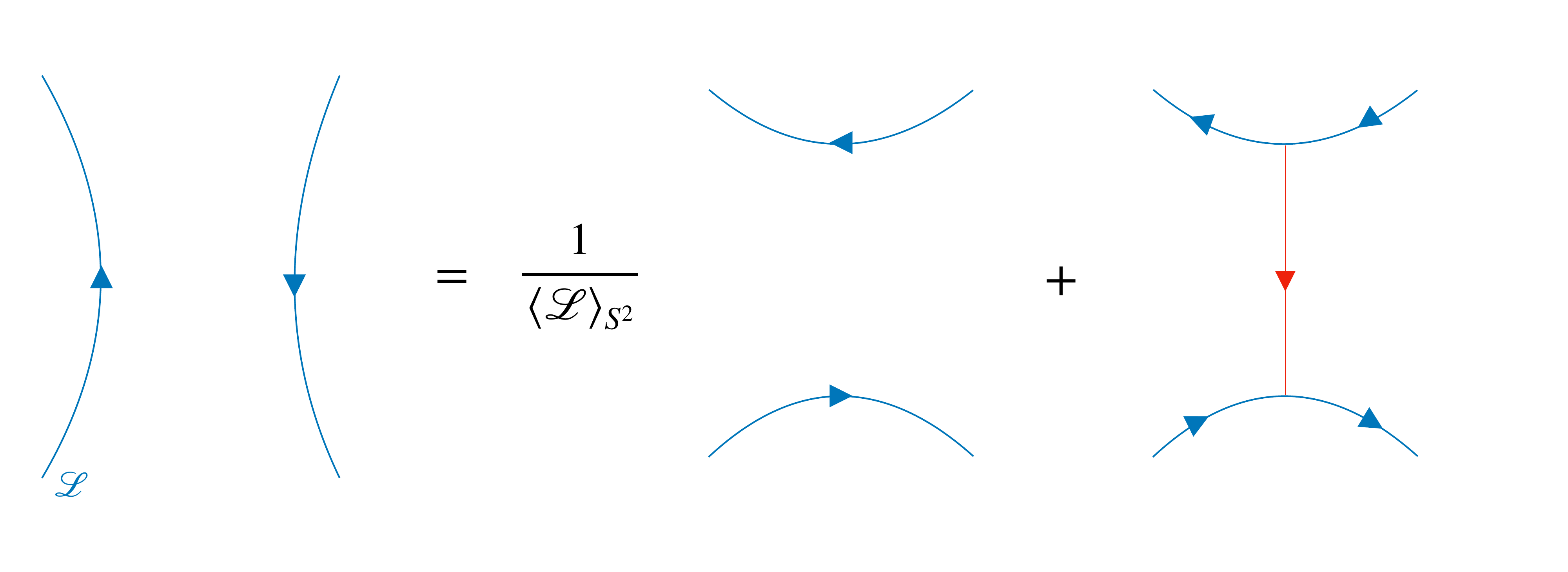}
    \caption{ \small Equivalent configurations under the action of the fusion move. In the second term on the right: the red defect is the superposition (with multiplicity) of the simple defects $\CL_k\neq \CI$ appearing in the fusion $\CL\CL^\dual$; the topological $3$-junction operators are given by the images of ${\bf 1}\otimes {\bf 1}$ by matrices $F^{\CL , \CL^\dual}_{\CL^\dual , \CL}(\CI,\CL_k)$.}
    \label{fig_7}
\end{figure}

This can be used to determine how the correlation function is modified when a simple defect $\CL$ is moved past a local operator $\phi(z,\bar z)$, see figure \ref{fig_8}. In particular, $\CL$ is transparent to $\phi$ if and only if $\hat\CL(\phi)=\langle \CL\rangle \phi$, because in this case one can prove that the term $\hat\CL^v(\phi)$ in figure \ref{fig_8} vanishes. We will use this property repeatedly in the following.

\begin{figure}[h!]
    \centering
    \includegraphics[scale=0.19]{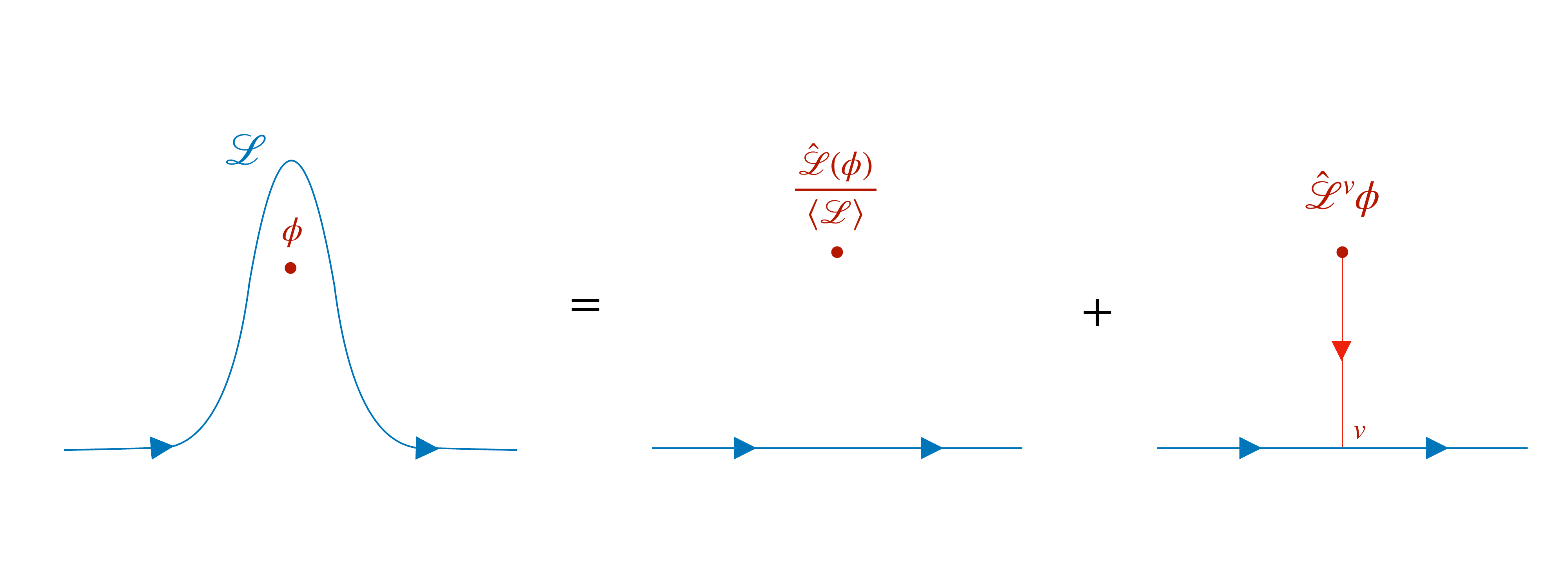}
    \caption{ \small Moving of a simple defect line through the local operator $\phi (z, \overline{z})$. The defect in the second term on the right is a superposition of all simple defect lines $\CL_i\neq \CI$ appearing in the fusion $\CL\CL^\dual=\CI+\ldots$. The factor $1/\langle \CL\rangle$ is determined by observing that the defect must be transparent to the vacuum operator and to its descendants. }
    \label{fig_8}
\end{figure}

\bigskip

Formally, the properties of TDLs described in this section can be formulated in terms of a \emph{fusion category}, where the defects $\CL$ are the objects and the topological junctions $u\in V_{\CL_1,\CL_2^\dual}\equiv {\rm Hom}(\CL_2,\CL_1)$ are the morphisms. Sometimes in the definition of fusion categories one requires that the number of simple objects is finite. This condition might be in general violated when a CFT is not rational with respect to the chiral and anti-chiral algebra preserved by the topological defects, as will be the case in this article.
\section{Topological defects in K3 models}\label{s:genK3models}

In this section, after a short review of non-linear sigma models on K3, we discuss the general properties of topological defects in such models.

\subsection{General properties of K3 models}\label{s:K3models}

Let us review some of the main properties of supersymmetric non-linear sigma models on K3 (or K3 models, for short), and fix our notation for the rest of the article. See \cite{Aspinwall:1996mn,Nahm:1999ps} for more information. 

\bigskip

All K3 models contain a holomorphic and an anti-holomorphic copy of the small $\CN=4$ superconformal algebra at central charge $c=6$ \cite{Eguchi:1987wf,Eguchi:1988af, Eguchi:1987sm}. The bosonic subalgebra of $\CN=4$ is generated by the stress-tensor and the three currents in a $\widehat{su}(2)_1$ current algebra. The $SU(2)$ group generated by the zero modes of such holomorphic currents is the R-symmetry group, and the four supercurrents belong to two R-symmetry doublets. 
The unitary representations of the $\CN=4$ algebra at $c=6$ can be labeled by a pair $(h,q)$, where $h$ is the conformal weight and $q\in\{0,1/2\}$ the $SU(2)$ highest weight of the ground states in the representation; similarly, we will denote by $(h,q;\tilde h,\tilde q)$ the representations of the full (holomorphic and antiholomorphic) $\CN=(4,4)$ algebra at $(c,\tilde c)=(6,6)$. The Neveu-Schwarz (NS) version of the $\CN=4$ algebra contains two unitary short (BPS) representations $(h,q)=(0,0)$ and $(h,q)=(\frac{1}{2},\frac{1}{2})$ and infinitely many long (massive) representations $(h,0)$ with $h>0$. The spectral flow automorphism relating the NS and the Ramond (R) versions of $\CN=4$ maps the NS shorts representations $(0,0)$ and $(\frac{1}{2},\frac{1}{2})$ to the short R representations $(\frac{1}{4},\frac{1}{2})$ and $(\frac{1}{4},0)$, respectively, and the long $(h,0)$ with $h>0$ to the long R representation $(h+\frac{1}{4},\frac{1}{2})$. 

\bigskip

 Every K3 model contains a single $(0,0;0,0)$ representation of $\CN=(4,4)$ (the vacuum and its descendants), and $20$ $(\frac{1}{2},\frac{1}{2};\frac{1}{2},\frac{1}{2})$ representations. There are no fields in the $(\frac{1}{2},\frac{1}{2};0,0)$ or $(0,0;\frac{1}{2},\frac{1}{2})$; such holomorphic and anti-holomorphic free fermions are a characteristic feature of  supersymmetric sigma models on $T^4$, where they are part of a `large' $\CN=(4,4)$ superconformal algebra at $(c,\tilde c)=(6,6)$. Besides these short-short (BPS) representations, a K3 model contains infinitely many different short-long, long-short and long-long representations, with multiplicities depending on the particular model. It is believed that the $\CN=4$ superconformal algebra is the full chiral algebra of a generic K3 model, that is therefore not a rational CFT. 
 At special points in the moduli space of K3 models, the chiral algebra might be extended so that the CFT is rational.

K3 models are invariant under spectral flow exchanging the NS-NS sector and R-R sector. This implies that the R-R sector contains a single $(\frac{1}{4},\frac{1}{2};\frac{1}{4},\frac{1}{2})$  and $20$ $(\frac{1}{4},0;\frac{1}{4},0)$ representations, as well as infinitely many $(h,\frac{1}{2};\tilde{h},\frac{1}{2})$, with $h+\tilde h>1/2$. The four R-R ground fields of $(\frac{1}{4},\frac{1}{2};\frac{1}{4},\frac{1}{2})$ transform in a $(\bf{2},\bf{2})$ representation of the holomorphic and anti-holomorphic $SU(2)\times SU(2)$ R-symmetry group; the OPE with such fields generate the spectral flow between the NS-NS and the R-R sectors. For this reason, we will refer to such fields as the \emph{spectral flow generators}.

It is often useful to think of the K3 model as the internal CFT in a six dimensional compactification of type IIA superstring on $\RR^{1,5}\times K3$. In this case, one should also consider the NS-R and R-NS sectors of the K3 sigma model, tensored with the corresponding sectors of the space-time and superghost CFTs, and with the correct GSO projections. The NS-R and R-NS sectors are also related to the NS-NS and R-R sectors by spectral flow, so that they contain a single $(\frac{1}{4},\frac{1}{2};0,0)$ and $(0,0;\frac{1}{4},\frac{1}{2})$ representations. The corresponding physical string states are the space-time gravitinos, whose zero modes is associated with the space-time $\CN=(1,1)$ supersymmetry in six dimensions.

In general, sigma models on K3 and on $T^4$ are the only known (and, conjecturally, the only consistent) unitary SCFTs  with $\CN=(4,4)$ superconformal algebra at $(c,\tilde c)=(6,6)$ and whose spectrum is invariant under spectral flow.

\bigskip

Every K3 model admits deformations that preserve the full $\CN=(4,4)$ superconformal symmetry, corresponding to an $80$-dimensional space of exactly marginal operators contained in the $20$ NS-NS $(\frac{1}{2},\frac{1}{2};\frac{1}{2},\frac{1}{2})$  representations of $\CN = (4,4)$. There are compelling arguments suggesting that conformal perturbation by any such operator converges in a neighborhood of each model; we will assume that this is true. 
The corresponding $80$-dimensional moduli space $\CM_{K3}$ of K3 models is given by a quotient
\be \CM_{K3}=O(4,20,\ZZ)\backslash \CT_{K3}
\ee where $O(4,20,\ZZ)$ is the integral orthogonal group. Here, the Teichm\"uller space $\CT_{K3}$ is an open subset in the  Grassmannian
\be\label{grassmannian} \CT_{K3}\subset O(4,20,\RR)/(O(4)\times O(20))
\ee parametrising positive definite four-dimensional subspaces $\Pi$ within the real space $\RR^{4,20}$ with signature $(4,20)$. The moduli space $\CM_{K3}$ admits the following physical interpretation: $O(4,20,\ZZ)$ is the T-duality group, and can be identified with the group of automorphisms of the lattice $\Gamma^{4,20}$ of D-brane charges, which is an even unimodular lattice of signature $(4,20)$. One can think of this lattice (or rather its dual) as being embedded in the $24$-dimensional real space $V$ of (CPT self-conjugate) Ramond-Ramond ground states with conformal weights $h=\tilde h=\frac{1}{4}$:
\be \Gamma^{4,20}\subset V:=\{\text{R-R ground states}\}\cong \Gamma^{4,20}\otimes \RR\cong \RR^{4,20}\ .
\ee The space $V$ contains a positive definite subspace \be V\supset \Pi:=\{\text{spectral flow generators}\}\cong \RR^{4,0}\ , \ee spanned by the four spectral flow generators, i.e. the R-R ground states in a $(\frac{1}{4},\frac{1}{2};\frac{1}{4},\frac{1}{2})$ representation of $\CN=(4,4)$.  The orthogonal complement $\Pi^\perp\subset V$ is the space of R-R ground states in the $20$ $\CN=(4,4)$ representations $(\frac{1}{4},0;\frac{1}{4},0)$ \be V\supset \Pi^\perp:=\{\text{states in $(\frac{1}{4},0;\frac{1}{4},0)$ representations}\}\cong \RR^{0,20}\ .\ee Then, $\CM_{K3}$ is essentially the Grassmannian of the four-dimensional subspaces $\Pi\cong\RR^{4,0}$ within $V\cong \Gamma^{4,20}\otimes \RR$, modulo lattice automorphisms (T-dualities) $O(4,20,\ZZ)\cong O(\Gamma^{4,20})$. To be precise, one needs to exclude some points in this Grassmannian, where the CFT is believed to be inconsistent \cite{Aspinwall:1996mn}. From a string theoretical point of view, these are the points in the moduli space of type IIA superstrings where some D-brane becomes exactly massless, so that the perturbative description breaks down even at small string coupling.

Henceforth, we will denote by
\be \calC_\Pi:=\text{K3 model corresponding to }\Pi\subset V\cong \Gamma^{4,20}\otimes \RR
\ee
the K3 model corresponding to a choice of $\Pi\subset V$, i.e. to a point in the Grassmannian $\CT_{K3}$.

\subsection{Symmetries and topological defects}\label{s:topdefsK3}

The main focus of this article are the topological defects in non-linear sigma models of K3. A generic topological defect commutes only with the chiral and anti-chiral Virasoro algebra, i.e. it is `transparent' only to the holomorphic and antiholomorphic stress-energy tensor $T(z)$ and $\tilde T(\bar z)$. In this article, we will focus on a subcategory of topological defects $\CL$ that satisfy some further constraints, namely:
\begin{enumerate}
    \item They commute with the full $\CN=(4,4)$ superconformal algebra, i.e. they are transparent to all supercurrents and to the $\widehat{su}(2)_1$ R-symmetry currents, besides the stress-energy tensor;
    \item They commute with spectral flow generators. This implies that they are transparent to the four R-R ground fields in the $(\frac{1}{4},\frac{1}{2};\frac{1}{4},\frac{1}{2})$ representation of $\CN=(4,4)$. When the K3 model is the internal SCFT in a full type IIA compactification, we also require the defect to be transparent with respect to the NS-R and R-NS ground fields corresponding to the space-time gravitini. These fields generate the purely holomorphic or purely anti-holomorphic spectral flows.
\end{enumerate}
For a K3 model $\calC\equiv \calC_\Pi$, corresponding to a choice of a four dimensional positive definite subspace $\Pi\subset \Gamma^{4,20}\otimes \RR$, we denote by
\be \cTop_{\calC}\equiv \cTop_\Pi
\ee the category of topological defects of $\calC_\Pi$ satisfying the properties 1 and 2.

\bigskip

Properties 1 and 2 lead to a number of important consequences:

\begin{itemize}
\item As explained in section \ref{s:generaldefects}, each topological defect $\CL$ is associated with a linear operator
\be \hat\CL:\Hh\to\Hh
\ee on the Hilbert space of states on the circle $S^1$ (or, equivalently, the space of local point-like operators) of the K3 model $\calC$. The action of $\hat\CL$ is defined  by inserting a defect $\CL$ along the circle $S^1$ on a cylinder  $S^1\times \RR$. Because of the properties 1 and 2, this operator commutes with the $\CN=(4,4)$ algebra and the spectral flow.  Therefore, it maps $\CN=4$ primaries into $\CN=4$ primaries in the same representation. Furthermore, because the defect $\CL$ is transparent with respect to the spectral flow generators, once the action of $\hat\CL$  is known on one of the sectors ($\Hh_{NS-NS}$, $\Hh_{R-R}$, $\Hh_{R-NS}$ or $\Hh_{NS-R}$), then it is uniquely determined in all the other sectors as well. 
\item In general, the properties of topological defects in fermionic CFTs are more complicated than the ones in purely bosonic ones. In particular, fermionic CFTs can contain topological defects $\CL$ of `q-type', that admit topological $2$-junctions $\psi\in V_{\CL,\CL^\dual}$ with odd fermion number $(-1)^{F_L+F_R}=-1$, see \cite{Chang:2022hud,Runkel:2022fzi}. However, none of the topological defects in the category $\cTop$ is of q-type. Indeed, in the K3 models we consider, the holomorphic and anti-holomorphic fermion numbers $(-1)^{F_L}$ and $(-1)^{F_R}$ can be identified with the central $\ZZ_2$ elements in the holomorphic and anti-holomorphic $SU(2)$ R-symmetry group. The topological defects $\CL\in \cTop$ are transparent to the $\widehat{su}(2)_1$ currents in the $\CN=(4,4)$ algebra, and, as a consequence, they commute with the $SU(2)$ R-symmmetry groups generated by their zero modes, and therefore with the fermion numbers $(-1)^{F_L}$ and $(-1)^{F_R}$. Furthermore, all spaces of defect and junction operators are representations of the $\CN=(4,4)$ superconformal algebra, and the fermion number of each state is determined by their $SU(2)$ R-symmetry charges. In particular, topological operators with $h=0=\tilde h$ have zero R-symmetry charges, and therefore they are always bosons.\\
In fact, each $\CL\in \cTop_\calC$ induces a topological defect in the purely bosonic CFT $\calC^{bos}$ obtained by a type $0$ GSO projection, i.e. by including only the NS-NS and R-R fields of the K3 model $\calC$ with positive fermion number. Furthermore, all properties of the topological defect $\CL$ in the original supersymmetric model $\calC$ are completely determined in terms of the action of the operator $\hat \CL$ in the bosonic model $\calC^{bos}$.  
Note, however, that the definition of the category $\cTop_\calC$ is much more natural in the supersymmetric setup. In particular, we do not know whether the condition of preserving the $\CN=(4,4)$ superconformal algebra admits an equivalent formulation in the bosonic model $\calC^{bos}$.
\item With each defect $\CL$ is associated a $\CL$-twisted space of states $\Hh_{\CL}$, with different sectors $\Hh_{\CL}^{NS-NS}$, $\Hh_{\CL}^{R-R}$, $\Hh_{\CL}^{R-NS}$, $\Hh_{\CL}^{NS-R}$. When $\CL\in \cTop_\calC$, each of these sectors decomposes into representations of the $\CN=(4,4)$ superconformal algebra. Furthermore, the sectors are related to each other by spectral flow.
\end{itemize}

Topological defects that are invertible form  the group of symmetries of the K3 model. In \cite{Gaberdiel:2011fg}, all group of symmetries commuting with the $\CN=(4,4)$ algebra and the spectral flow generators have been classified. In particular, consider a K3 model $\calC_\Pi$ in the moduli space $\CM_{K3}$, corresponding to the choice of the positive definite four-dimensional subspace $\Pi$ of spectral flow generators in the space of RR ground fields $V\cong \Gamma^{4,20}\otimes  \RR$. Then,  the group $G_\Pi$ of symmetries of $\calC_\Pi$ satisfying 1 and 2 is isomorphic to the subgroup ${\rm Stab}(\Pi)$ of $O(\Gamma^{4,20})\cong O(4,20,\ZZ)$ fixing $\Pi$ pointwise \cite{Gaberdiel:2011fg}. 

Let us revisit the argument that led to this result, and then discuss to what extent such argument can be generalized to the case of topological defects.
Every symmetry $g\in G_\Pi$ maps 1/2 BPS boundary states to 1/2 BPS boundary states, and therefore maps the lattice $\Gamma^{4,20}$ of RR charge vectors into itself. The map must be linear and preserve the bilinear form of the lattice -- indeed, the bilinear form is a Witten index counting the Ramond ground states for open strings suspended between two D-branes, and is invariant under the action of $G_\Pi$. Therefore, every $g\in G_\Pi$ induces an automorphism of the lattice $\Gamma^{4,20}$. Furthermore, by property 2, the induced action on RR ground states must act trivially on the spectral flow generators in $\Pi$. We conclude that there is a homomorphism $\rho:G_\Pi \to {\rm Stab}(\Pi)\subset O(\Gamma^{4,20})$. Then one proves that such homomorphism is both injective and surjective. To show surjectivity, one notices that $O(\Gamma^{4,20})\cong O(4,20,\ZZ)$ is the T-duality group, and ${\rm Stab}(\Pi)$ is a subgroup of dualities mapping the model $\calC_\Pi$ into itself, i.e. self-dualities. But all self-dualities are symmetries of the model, so they must correspond to some $g\in G_\Pi$. As for injectivity, let $K_\Pi\subseteq G_\Pi$ be the kernel of $\rho$. Then $K_\Pi$  acts trivially on $\Gamma^{4,20}$, and, by linearity, on all RR ground states. But the RR ground states in the twenty $(\frac{1}{4},0;\frac{1}{4},0)$ representations of $\CN=4$ are related by spectral flow to the $80$-dimensional space of exactly marginal operators in the NS-NS sector. Thus, the symmetries in $K_\Pi$ act trivially on all such exactly marginal operators, and as a consequence they are not broken by any deformation of the model. Given that the moduli space $\CM_{K3}$ is connected, we conclude that the kernel of $\rho$ is the same group $K_\Pi\equiv K$ for all K3 models. At this point, one just needs to consider a simple example of non-linear sigma model on K3, where the group $K$ can be explicitly described -- for example the model considered in section \ref{s:Z28M20}. It turns out that $K$ is trivial in that model, and therefore is trivial everywhere in the moduli space $\CM_{K3}$. We conclude that $\rho:G_\Pi \to {\rm Stab}(\Pi)$ is an isomorphism. In \cite{Gaberdiel:2011fg}, it was then proved that every group of the form ${\rm Stab}(\Pi)$ is isomorphic to a subgroup of the Conway group $Co_0$, the group of automorphisms of the Leech lattice $\Lambda$, fixing a sublattice of $\Lambda$ of rank at least $4$. All subgroups of $Co_0$ that are lattice stabilizers were classified in \cite{HohnMason2016}.

\bigskip

Let us now discuss how a similar argument could be generalized to a classification of topological defects $\CL\in \cTop$. As described in section \ref{s:Dbranes}, the fusion of a boundary state $||\alpha\rrangle$ and a defect $\CL$ yields a new boundary state $||\CL\alpha\rrangle$. In particular, the defects $\CL\in\cTop$  preserve the space-time supersymmetry, so they map 1/2 BPS D-branes into 1/2 BPS D-branes, and RR charge vectors to RR charge vectors. Therefore, we have a map
\begin{align}\label{homtoend}
    \cTop_\Pi &\to {\rm End}(\Gamma^{4,20})\\
    \CL&\mapsto \nL\notag
\end{align}
that assigns a $\ZZ$-linear function $\nL:\Gamma^{4,20}\to\Gamma^{4,20}$ to each defect $\CL\in\cTop_\Pi$. This map is compatible with fusion product, i.e. it gives rise to a ring homomorphism from the fusion ring of $\cTop_\Pi$ to ${\rm End}(\Gamma^{4,20})$. The extension of $\nL$  by linearity to the real space of R-R ground fields $V\cong \Gamma^{4,20}\otimes \RR$ coincides with the restriction $\hat \CL_{|V}\in {\rm End}_{\RR}(V)$ of the linear operator $\hat \CL$ to $V$, \be \hat \CL_{|V}:V\to V\ .\ee To summarize:
\begin{itemize}
    \item[(a)] The restriction $\hat\CL_{|V}:V\to V$ maps $\Gamma^{4,20}\subset V$ into $\Gamma^{4,20}$, i.e. it is the extension by $\RR$-linearity of some lattice endomorphism $\nL:\Gamma^{4,20}\to \Gamma^{4,20}$. 
    \end{itemize}
Henceforth, we use the symbol $\nL$ to denote both maps $\nL:\Gamma^{4,20}\to\Gamma^{4,20}$ and $\hat \CL_{|V}:V\to V$. Because $\hat\CL$ commutes with the $\CN=(4,4)$ algebra, it cannot mix RR ground fields in different representations. Furthermore, the condition that the spectral flow generators are transparent with respect to $\CL$ implies that the map $\hat\CL$ acts on them in the same way as on the vacuum, i.e. by multiplication by the quantum dimension $\langle \CL\rangle\ge 1$. The following property then follows:
\begin{itemize}
    \item[(b)] $\nL:V\to V$ is block-diagonal with respect to the orthogonal decomposition $V=\Pi\oplus \Pi^\perp$, i.e. $\nL(\Pi)\subseteq \Pi$ and $\nL(\Pi^\perp)\subseteq\Pi^\perp$. Furthermore, the restriction $\nL_{\rvert \Pi}$ is proportional to the identity $\nL_{\rvert\Pi}=\langle \CL\rangle{\rm id}_\Pi$, where $\langle \CL\rangle\ge 1$ is the quantum dimension of the defect.
\end{itemize}
It is useful to introduce the real vector space \be \label{blkspace}
B^{4,20}(\RR):=\left\{\left(\begin{smallmatrix}
    d\cdot \mathbf{1}_{4\times 4} & 0\\
    0 & b_{20\times 20}
\end{smallmatrix}\right)\mid d\in\RR, b_{20,20}\in Mat_{20\times 20}(\RR)\right\}\ee of block diagonal real $24\times 24$ matrices, with a $4\times 4$ upper left-corner proportional to the identity and an unconstrained $20\times 20$ lower-right block, and its subset 
\be\label{blkspaceplus}
B^{4,20}_+(\RR):=\left\{\left(\begin{smallmatrix}
    d\cdot \mathbf{1}_{4\times 4} & 0\\
    0 & b_{20\times 20}
\end{smallmatrix}\right)\mid d\ge 1, b_{20,20}\in Mat_{20\times 20}(\RR)\right\}\subset B^{4,20}(\RR)\ .\ee
Upon choosing a suitable orthonormal basis of $V$, compatible with the splitting $V\cong \Pi\oplus\Pi^\perp$, the space of linear maps $V\to V$ that satisfy property (b) can be identified with $B^{4,20}_+(\RR)$. Furthermore, if we define
\be B^{4,20}_\Pi(\ZZ):={\rm End}(\Gamma^{4,20})\cap B^{4,20}(\RR)\ ,\qquad \qquad B^{4,20}_{\Pi,+}(\ZZ):={\rm End}(\Gamma^{4,20})\cap B^{4,20}_+(\RR)\ ,
\ee then the maps $V\to V$ satisfying both properties (a) and (b) can be identified with $B^{4,20}_{\Pi,+}(\ZZ)$. Then the image of the map \eqref{homtoend} is actually contained in $B^{4,20}_{\Pi,+}(\ZZ)$, so that we can restrict the target and consider the map
\begin{align}\label{homtoB}
    \cTop_\Pi &\to B^{4,20}_{\Pi,+}(\ZZ)\subset {\rm End}(\Gamma^{4,20})\\
    \CL&\mapsto \nL\notag
\end{align} which gives rise to a homomorphism of semirings. As the notation suggests, the intersections $B^{4,20}_{\Pi}(\ZZ)$ and $B^{4,20}_{\Pi,+}(\ZZ)$ depend on the way the four-dimensional space $\Pi$ is embedded in $\Gamma^{4,20}\otimes \RR$, i.e. on the point on the moduli space $\CM_{K3}$.

\medskip

It is plausible that the maps $\nL$ satisfy some further constraints associated with unitarity. 
Let $\CL$ be a \emph{simple} defect, so that the Hilbert space $\Hh_{\CL\CL^\dual}$ admits an orthogonal decomposition as $$\Hh_{\CL\CL^\dual} \cong \Hh\oplus \Hh_{\CL\CL^\dual-\CI}\ ,$$ where $\Hh_{\CL\CL^\dual-\CI}$ is a sum (with suitable multiplicities) of simple defect spaces $\Hh_{\CL_i}$ with $\CL_i\neq \CI$. Suppose that, in a correlation function, we move the support of the defect line $\CL$ past the insertion point $z$ of a point-like operator $\phi(z)$, with $\phi\in \CH$. Then, $\phi(z)$ gets replaced with the sum of an operator $\frac{1}{\langle \CL\rangle}\hat\CL(\phi)\in \Hh$ plus (possibly vanishing) contributions from each of the components of $\Hh_{\CL\CL^\dual-\CI}$ (see figure \ref{fig_8} in section \ref{s:generaldefects}). This move defines a linear map $\Hh\to \Hh_{\CL\CL^\dual}$. In a unitary CFT, it is natural to expect such a map to be an isometry; we call this assumption a \emph{strong unitarity hypothesis}. Because for a simple defect, the contribution  $\frac{1}{\langle \CL\rangle}\hat\CL(\phi)\in \Hh$ is orthogonal to the contributions from $\Hh_{\CL\CL^\dual-\CI}$, we get as a consequence
\be\label{weakunitarity}
\frac{\|\hat\CL(\phi)\|^2}{|\langle \CL\rangle|^2}\le \|\phi\|^2\ ,\qquad \forall \phi\in \Hh\ ,
\ee where the equality holds if and only if all the other contributions vanish. We call the condition \eqref{weakunitarity} the \emph{weak unitarity hypothesis}. Notice that if \eqref{weakunitarity} holds for all \emph{simple} defects $\CL$, then it must hold for all superpositions as well. While we do not know any counterexample to these hypotheses\footnote{There are well-known counterexamples in non-unitary theories though. For example, the Lee-Yang model with central charge $-\frac{22}{5}$ admits a simple defect $\CL$ of quantum dimension $\langle \CL\rangle=\frac{\sqrt{5}-1}{2}$ where one of the eigenvalues of $\hat\CL$ is $\frac{\sqrt{5}+1}{2}>\langle \CL\rangle$ \cite{Chang_2019}.}, we are not aware of any general proof either. A proof of \eqref{weakunitarity} was given in \cite{Chang:2020imq} (see proposition 8), under certain conditions on $\CL$ and $\phi$. In particular, eq.\eqref{weakunitarity} holds for all Verlinde lines in unitary rational CFT. Unfortunately, K3 models are not rational with respect to the $\CN=(4,4)$ algebra and we were not able to prove that such conditions are satisfied for all $\CL\in \cTop$. If \eqref{weakunitarity} holds, an immediate consequence is the following:
\begin{itemize}
    \item [(c)] (Assuming \eqref{weakunitarity} holds.) The operatorial norm $\|\nL\|:=\sup_{0\neq v\in V}\frac{\|\nL v\|}{\|v\|}$ equals $\langle \CL\rangle$. Here, $\|v\|$ denotes the \emph{Euclidean} norm on $V$.
\end{itemize}

Let us define the `bounded' sets
\be\label{blkspaceplusbound}
B^{4,20}_{+,b}(\RR)=\left\{\left(\begin{smallmatrix}
    d\cdot \mathbf{1}_{4\times 4} & 0\\
    0 & b_{20\times 20}
\end{smallmatrix}\right)\mid d\ge 1, b_{20,20}\in Mat_{20\times 20}(\RR),\ \|b_{20,20}\|\le d\right\}\subset B^{4,20}_+(\RR)\ ,\ee and
\be B^{4,20}_{\Pi,+,b}(\ZZ)={\rm End}(\Gamma^{4,20})\cap B^{4,20}_{+,b}(\RR)\ .\ee If property (c) holds we can further restrict the target of the map \eqref{homtoB} to 
\begin{align}\label{homtoBbound}
    \cTop_\Pi &\to B^{4,20}_{\Pi,+,b}(\ZZ)\subset {\rm End}(\Gamma^{4,20})\\
    \CL&\mapsto \nL\ .\notag
\end{align} The map \eqref{homtoBbound} still gives rise to a homomorphism of semirings. Notice that, for any given real number $d>0$, there are finitely many maps $\nL\in B^{4,20}_{\Pi,+,b}(\ZZ)$ with quantum dimension $\langle \CL\rangle\le d$. 
We will not use property (c) to prove any of the claims in the rest of the paper.

\medskip

The set $B^{4,20}_{\Pi,+}(\ZZ)$ contains the group ${\rm Stab}(\Pi)\subset O(\Gamma^{4,20})$ of lattice automorphisms fixing $\Pi$. In fact,  ${\rm Stab}(\Pi)$ can be characterized as the subset of invertible elements in $B^{4,20}_{\Pi,+}(\ZZ)$, this follows immediately by noticing that if $\nL\in {\rm End}(\Gamma^{4,20})$ admits a multiplicative inverse $\nL^{-1}\in {\rm End}(\Gamma^{4,20})$, then both $\nL$ and $\nL^{-1}$ are in $O(\Gamma^{4,20})$.

Therefore, the semiring homomorphism \eqref{homtoB} (or \eqref{homtoBbound}, if \eqref{weakunitarity} holds) can be understood as an extension of the group  isomorphism $\rho:G_\Pi\to {\rm Stab}(\Pi)$.

Unfortunately, in general we expect the homomorphism \eqref{homtoB} to be neither injective nor surjective. As for surjectivity, recall that a topological defect $\CL$ is invertible if and only if its quantum dimension is $\langle \CL\rangle=1$. On the other hand, it is quite easy to construct elements of $B^{4,20}_{\Pi,+}(\ZZ)$ that have dimension $d=1$ but are not invertible, and therefore are not in the image of $\rho$. Of course, one could put further restrictions on $B^{4,20}_{\Pi,+}(\ZZ)$ by simply excluding such elements. However, we have no guarantee that \emph{all} the elements of $B^{4,20}_{\Pi,+}(\ZZ)$ with dimension larger than $1$ are associated with topological defects.

As for injectivity, we can try to run an argument analogous to the one used in \cite{Gaberdiel:2011fg} to classify the symmetry groups $G_\Pi$. Let $\mathsf{K}_\Pi$ denote the subcategory of topological defects $\CL$ of the model $\calC_\Pi$ preserving $\CN=(4,4)$ and spectral flow, and such that the corresponding maps $\nL$ are proportional to the identity on $V$, i.e. such that
\be \rho(\mathsf{K}_\Pi)\subseteq\{\nL=d\cdot {\rm id}_V,\ d\ge 1\}\subset B^{4,20}_{\Pi,+}(\ZZ)\ .
\ee
We can prove the following:

\begin{claim}\label{th:onlyid}
    For all K3 models $\calC_\Pi$ the category $\mathsf{K}_\Pi$  of defects preserving $\CN=(4,4)$ and spectral flow, and acting by multiplication by some $d\in \RR$ on the space of RR ground fields $V$ is generated by the trivial defect
    \be \mathsf{K}_\Pi=\{d\CI,\ d\in \NN\}\ .
    \ee
\end{claim}

\begin{proof}
    Let $\CL$ be a defect in $\mathsf{K}_\Pi$. Because $\CL$ is transparent to the spectral flow generators, the real number $d$ must be the quantum dimension $d=\langle \CL\rangle$. As a consequence, all R-R ground fields are transparent to $\CL$. This implies that $\CL$ is transparent to all exactly marginal operators of $\calC_\Pi$, and it cannot be lifted by any deformation of the model. Because the moduli space of K3 is connected, this means that the category $\mathsf{K}_\Pi$  is the same for all K3 models. Therefore, it is sufficient to determine $\mathsf{K}_\Pi\equiv \mathsf{K}$ in a specific model. In section \ref{s:Z28M20}, we will show that in a certain torus orbifold the defects of $\mathsf{K}_\Pi$ are necessarily superpositions of $d$ copies of the identity defect (see Claim \ref{th:nootherdef}); in particular, $d$ is a natural number.
\end{proof}

This result generalizes the analogous statement for symmetry groups that the kernel of $\rho$ is trivial. However, in the case of defects, this is not sufficient to conclude that $\rho$ is injective. If $g,h$ are elements of the group $G_\Pi$, then $\rho(g)=\rho(h)$ implies that $\rho(gh^{-1})= \rho(g)\rho(h)^{-1}=1$, and therefore $gh^{-1}\in \ker\rho$ must be the identity and $g=h$. But if $\CL$ and $\CL'$ are non-invertible defects, the fact that $\rho(\CL)=\rho(\CL')$ does not imply that $\CL$ and $\CL'$ can be obtained from each other by fusion with a defect in $\mathsf{K}$. Indeed, in sections \ref{s:torusOrbs}, \ref{s:Z28M20} and \ref{s:onetosix} we will see examples of continuous families of distinct defects $\CL_\theta$, all with the same image $\nL$.

\bigskip

The possible quantum dimensions of defect $\CL\in \cTop_\Pi$ are strongly constrained by properties (a) and (b). In section \ref{s:Dbranes} we will prove the following:
\begin{claim}\label{th:qdim}
    The quantum dimension $\langle \CL\rangle$ of a defect $\CL\in \cTop_\Pi$ is an algebraic integer of degree at most $6$. Furthermore, if $\Pi\cap \Gamma^{4,20}\neq 0$, then $\langle \CL\rangle$ is integral for all $\CL\in \cTop_{\Pi}$.
\end{claim}

We recall that an algebraic integer is the root of a monic polynomial $p(x)$ with integral coefficients, and its degree is $d$ if any such $p(x)$ has degree at least $d$. We do not know whether the upper bound on the degree is sharp. 
A slightly weaker necessary condition for the quantum dimension to be integral is given in proposition \ref{th:intdim} in section \ref{s:Dbranes}.

The condition $\Pi\cap \Gamma^{4,20}\neq 0$ has a nice physical interpretation. Let $v\neq 0$  be a primitive vector in $\Pi\cap \Gamma^{4,20}$. Because the subspace $\Pi$ is positive definite, the vector $v$ has positive norm $v^2>0$. From the viewpoint of type IIA superstring, a primitive vector $v\in \Gamma^{4,20}$ with $v^2>0$ represents the charge of a BPS D0-D2-D4-brane configuration.\footnote{While the charges of BPS D-branes span the lattice $\Gamma^{4,20}$, a \emph{generic} vector $v\in \Gamma^{4,20}$ is the charge of a system of branes and anti-branes that is not by itself BPS.} The mass of such a BPS configuration depends on the moduli, and is proportional to $v_\Pi^2$, where $v_\Pi$ and $v_\perp$ are the orthogonal projections of $v$ along $\Pi$ and $\Pi^\perp$, so that $v^2=v_\Pi^2-v_\perp^2$. An attractor point in the moduli space for the BPS state with charge $v$ is a point where the BPS mass $v_\Pi^2=v^2+v_\perp^2$ is minimized, and  this happens if and only if $v\in \Pi$ \cite{Andrianopoli:1998qg,Dijkgraaf:1998gf,Ferrara:1995ih,Moore:1998pn}. Thus, the points in the moduli space where $\Pi \cap  \Gamma^{4,20}\neq 0$ are exactly the attractor points for some BPS brane configurations. Claim \ref{th:qdim} then implies that whenever the K3 model $\calC$ is `attractive', all topological defects $\CL\in \cTop_\calC$ have integral quantum dimension.

\bigskip
By combining Claims \ref{th:onlyid} and \ref{th:qdim}, we show that generically, i.e. outside of a subset of null measure in the moduli space $\CM_{K3}$, the category $\cTop_\Pi$ is essentially trivial:

\begin{claim}\label{th:generic}
For a generic K3 sigma model $\calC_\Pi$, 
the only topological defects in $\cTop_\Pi$ are integral multiples of the identity. 
\end{claim}
See section \ref{s:Dbranes} for the proof. 
As a consequence, for any non-trivial topological defect $\CL\in \cTop_\Pi$ there is a deformation of the K3 model $\calC_\Pi$ that lifts it. 
This is in particular true for the K3 sigma models allowing for a non-trivial  ${\rm Stab}(\Pi)$, whose elements can in general be characterized as the invertible topological defect lines in $\cTop_\Pi$. We stress that our result does not exclude that the subset of the moduli space where $\cTop$ is non-trivial is dense in $\CM_{K3}$.

\subsection{Action on D-branes}\label{s:Dbranes}

Let us consider the set of BPS boundary states in a K3 model $\calC_\Pi$, corresponding to D-branes in the full string theory that are $1/2$-BPS, i.e. that preserve $8$ out of $16$ space-time supersymmetries.

As in the previous section, we denote by $V$ the $24$-dimensional real space of RR ground states that are CPT self-conjugate. The hermitian form on $\CH$ induces a positive definite bilinear form on $V$.

With each 1/2-BPS boundary state $||\alpha\rrangle$ is associated a $24$-dimensional charge vector $q_\alpha\in V^*$. The pairing of $q_\alpha$ with a R-R ground state $\psi\in V$ is given by the
 amplitude \be q_\alpha(\psi):=\llangle\alpha||\psi\rangle\in \RR\ ,\qquad \ee on a half-cylinder \be S^1\times \RR_{\le 0}=\{(x,t)\mid x\in\RR/2\pi\ZZ,\ t\le 0\}\ee where $|\psi\rangle$ is  the asymptotic state at $t\to-\infty$ and $\alpha$ the boundary condition at $t=0$, see figure \ref{fig9}.

\begin{figure}[hbt]
    \centering
    \includegraphics[scale=0.16]{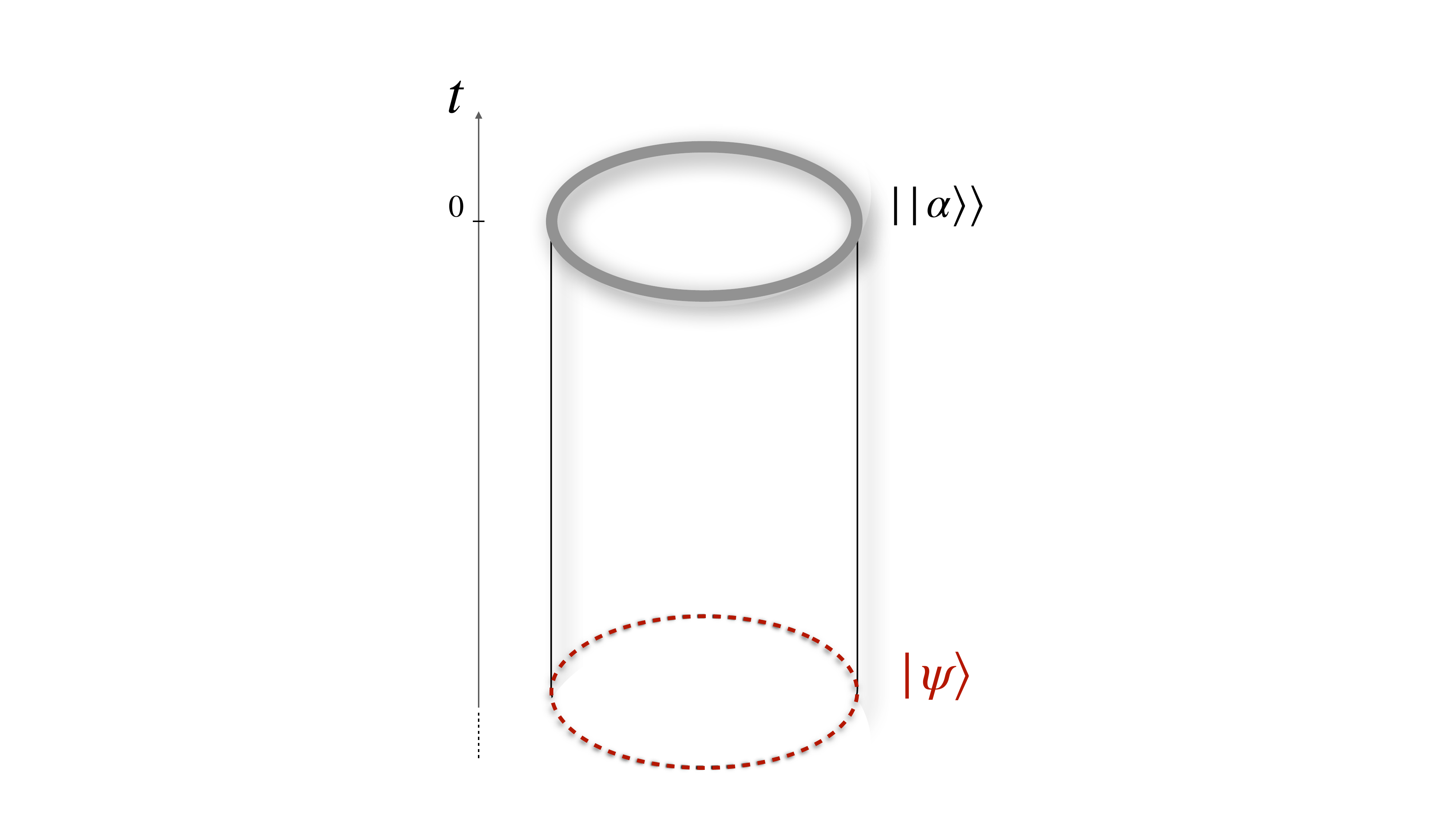}
    \caption{\small Amplitude $\llangle  \alpha \vert\vert \psi \rangle$ in the half-cylinder $S^1 \times \mathbb{R}_{\leq 0}$ with boundary $\vert\vert \alpha \rrangle$ at $t=0$.}
    \label{fig9}
\end{figure}

Geometrically, if $S$ is the target K3 surface, we can identify $V$ with the even real cohomology $H^{even}(S,\RR)$ (in fact, RR ground fields correspond to harmonic forms on $S$), and the lattice of D-brane charges with $H_{even}(S,\ZZ)$, the integral even homology. We can define a bilinear form $(q_\alpha,q_\beta)$ on the lattice of D-brane charges, corresponding to the Mukai pairing on $H_{even}(S,\ZZ)$:
\be (q_\alpha,q_\beta):=\Tr_{\CH^{open}_{R,\alpha,\beta}}(q^{H_{open}}(-1)^F)=\llangle \alpha||(-1)^{F_L+1}{\tilde{q}}^{H_{closed}}||\beta\rrangle_{R-R}
\ee
Physically, the bilinear form is given by a string amplitude on a \emph{bounded} cylinder $S^1\times [0,1]$ with boundaries $\alpha$ and $\beta$, describing a loop of Ramond open strings with periodic conditions for the fermions. 
This amplitude reduces to an index counting Ramond ground states in the space $\CH^{open}_{R,\alpha,\beta}$ of  open strings with Chan-Paton factors $\alpha$--$\beta$. In particular, $(q_\alpha,q_\beta)$  is just an integral (in fact, even) number independent of the size of the cylinder. With respect to this pairing, the RR charge lattice is isomorphic to the (unique) even unimodular lattice with signature $(4,20)$.

In the closed string channel, the amplitude gets contributions only from the Ramond-Ramond ground fields propagating between the corresponding boundary states, with the insertion of a left-moving fermion number $(-1)^{F_L+1}$. The latter acts by $+1$ on $\Pi$ (i.e. on the spectral flow generators in the $(\frac{1}{4},\frac{1}{2};\frac{1}{4},\frac{1}{2})$ representation of $\CN=(4,4)$) and by $-1$ on $\Pi^\perp$ (i.e. on states in $(\frac{1}{4},0;\frac{1}{4},0)$), and can be used to define a non-degenerate bilinear form with signature $(4,20)$ on $V$
\be\label{bilin420} (\psi,\psi'):=-\langle \psi|(-1)^{F}|\psi'\rangle\ .
\ee 
 Let us focus on the dual of the charge lattice,  i.e. the lattice $\Gamma^{4,20}\subset V$ of states with respect to which the charge of any boundary state is integral
\be \Gamma^{4,20}:=\{\Psi\in V\mid \llangle\alpha|\Psi \rangle\in \ZZ\ ,\ \forall || \alpha\rrangle\}\ .
\ee Since the charge lattice is self-dual, $\Gamma^{4,20}$ is again an even unimodular lattice of signature $(4,20)$ with respect to the bilinear form \eqref{bilin420}.

\bigskip 
Let us now consider, as above, a half-cylinder amplitude $\llangle\alpha|\hat \CL|\psi\rangle$ for a boundary $\alpha$ and a state $\psi\in V$, with the insertion of a topological defect line $\CL\in\cTop$ wrapping once along the circle $S^1$.
\begin{figure}[hbt]
    \centering
    \includegraphics[scale=0.2]{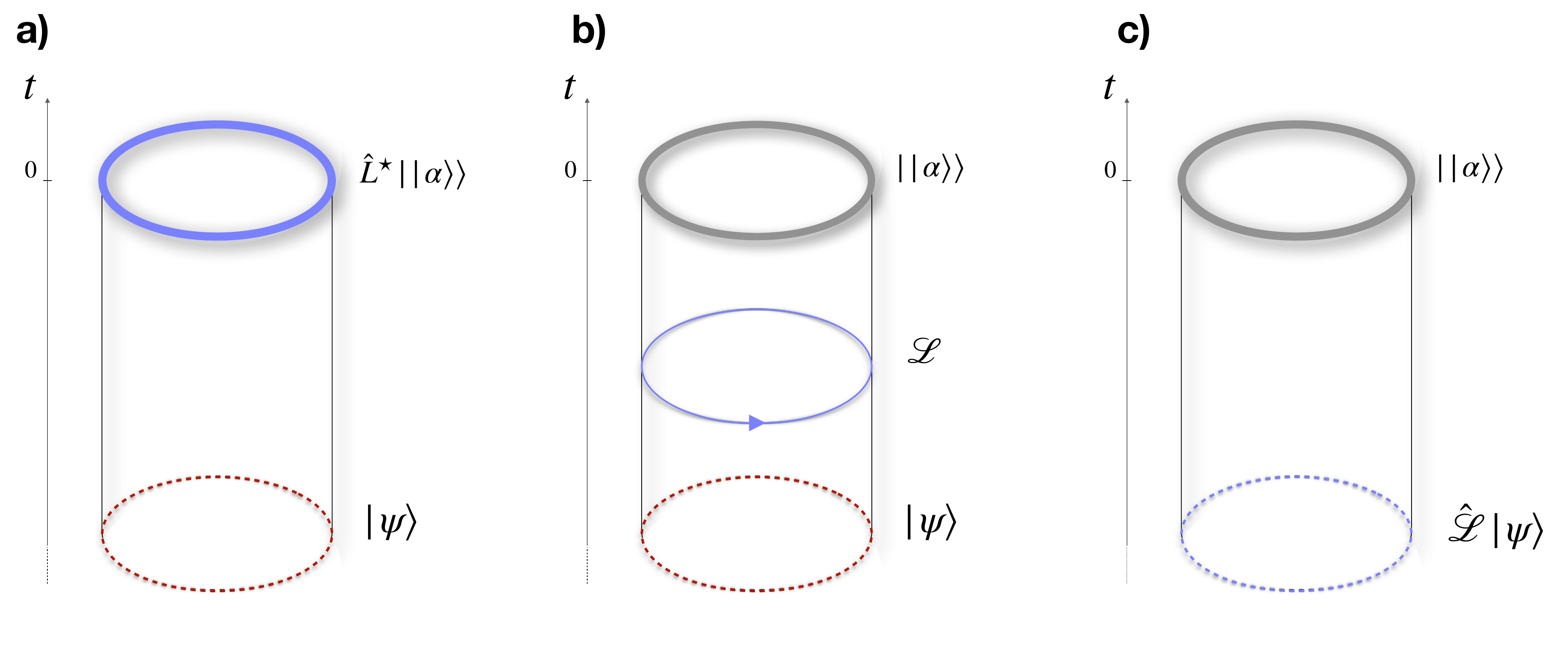}
    \caption{\small Amplitude $\llangle \alpha \vert \hat{\mathcal{L}} \vert \psi \rangle$ in the half-cylinder $S^1 \times \mathbb{R}_{\leq 0}$ with insertion of a TDL $\mathcal{L} \in Top$.}
    \label{fig10}
\end{figure}

 By moving the line $\CL$ along the cylinder to $t\to -\infty$  or $t\to 0$,  $\llangle\alpha|\hat \CL|\psi \rangle$ can be interpreted as either the RR charge of the D-brane $\alpha$ with respect to $\hat\CL|\psi\rangle$, or as the charge of $\CL^\dual ||\alpha\rrangle$ with respect to $|\psi\rangle$. Therefore, 
 \be q_\alpha(\hat\CL\psi)=\llangle\alpha|\hat\CL|\psi\rangle=q_{\hat\CL^\dual\alpha}(\psi)\ .
 \ee As in the previous section, we denote by
 $ \nL:=\hat\CL_{\rvert V}\ ,
 $ the restriction of $\hat\CL$ to $V$.
  In particular, if $\Psi\in \Gamma^{4,20}\subset V$, then
 \be q_\alpha(\nL\Psi)=q_{\hat\CL^\dual\alpha}(\Psi)\in \ZZ\ ,
 \ee for all $||\alpha\rrangle$. It follows that $\nL(\Psi)\in \Gamma^{4,20}$, i.e. $\hat\CL$ gives rise to an endomorphism $\nL\in {\rm End}_\ZZ(\Gamma^{4,20})$ of the unimodular lattice $\Gamma^{4,20}$
 \begin{align*}
     \nL:\Gamma^{4,20}&\to\Gamma^{4,20}\\
     \Psi&\mapsto \nL(\Psi)\ .
 \end{align*}
 The operator $\hat\CL$ commutes with the left-moving fermion number $(-1)^{F_L}$
\be [(-1)^{F_L},\hat\CL]=0\ ,
\ee so that the linear map $\nL$ preserves the orthogonal decomposition $V=\Pi\oplus \Pi^\perp$. In particular, it is block diagonal with respect to an orthonormal basis $\{\psi_1,\ldots,\psi_{24}\}\subset V$ satisfying \be\label{orthspecial} \psi_1,\ldots,\psi_4\in \Pi\ ,\qquad \psi_5,\ldots,\psi_{24}\in \Pi^\perp\ .
\ee Finally, we notice that on a spectral flow generators $\psi\in \Pi$, $\nL$ acts by
\be \nL(\psi)=\langle \CL\rangle \psi\ ,\qquad \forall\psi\in \Pi\subset V\ ,
\ee i.e. the $4\times 4$ upper-left block of $\nL$ is just $\langle \CL\rangle$ times the identity. This argument leads to the conditions (a) and (b) described in section \ref{s:topdefsK3}.

\medskip

An easy but powerful consequence of this construction is the following:

\begin{proposition}\label{th:intdim}
    Let $\CL\in\cTop_{\calC_\Pi}$, and suppose there is some $\Psi\in \Gamma^{4,20}$, $\Psi\neq 0$, such that $\nL(\Psi)=\langle \CL\rangle\Psi$, where $\langle\CL\rangle$ is the quantum dimension of $\CL$. Then, $\langle \CL\rangle$ is integral. In particular, if for a certain model $\calC_\Pi$, one has $\Pi\cap \Gamma^{4,20}\neq 0$, then all $\CL\in \cTop_{\calC_\Pi}$ have integral quantum dimension.
\end{proposition}
Indeed, without loss of generality, we can assume the vector $\Psi\in \Gamma^{4,20}$ is primitive (i.e., not an integral multiple of a shorter vector in the lattice). Then, because $\langle\CL\rangle\Psi\in \Gamma^{4,20}$, it must necessarily be an integral multiple of $\Psi$.

More generally, because the matrix $\nL$ representing the action of $\hat\CL$ on the space of RR ground fields is an integral $24\times 24$ matrix, it satisfies $p(\nL)=0$ where $p(x)$, the characteristic polynomial, is a monic polynomial of degree $24$ with integral coefficients. This implies that the quantum dimension $\langle \CL\rangle$, which is an eigenvalue of $\nL$, is also a root of the same monic polynomial, i.e. it is an algebraic integer. Let $d=\langle\CL\rangle$ be an algebraic integer, and let $r(x)$ be the minimal polynomial for $d$, i.e. the least degree monic polynomial with integral coefficient having $d$ as a root. Since $d$ has multiplicity at least $4$ as an eigenvalue, it follows that $r(x)^4$ divides $p(x)$, so that $r(x)$ must have degree at most $6$. We conclude that:

\begin{proposition}
    The quantum dimension $\langle \CL\rangle$ of any defect $\CL\in\cTop$ is an algebraic integer of degree at most $6$.
\end{proposition}

This result, together with Claim \ref{th:onlyid}, leads us to a proof of Claim \ref{th:generic}.

\begin{proof}[Proof of Claim \ref{th:generic}.]  Consider a K3 model $\calC_\Pi$. Let $\Psi_1,\ldots,\Psi_{24}$ be generators of the lattice $\Gamma^{4,20}\subset V$, let $d:=\langle \CL\rangle$ be the quantum dimension of a defect $\CL\in \cTop_\Pi$, and let $\QQ[d]$ be the extension of the field $\QQ$ by the algebraic integer $d$.    Then, for a generic model $\calC_\Pi$, we expect the scalar products $\langle \psi|\Psi_1\rangle,\ldots,\langle \psi|\Psi_{24}\rangle$, with some spectral flow generator $\psi\in \Pi$, to generate a space of maximal rank $24$ over the field extension $\QQ[d]$. Indeed, because the set of algebraic numbers has measure zero in $\RR$, for a generic model $\calC_\Pi$ every non-trivial linear relation among $\langle \psi|\Psi_1\rangle,\ldots,\langle \psi|\Psi_{24}\rangle$ will contain some transcendental coefficient. 
Let us define the matrix elements $\nL_{ij}$ by $\nL|\Psi_i\rangle=\sum_j \nL_{ij}|\Psi_j\rangle$, so that $\nL_{ij}\in \ZZ$. Using
\be d\langle\psi|\Psi_i\rangle= \langle\psi|\nL|\Psi_i\rangle=\sum_{j=1}^{24} \nL_{ij} \langle\psi|\Psi_j\rangle\ ,
\ee we get
\be \sum_j (d\delta_{ij}-\nL_{ij})\langle\psi|\Psi_j\rangle=0\qquad \forall i=1,\ldots,24\ .
\ee Because $d\delta_{ij}-\nL_{ij}\in \QQ[d]$ and $\langle\psi|\Psi_j\rangle$ are linearly independent over $\QQ[d]$, it follows that
\be d\delta_{ij}-\nL_{ij}=0\qquad \forall i,j=1,\ldots,24\ .
\ee This means that $\nL$ is $d$ times the identity, and by proposition \ref{th:intdim} $d$ is integral. By Claim \ref{th:onlyid}, this means that $\CL=d\CI$, and we conclude. \end{proof}

\subsection{Defects as boundaries in the doubled theory and other approaches}\label{s:doubletheory}

In this section we discuss other approaches to topological defects that are not used elsewhere in this article.

First of all, a topological defect in a CFT $\calC$ can be described as a boundary state in the doubled theory $\calC\times\calC^*$, where $\calC^*$ is obtained from $\calC$ by worldsheet parity. In particular, when $\calC$ is a sigma model on K3, the doubled theory $\calC\times\calC^*$ is a sigma model on the Cartesian product of two K3 surfaces. This CFT contains a $\CN=(4,4)$ superconformal algebra at central charge $c=\tilde c=12$. Notice that the $\CN=4$ algebra at $c=12$ contains a copy of $\widehat{su}(2)_2$ at level $2$, rather than $1$. 

The space of R-R ground fields of weight $h=\bar h=\frac{1}{2}$ of $\calC\times \calC^*$ has dimension $24^2=576$, and is isomorphic to $V\otimes V^*\cong {\rm End}_\RR(V)$. Similarly, the lattice of RR charges is the tensor product\footnote{Actually, because $\Gamma^{4,20}$ is unimodular, the non-degenerate bilinear form defines isomorphisms $\Gamma^{4,20}\cong(\Gamma^{4,20})^*$ and $V\cong V^*$. } $\Gamma^{4,20}\otimes (\Gamma^{4,20})^*\cong  {\rm End}_\ZZ(\Gamma^{4,20})$. Therefore, for each given $\CL\in \cTop_\calC$ one can naturally identify $\nL\in {\rm End}_\ZZ(\Gamma^{4,20})$ with the RR charge of the corresponding boundary state in $\calC\times \calC^*$. 
It follows that property (a) is just quantization of RR charges in the doubled theory.

In order to understand the analogue of property (b), let us consider the decomposition of the space \be V\otimes V^*\cong V\otimes V=(\Pi\oplus\Pi^\perp)\otimes (\Pi\oplus\Pi^\perp) \ee into representations of the $\CN=(4,4)$ algebra. There are three kinds of (short) Ramond representations of $\CN=4$ at $c=12$ with $h=\frac{1}{2}$, that are labeled by the highest weight (charge) $q\in \{0,1/2,1\}$ with respect to the $\widehat{su}(2)_2$ algebra. In particular, the subspace $\Pi^\perp\otimes \Pi^\perp$ is contained in $20^2=400$ representations of $\CN=(4,4)$ with holomorphic and antiholomorphic charges $q=\tilde q=0$; the $80$-dimensional subspaces $\Pi\otimes \Pi^\perp$ and $\Pi^\perp\otimes \Pi$ are contained in $40$ $\CN=(4,4)$-representations with $(q,\tilde q)=(\frac{1}{2},0)$ and $(q,\tilde q)=(0,\frac{1}{2})$, respectively; finally, the $16$-dimensional subspace $\Pi\otimes \Pi$ decomposes into one $\CN=(4,4)$-representation with $(q,\tilde q)=(1,1)$ ($9$ RR ground states), one representation with $(q,\tilde q)=(1,0)$ ($3$ states), one with $(q,\tilde q)=(0,1)$ ($3$ states) and one with $(q,\tilde q)=(0,0)$ (one state). This gives an orthogonal decomposition
\be V\otimes V^*=\Sigma_{0,0}\oplus \Sigma_{\frac{1}{2},0}\oplus \Sigma_{0,\frac{1}{2}}\oplus \Sigma_{1,0}\oplus \Sigma_{0,1}\oplus \Sigma_{1,1}\ ,
\ee where $\Sigma_{q,\tilde q}$ is the subspace of RR ground states of $\calC\times\calC^*$ contained in the $(q,\tilde q)$  representation of $\CN=(4,4)$. Here,
$$ \dim\Sigma_{0,0}=401\ ,\quad  \dim\Sigma_{\frac{1}{2},0}=\dim\Sigma_{0,\frac{1}{2}}=80\ ,\quad \dim\Sigma_{1,0}=\dim\Sigma_{0,1}=3\ ,\quad \dim\Sigma_{1,1}=9\ .
$$
It is clear that property (b) for $\CL\in\cTop_\calC^*$ is equivalent to the condition the corresponding boundary state in $\calC\times\calC^*$ is charged only under the $401$ RR ground states in $\Sigma_{0,0}$, and neutral with respect to the other ones. In particular, the condition that the restriction of $\nL$ to $\Pi$ is proportional to the identity implies that the map $\nL$ commutes with the left- and right-moving $SU(2)$ R-symmetry groups; in turn, this means that the corresponding boundary state in $\calC\times \calC^*$ can only be charged under the singlet in $\Pi\otimes \Pi$. In other words, the space $B^{4,20}(\RR)$ can be identified with the $401$-dimensional subspace $\Sigma_{0,0}$ of R-R ground states of $\calC\times\calC^*$.

This construction also provides some intuition as for why, for a generic K3 sigma model $\calC$, the only simple defect in $\cTop_\calC$ is the identity. 
Indeed, while there are infinitely many boundary states in $\calC\times\calC^*$, whose RR charges span the whole lattice $\Gamma^{4,20}\otimes (\Gamma^{4,20})^*$, it is not obvious at all that \emph{any} of these lattice vectors is contained in the subspace $\Sigma_{0,0}\equiv B^{4,20}(\RR)$. Let us elaborate on this point in more detail. We know that in any K3 model $\calC$ there is the identity defect $\CI$, and the charge of the corresponding boundary state in the doubled theory is an element $\mathsf{1}\in \Gamma^{4,20}\otimes (\Gamma^{4,20})^*$ lying in the $401$-dimensional subspace $\Sigma_{0,0}$ with $(q,\tilde q)=(0,0)$. Under the identification $\Gamma^{4,20}\otimes (\Gamma^{4,20})^*\cong {\rm End}(\Gamma^{4,20})$, $\mathsf{1}$ simply corresponds to the identity map. Therefore, the intersection 
\be\label{lattsigma00} (\Gamma^{4,20}\otimes (\Gamma^{4,20})^*)\cap \Sigma_{0,0}={\rm End}(\Gamma^{4,20})\cap B^{4,20}(\RR) = B^{4,20}_\Pi(\ZZ) 
\ee is always at least one-dimensional.
Now, deformations of the K3 model correspond to $O(4,20,\RR)$ transformations of $V\cong \Pi\oplus \Pi^\perp$ within $\Gamma^{4,20}\otimes \RR$. While $V$ is an irreducible representation of the orthogonal group $O(4,20,\RR)$ (the vector representation), the tensor product $V\otimes V^*\cong V^{\otimes 2}$ decomposes as $$V^{\otimes 2}=\mathbf{1}\oplus \wedge^2 V\oplus {\rm Sym^2_0}V\ ,$$ where $\mathbf{1}$, $\wedge^2 V$, and ${\rm Sym^2_0}V$ are, respectively, the trivial, the anti-symmetric and the traceless symmetric representations of $O(4,20,\RR)$. 
Under the identification $V\otimes V^*\cong {\rm End}_\RR(V)$, clearly the trivial representation $\mathbf{1}$ corresponds to the identity map, i.e. to the charge $\mathsf{1}$ of the identity defect $\CI$. Because it is invariant under $O(4,20,\RR)$ transformations, the charge $\mathsf{1}\in \Gamma^{4,20}\otimes (\Gamma^{4,20})^*$ is always contained in $\Sigma_{0,0}$. This fits with the obvious fact that the identity defect is not lifted by any deformation of the K3 model. On the other hand, both irreducible representations $\wedge^2 V$, and ${\rm Sym^2_0}V$ have non-trivial intersection with $\Sigma_{0,0}$ and with other $\Sigma_{q,\tilde q}$. Thus, a generic $O(4,20,\RR)$ transformation will mix any vector in the $400$-dimensional space $$(\Sigma_{0,0}\cap \mathsf{1}^\perp) \subset \wedge^2 V\oplus {\rm Sym^2_0}V$$ with vectors in the other components $\Sigma_{q,\tilde q}$. Therefore, generically, we do not expect any other lattice vector in $\Gamma^{4,20}\otimes (\Gamma^{4,20})^*$ to lie along $\Sigma_{0,0}\cap \mathsf{1}^\perp$. This means that the intersection \eqref{lattsigma00} is generically one-dimensional, and contains only the maps $\nL$ that are multiple of the identity. By claim \ref{th:onlyid}, we conclude that $\CL$ is a superposition of identity defects.\\
We stress that this result is not in contrast with the possibility that $\cTop_\Pi$ is non-trivial in a subset of null measure in the moduli space -- in fact, we know that there are families of K3 models with non-trivial symmetry group $G_\Pi$. Indeed, one can tune a $O(4,20,\RR)$ rotation so that the space $\Sigma_{0,0}$ contains a non-trivial sublattice of $\Gamma^{4,20}\otimes (\Gamma^{4,20})^*$. In fact, this happens whenever $\Sigma_{0,0}\cap \mathsf{1}^\perp$ has a non-trivial intersection with the rational space $(\Gamma^{4,20}\otimes (\Gamma^{4,20})^*)\otimes \QQ\subset V\otimes V^*$. Starting from such a point in the moduli space $\CM_{K3}$ of K3 models, corresponding to some $\Pi\subset V$, and acting by a \emph{rational} orthogonal transformation $O(4,20,\QQ)$, one gets another point in $\CM_{K3}$ with non-trivial intersection, and therefore potentially non-trivial defects.  Thus, the subset of $\CM_{K3}$ where the category $\cTop$ is potentially non-trivial is dense in $\CM_{K3}$ (though we stress that, apart from invertible symmetries, we are not able to prove that such potential defects actually exist). Our argument, however, shows that such a subset cannot contain a full open neighborhood of any point in the moduli space, and therefore has zero measure. Indeed, a \emph{real} $O(4,20,\RR)$ transformation of $\Pi\subset V$, even an infinitesimal one, will generically lift any intersection.\footnote{The subset of $\CM_{K3}$ with non-trivial $\cTop_\Pi$, however, might contain submanifolds of strictly positive codimension, corresponding to rotations of $\Pi$ by proper subgroups of the real group $O(4,20,\RR)$. Again, this is compatible with the known fact that there are families of K3 models with the same non-trivial group of symmetries.}

There is one point in our argument where we have been slightly naive. In general, boundary states in $\calC\times\calC^*$ correspond to defects in $\calC$ that are conformal, but not necessarily topological, i.e. they preserve a particular linear combination of the holomorphic and anti-holomorphic stress energy tensor, but not $T(z)$ and $\bar T(\bar z)$ separately. The topological defects in $\cTop$ correspond to a particular choice of gluing conditions for boundary states in $\calC\times\calC^*$, where the two copies of the holomorphic $\CN=4$ algebra are identified with the two anti-holomorphic copies up to an automorphism exchanging them (i.e., they are `permutation branes'). In our treatment, we are ignoring this restriction, and therefore what we get is an upper bound on the actual number of topological defects. 
Boundary states satisfying the correct gluing conditions are in one-to-one correspondence with topological defect operators $\hat\CL:\CH\to\CH$ commuting with the $\CN=(4,4)$ algebra. The requirement that the defect preserves the spectral flow puts further constraints on the R-R charge of the corresponding boundary states. In particular, the boundary state must preserve space-time supersymmetry, so that it must carry a non-zero R-R charge in $\Gamma^{4,20}\otimes (\Gamma^{4,20})^*$. 
The argument above then shows that, outside of a null measure subset in the moduli space, all operators $\hat\CL$ corresponding to a defect $\CL\in \cTop_\Pi$ must act on the space of R-R ground fields as a multiple of the identity. Then, one needs to use claim \ref{th:onlyid} to conclude that all such defects $\CL$ are multiples of the trivial one. 

\bigskip

Let us now comment on other common approaches to determine topological defects in two dimensional conformal field theories. One simple and powerful method pioneered by Petkova and Zuber in \cite{Petkova:2000ip} consists in imposing the analogue of the Cardy condition for boundary states. In practice, a generic defect $\CL$ preserving the chiral and antichiral algebra $\CA\times\bar \CA$ is parametrized in terms of the (unknown) eigenvalues of $\hat\CL$ on the primary fields. One can then compute the torus $\CL$-twined partition function, obtained by inserting the $\CL$-line along the `space' direction of the worldsheet torus, as a function of these parameters. A modular S-transformation relates this $\CL$-twined partition function to the $\CL$-twisted one, where the defect runs along the `time' direction. Because the $\CL$-twisted partition function is just a trace over the  $\CL$-twisted space $\Hh_\CL$, it must decompose into characters of the $\CA\times\bar\CA$ algebra with non-negative integral coefficients. 
Imposing such property on the latter coefficients gives a set of quantization conditions on the unknown parameters of the defect $\CL$, whose solutions form a lattice. In rational CFTs, where there is a finite number of primary fields and of unknown parameters, this method is very effective and very often allows one to determine all possible topological defects in the theory.\footnote{More generally, in order to determine all simple defects in a rational CFT, one needs to impose a Cardy-like condition to all possible products $\CL_1\CL_2^\dual$, for all pairs of simple defects $\CL_1$, $\CL_2$. In the doubled theory, the Petkova-Zuber method just corresponds to imposing the Cardy condition for open strings stretched between any pair of boundary states.} 

The main difficulty in applying this method to study $\cTop_\calC$ in a K3 model $\calC$ is that the theory is not rational with respect to the $\CN=(4,4)$ algebra, so that there infinitely many primary operators and unknown parameters. It is known that the character of the $\CN=4$ algebra (in particular, for the BPS representations) are mock modular forms \cite{Eguchi:1987wf,Eguchi:1988af,Eguchi:1987sm}, and their S-tranformations give rise to a continuum of massive characters. 
As is well-known from the study of boundary states in non-rational theories, imposing Cardy-like conditions to these cases is technically very difficult.

Despite these obstacles, modular properties of torus amplitudes have been successfully used in the past to get information about the action of finite symmetry groups on the space $\Hh$ (or, at least, on some subspace) of states  of a K3 model. In particular, when $\CL$ is an invertible defect preserving the superconformal algebra, one can consider the $\CL$-twined elliptic genus, which is a (weakly) holomorphic Jacobi function for a certain subgroup of the modular group $SL(2,\ZZ)$. Because the space of such Jacobi functions is finite dimensional, the problem is once again reduced to determining a finite number of unknown coefficients. In  particular, these methods were applied in the context of Moonshine conjectures for K3 models \cite{Cheng:2010pq,Eguchi:2010fg,Eguchi:2010ej,Gaberdiel:2010ca,Gaberdiel:2010ch,Paquette:2017gmb}.

One crucial step in these methods is the fact that, for an invertible symmetry of order $n$, the $\CL$-twined genus is modular with respect to a level $n$ congruence subgroup of $SL(2,\ZZ)$. This is a consequence of the fusion relation $\CL^n=\CI$. Furthermore, all possible orders $n$ of these symmetries can be determined by the general classification theorem of \cite{Gaberdiel:2011fg}.  Unfortunately, when $\CL$ is an unknown non-invertible defect in some K3 model, we do not have any information about the relations in the fusion ring generated by $\CL$. In principle, the decomposition of $\CL^k$ might involve new simple defects at every power $k$. For this reason, one cannot predict what the modular properties of the $\CL$-twined genus are. In fact, we hope that the methods developed in the present article will provide enough information about the fusion ring of $\cTop_\calC$ so as to make these methods effective.

\bigskip

Finally, when the CFT $\calC$ can be obtained as a IR fixed of a RG flow from a Landau-Ginzburg model, there are well-developed techniques to study its topological defects and their fusion with supersymmetric boundary conditions, relying in particular on matrix factorization \cite{Behr:2020gqw,Brunner:2014lua,Brunner:2020miu,Brunner:2007qu,Brunner:2009zt,Carqueville:2012st,Carqueville:2012dk}. It would be interesting to apply these methods to K3 models and use them to verify and extend the general results of our article.  %

\section{Topological defects in torus orbifolds}\label{s:torusOrbs}

Many interesting examples of K3 models can be described in terms of orbifolds of supersymmetric sigma models on $T^4$ by some finite group of symmetries. In this section, we will show that for all such K3 models $\calC$ the category $\cTop_\calC$ contains some family of simple defects $\CL_\theta$ parametrized by continuous parameters $\theta$. This result is an immediate generalization of known properties of the orbifold $S^1/\ZZ_2$ of a single free boson on a circle \cite{Becker:2017zai,Chang:2020imq,Fuchs:2007tx,Thorngren:2021yso}.

\subsection{Supersymmetric sigma models on $T^4$}\label{s:generalT4}

A supersymmetric sigma model $\CT$ on $T^4$ contains four holomorphic and four anti-holomorphic $u(1)$ currents $i\partial X^k$, $i\bar\partial X^k$, $k=1,\ldots, 4$, as well as their superpartners, the four holomorphic and anti-holomorphic free fermions $\psi^k$, $\tilde \psi^k$ with weights $(h,\bar h)=(\frac{1}{2},0)$ and $(0,\frac{1}{2})$, respectively. The full current algebra is $\widehat{so}(4)_1\oplus u(1)^4$, where the six currents of $\widehat{so}(4)_1\cong \widehat{su}(2)_1\oplus \widehat{su}(2)_1$ are given by normal-ordered products $:\!\!\psi\psi\!\!:$ of pairs of free fermions.  These fields generate the full chiral and anti-chiral algebra of the $T^4$ model at a generic point in the moduli space -- the algebra can be enhanced at subloci of positive codimension. 

The NS-NS primary operators with respect to this algebra  are the vertex operators $V_\lambda(z,\bar z)$ labeled by vectors \be \lambda\equiv (\lambda_L,\lambda_R)=(\lambda^1_L,\ldots,\lambda^4_L;\lambda^1_R,\ldots,\lambda^4_R)\in \RR^{8}\ ,\ee with conformal weights $h=\frac{\lambda_L\cdot \lambda_L}{2}$ and $\bar h=\frac{\lambda_R\cdot \lambda_R}{2}$. For each given $T^4$ model, the allowed vectors $\lambda\in \RR^8$ are the points of a $8$-dimensional lattice $\Gamma^{4,4}$ (the Narain lattice) that is even and unimodular with respect to the bilinear form
\be (\lambda,\mu):=\lambda_L\cdot\mu_L-\lambda_R\cdot\mu_R\ ,\qquad \lambda,\mu\in\Gamma^{4,4}\ , 
\ee with signature $(4,4)$. Notice that all even unimodular lattices $\Gamma^{4,4}$ with signature $(4,4)$ are isomorphic to each other. There is a (non-unique) choice of basis $\{n_1,\ldots,n_4,w_1,\ldots,w_4\}$ for $\Gamma^{4,4}$ with respect to which the bilinear form is
\be\label{basis44} ( n_i,n_j) =0=( w_i,w_j)\ ,\qquad ( n_i,w_j)=\delta_{ij}\ .\ee The vector $\lambda\in \RR^8$  represents the weights of the corresponding primary state $|\lambda\rangle$ with respect to the $U(1)^4\times U(1)^4$ group generated by the four bosonic currents $i\partial X^k$, $i\bar\partial X^k$. In other words, the invertible defects \be W_\theta\ ,\qquad\qquad \theta\in (\Gamma^{4,4}\otimes\RR)/\Gamma^{4,4}\cong (\RR/\ZZ)^8 \ ,\ee corresponding to the $U(1)^8$ symmetries, commute with the full chiral algebra and act by
\be \hat W_\theta|\lambda\rangle= e^{2\pi i (\theta,\lambda)}|\lambda\rangle\ ,
\ee on the primary state $\lambda$. Geometrically, $U(1)^8\cong U(1)^4\times U(1)^4$ is the product of the $U(1)^4$ group of translations along the four direction of the torus $T^4$, times the $U(1)^4$ group of translations along the T-dual torus.

The RR sector representations of the chiral algebra are labeled by the same vectors $\lambda\in \Gamma^{4,4}$. For each such $\lambda\in \Gamma^{4,4}$, there are $16$ degenerate ground states with conformal weights $h=\frac{\lambda_L\cdot \lambda_L}{2}+\frac{1}{4}$ and $\bar h=\frac{\lambda_R\cdot \lambda_R}{2}+\frac{1}{4}$, which form an irreducible representation of the Clifford algebra of zero modes of the $8$ fermions $\psi^k\,\tilde \psi^k$, $k=1,\ldots,4$. With respect to the holomorphic and antiholomorphic $\widehat{so}(4)_1\cong \widehat{su}(2)_1 \oplus \widehat{su}(2)_1 $ algebras and the corresponding $Spin(4)\cong SU(2)\times SU(2)$ groups, the $16$ ground states decompose into four irreducible representations with weights $(s_1,s_2;\tilde s_1,\tilde s_2)$
\be (\frac{1}{2},0;\frac{1}{2},0)\ , \qquad (\frac{1}{2},0;0,\frac{1}{2})\ , \qquad (0,\frac{1}{2};\frac{1}{2},0)\ , \qquad (0,\frac{1}{2};0,\frac{1}{2})\ .
\ee

The chiral algebra of the $T^4$ model contains several copies of the small $\CN=4$ superconformal algebra, corresponding to a choice of $\widehat{su(2)}_1$ current algebra (the R-symmetry algebra) inside $\widehat{so(4)}_1$; we choose a small $\CN=(4,4)$ once and for all. The four R-R ground states with $\lambda=0$ (i.e. conformal weights $h=\bar h=\frac{1}{4}$) that are charged with respect to both the holomorphic and antiholomorphic $\widehat{su}(2)_1$ algebra are the spectral flow generators, and belong to a single $(\frac{1}{4},\frac{1}{2};\frac{1}{4},\frac{1}{2})$ representation of $\CN=(4,4)$. The remaining $12$ RR ground states belong to two $(\frac{1}{4},\frac{1}{2};\frac{1}{4},0)$, two $(\frac{1}{4},0;\frac{1}{4},\frac{1}{2})$, and four $(\frac{1}{4},0;\frac{1}{4},0)$ representations of $\CN=(4,4)$.

For a generic $T^4$ model $\CT$, the group $G_\CT$ of symmetries preserving the small $\CN=(4,4)$ algebra and the spectral flow generators is 
\be G_\CT\cong U(1)^8\rtimes \ZZ_2\ ,
\ee where the $\ZZ_2$ symmetry $\calR$ is the coordinate reflection
\be\label{refl} \partial X^k\to -\partial X^k\ ,\qquad \bar\partial X^k\to -\bar\partial X^k\qquad \psi^k\to-\psi^k\ ,\qquad\tilde\psi^k\to-\tilde\psi^k\ .
\ee The fusion of $\calR$ with $W_\theta$ is
\be \calR W_\theta=W_{-\theta}\calR\ .
\ee At special loci in the moduli space, $G_\CT$ can be enhanced to a larger group. In general, the group $G_{\CT}$ is always a group extension
\be 1\longrightarrow U(1)^8 \longrightarrow G_\CT \longrightarrow G^0_\CT\longrightarrow 1\ ,
\ee where $G^0_\CT\subset O(\Gamma^{4,4})\cong O(4,4,\ZZ)$ is a finite group of lattice automorphisms, that always contains  the reflection \eqref{refl} as a central element. Physically, $O(\Gamma^{4,4})\cong O(4,4,\ZZ)$ is the group of torus T-dualities, and $G^0_\CT$ is the group self-dualities of the model $\CT$. See \cite{Volpato:2014zla} for a complete classification of the groups $G^0_\CT$.

\bigskip

The orbifold of a $T^4$ model $\CT$ by a non-anomalous finite group of symmetries $H\subset G_{\CT}$ is again a $\CN=(4,4)$ superconformal field theory at $c=\bar c=6$, so it is a sigma model on either K3 or $T^4$. In this case, the elements $g\in G_{\CT}$ that normalize $H$, i.e. such that $gHg^{-1}=H$, commute with $\sum_{h\in H} \CL_h$ and therefore give rise to invertible defects in the orbifold theory $\CT/H$. On the other hand, if $g\in G_{\CT}$ does not commute with $\sum_{h\in H} \CL_h$, then it can induce a non-invertible defect in $\CT/H$. In the next sections, we will see some examples where $H$ is cyclic and the orbifold $\CT/H$ is a K3 model. 

If $H$ is a finite non-anomalous subgroup of $U(1)^8$, then the orbifold $\CT/H$ is again a torus model. 
If the symmetry group of $\CT$ is the generic $G_\CT=U(1)^8\rtimes \ZZ_2$, then every $g\in G_{\CT}$ normalizes\footnote{Notice in particular that $\calR h\calR=h^{-1}$ for all $h\in U(1)^8$.} $H\subset U(1)^8$, and therefore necessarily induces an invertible symmetry in $\CT/H$. On the other hand, if one starts from a model $\CT$ with non-trivial $G^0_\CT$, then one can get non-invertible defects in infinitely many different torus models $\CT/H$, by varying the orbifold subgroups $H\subset U(1)^8$. This is the higher dimensional analogue of a phenomenon observed for a single free boson on $S^1$ \cite{Fuchs:2007tx,Bachas:2012bj}: at the self-dual radius $R_{sd}$ there is a $SU(2)\times SU(2)$ group of symmetries commuting with the Virasoro algebra; by taking orbifolds by suitable non-anomalous $H=\ZZ_N\times \ZZ_M$ groups, one can get $S^1$ models at every rational multiple $\frac{N}{M}R_{sd}$ of the self-dual radius, so that in every such model there are non-invertible topological symmetries labeled by a pair $(g,h)\in SU(2)\times SU(2)$ \cite{Fuchs:2007tx}. A simple $T^4$ model example is given by the product $\CT=(S^{1})^4$ of four $S^1$ free bosons at the self-dual radius (this corresponds to the model denoted as  $A_1^4$ in \cite{Volpato:2014zla}). The model $\CT$ is self-dual under T-duality along each of the four circles; furthermore, T-dualities in an even number of direction preserve the $\CN=(4,4)$ algebra and spectral flow, so that they correspond to elements $g\in G_{\CT}$. By taking orbifolds by finite subgroups $H\subset U(1)^8$, one can get infinitely many $T^4$ models with non-invertible topological defects induced by $g$. We stress that while the subset of such $T^4$ models is dense in the moduli space, it is still of measure zero.

\subsection{Continuous defects in $T^4/\ZZ_2$.}\label{s:contDefects}

Given any torus model $\CT$, the orbifold $\calC=\CT/\ZZ_2$ by the $\ZZ_2$ symmetry \eqref{refl} is a non-linear sigma model on K3. 

A torus orbifold $\calC=\CT/\ZZ_2$ always admits an invertible defect $Q$ (the quantum symmetry) acting by $-1$ on the twisted sector and trivially on the untwisted one. Besides $Q$, it also contains some invertible simple defects that are induced by the simple defects $W_\theta$ of the torus model $\CT$ that commute with $\calR$, i.e such that $\theta\equiv -\theta\mod \Gamma^{4,4}$. More precisely, with each defect $W_{\frac{\lambda}{2}}$, $\lambda\in \Gamma^{4,4}/2\Gamma^{4,4}$ of $\CT$, are associated two defects
$\eta_\frac{\lambda}{2},\eta'_{\frac{\lambda}{2}}\equiv Q\eta_\frac{\lambda}{2}$ of $\calC$. In particular, for $\lambda=0$, one has $\eta_0=I$ (the identity defect) and $\eta'_0=Q$. 

While the defects $W_{\frac{\lambda}{2}}$, $\lambda\in \Gamma^{4,4}$, generate an abelian group of symmetries $\ZZ_2^8\subset U(1)^8$ of the torus model $\CT$, the group generated by the $\eta_{\lambda/2}$ is a non-abelian extension of $\ZZ_2^8$ by a central $\ZZ_2$ generated by the quantum symmetry $Q$. Consider a basis $\{n_1,\ldots,n_4,w_1,\ldots,w_4\}$ for $\Gamma^{4,4}$ as in \eqref{basis44}. The defects $\eta_{n_i}$, $\eta_{w_i}$ and $Q$ obey the following relations
\be \eta_{\frac{n_i}{2}}^2=I=\eta_{\frac{w_i}{2}}^2\ ,\qquad Q^2=1\ ,\qquad \eta_{\frac{n_i}{2}}Q=Q\eta_{\frac{n_i}{2}}\ ,\qquad \eta_{\frac{w_i}{2}}Q=Q\eta_{\frac{w_i}{2}}
\ee
 If $\lambda=\sum_i (a_in_i+b_iw_i)$, $a_i,b_i\in \ZZ/2\ZZ$, is an element of $\Gamma^{4,4}/2\Gamma^{4,4}$, then we define
\be \eta_\frac{\lambda}{2}=\eta_{\frac{1}{2}\sum_i (a_in_i+b_iw_i)}:=\eta_{\frac{a_1n_1}{2}}\cdots \eta_{\frac{a_4n_4}{2}}\eta_{\frac{b_1w_1}{2}}\cdots \eta_{\frac{b_4w_4}{2}}\ , \qquad \eta'_\frac{\lambda}{2}:=Q\eta_\frac{\lambda}{2}\ .
\ee We have the fusion rules
\be Q^2=I\ ,\qquad  \eta_{\frac{\lambda}{2}} Q=Q\eta_{\frac{\lambda}{2}} \ ,\qquad (\eta_{\frac{\lambda}{2}})^2=Q^{(\lambda,\lambda)/2}
\ee  which imply
\be \eta_{\frac{n_i}{2}}^2=I=\eta_{\frac{w_i}{2}}^2\ ,\qquad \eta_{\frac{\lambda}{2}}\eta_{\frac{\mu}{2}}=Q^{(\lambda,\mu)}\eta_{\frac{\mu}{2}}\eta_{\frac{\lambda}{2}}
\ee Different choices of the basis $\{n_1,\ldots,n_4,w_1,\ldots,w_4\}$ just lead to exchanging $\eta_{\frac{\lambda}{2}}\leftrightarrow \eta'_{\frac{\lambda}{2}}$ for some values of $\theta$ -- this is just a relabeling of the defects. Altogether, these invertible elements have group-like fusion rules, corresponding to an extraspecial group $2^{1+8}$ \cite{Atlas}
\be 1\longrightarrow \langle Q\rangle\cong\ZZ_2 \longrightarrow 2^{1+8}\longrightarrow \ZZ_2^8\longrightarrow 1\ .
\ee

\bigskip

Besides invertible defects, torus orbifolds $\CT/\ZZ_2$ always contain a continuum of defects $T_\theta$, parametrized by $\theta \in ((\RR/ \ZZ)^8)/{\pm 1}$, that preserve the $\CN=(4,4)$ algebra and the spectral flow generators, and that are induced by the $\calR$-invariant superposition $W_\theta+W_{-\theta}$ of topological defects  of the torus model $\CT$. The defects $T_\theta$ have dimension $2$ and satisfy the fusion rules
$$ T_\theta T_{\theta'}=T_{\theta+\theta'}+T_{\theta-\theta'}\ ,
$$ While $W_\theta+W_{-\theta}$ are clearly a superposition, the $T_\theta$ are actually simple for generic values of $\theta$. The only exceptions are when $\theta$ is one of the $\calR$-fixed points $\theta=\frac{\lambda}{2} \in \frac{1}{2}\Gamma^{4,4}/\Gamma^{4,4}$, and in this case they are superpositions
\be T_\frac{\lambda}{2}=\eta_\frac{\lambda}{2}+Q\eta_\frac{\lambda}{2}\ .
\ee 
Notice that $T_0=\CI+Q$, so that
\be\label{Tthetasquare} (T_\theta)^2=T_0+T_{2\theta}=\CI+Q+T_{2\theta}\ ,
\ee that implies that $T_\theta$ is unoriented, $(T_\theta)^\dual=T_{\theta}$. The fusion with the invertible defects is
\be QT_\theta=T_\theta Q=T_\theta\ ,\qquad \eta_{\frac{\lambda}{2}}T_\theta=T_\theta\eta_{\frac{\lambda}{2}}=T_{\theta+\frac{\lambda}{2}}\ ,
\ee where in the last identity, one uses $\frac{\lambda}{2}\equiv -\frac{\lambda}{2} \mod \Gamma^{4,4}$, so that $T_{\theta+\frac{\lambda}{2}}=T_{\theta-\frac{\lambda}{2}}$.

According to \cite{Thorngren:2021yso}, the operators $\hat T_{\theta}$ associated with these defects act on the untwisted sector as the operators $\hat W_\theta+\hat W_{-\theta}$ in the original theory, while they annihilate the twisted sector. On the RR ground states, all $\hat T_\theta$ act by twice the identity on the states in the untwisted sector, and annihilate all the states in the twisted sector.

If $2\theta\in (\Gamma^{4,4}\otimes\RR)/\Gamma^{4,4}$ is a  $\calR$-fixed point, i.e. if $2\theta\equiv -2\theta\mod \Gamma^{4,4}$, then $\theta=\frac{\lambda}{4}$ for some $\lambda\in \Gamma^{4,4}$, and the fusion product \eqref{Tthetasquare} becomes
\be\label{lambdafour} (T_\frac{\lambda}{4})^2=\CI+Q+\eta_{\frac{\lambda}{2}}+Q\eta_{\frac{\lambda}{2}}\ .
\ee The right-hand side is just the sum over all invertible defects in the order $4$ group generated by $Q$ and $\eta_{\frac{\lambda}{2}}$ (this is either $\ZZ_2\times \ZZ_2$ or $\ZZ_4$, depending on whether $(\eta_\frac{\lambda}{2})^2$ equals $\CI$ or $Q$, i.e. if $(\lambda,\lambda)$ equals $0$ or $2$ mod $4$). Therefore, $T_{\frac{\lambda}{4}}$ is a duality defect, providing an equivalence of our theory to the orbifold by this subgroup.

\subsection{Topological defects in $T^4/\ZZ_N$ models}

Let us generalize the construction of the previous sections to the case of a K3 model $\calC$ obtained as the orbifold $\CT/\ZZ_N$ of a torus model $\CT$ by a cyclic group $\langle g\rangle\cong\ZZ_N$. 

In order for $\CT/\ZZ_N$ to be a K3 model, all holomorphic and anti-holomorphic fields of spin $1/2$ must be projected out. Because the $U(1)^4\times U(1)^4$ group generated by $W_\theta$ acts trivially on such fields, this implies that $g$ must be a lift of an automorphism of the Narain lattice $\Gamma^{4,4}$. Furthermore, because the eight fermions transform in the same representation as $\Gamma^{4,4}\otimes \RR$, we require $g$ to have no fixed vectors in such representation. Notice that, for $N>2$, the group $G_\CT$ preserving the small $\CN=(4,4)$  algebra and spectral flow contains such lifts only at special loci in the moduli space of torus models. All such models and symmetries were classified in \cite{Volpato:2014zla}.

The quantum symmetry $Q$ has order $N$ and acts by multiplication by a phase $e^{\frac{2\pi ik }{N}}$ on the $g^k$-twisted sector. The orbifold procedure can be described as the gauging of the finite group $\langle g\rangle$, and defect line $Q^k$ is interpreted as a Wilson line  associated with the $1$-dimensional representation $\rho_k$ of $\langle g\rangle$, where $\rho_k(g)=e^{\frac{2\pi ik}{N}}$. In particular, operators that were local in the original torus model $\CT$ and transforming in the $\rho_k$ representation of $\langle g\rangle$ become operators in the defect space $\CH_{Q^k}$ in the orbifold $\calC$.

For each $\theta\in \Gamma^{4,4}\otimes \RR/\Gamma^{4,4}$, there  is a continuum of topological defects $T_\theta\in \cTop_\calC$ of dimension $N$ induced by a superposition of defects $W_{g^k(\theta)}$ of the torus model $\CT$
\be W_\theta+W_{g(\theta)}+\ldots+W_{g^{N-1}(\theta)}\longrightarrow T_\theta\ .
\ee  The defect $T_\theta$ only depends on the orbit of $\theta$ with respect to the action of $\langle g\rangle$, and its dual is $(T_\theta)^\dual=T_{-\theta}$. For generic $\theta$, when the $\ZZ_N$-orbit has $N$ distinct elements $g^k(\theta)$, $T_\theta$ is simple, while it decomposes
into a sum of $N/d$ simple defects when $\theta$ has a non-trivial stabiliser $\ZZ_{N/d}\subset \ZZ_N$. More precisely, when the stabiliser subgroup $Stab(\theta)$ is  $\langle g^{d}\rangle\cong\ZZ^{N/d} $ for some $d|N$, there are $N/d$ simple defects $\eta_{\theta},Q\eta_\theta,Q^{2}\eta_\theta,\ldots,Q^{N/d-1}\eta_\theta$  of dimension $d$ induced by
\be \sum_{k=0}^{d-1} W_{g^k(\theta)}\longrightarrow \eta_\theta
\ee
and we have a decomposition
\be T_\theta= \eta_\theta+ Q\eta_\theta +\ldots +Q^{N/d-1}\eta_\theta \ .
\ee This decomposition can be understood as follows: when $g^d(\theta)=\theta$, each defect space $\CH_{W_{g^k(\theta)}}$ in the torus model, with $k=0,\ldots,d-1$,  carries a non-trivial representation of $\langle g^d\rangle$, and operators in $\CH_{W_{g^k(\theta)}}$ with different $g^d$-charge must belong to different (non-isomorphic) irreducible modules  of the orbifold theory $\calC$. The OPE with defect operators in $\CH_Q$ modifies the $g^d$-charge, and therefore maps each of these $N/d$ modules into one another. Therefore, the defect lines for these modules can be written as $Q^k\eta_\theta\equiv \eta_\theta Q^k$, $k=0,\ldots,N/d-1$. We notice that operators in $Q^{N/d}$ have trivial $g^d$ charge, so that we have the fusion rule
\be Q^{N/d}\eta_\theta=\eta_\theta Q^{N/d}=\eta_\theta\ .
\ee This fusion rule also implies that the linear operator $\hat\eta_\theta:\CH\to\CH$ associated with $\eta_\theta$ must annihilate all $g^k$-twisted states, unless $k$ is a multiple of $d$. Similarly, the fusion rule
\be QT_\theta=T_\theta Q=T_\theta\ ,
\ee implies that $\hat T_\theta$ annihilates all $g^k$-twisted sectors, for all $k\neq 0\mod N$. 
 
 The fusion rules are
\be T_\theta T_{\theta'}=\sum_{k=0}^{N-1} T_{\theta+g^k(\theta')}\ ,
\ee and, in particular, for generic $\theta$,
\be T_\theta T_{-\theta}=T_0+\sum_{k=1}^{N-1} T_{\theta-g^k(\theta)}=1+Q+\ldots+Q^{N-1}+\sum_{k=1}^{N-1} T_{(1-g^k)(\theta)}\ .
\ee
Finally, because  the reflection \eqref{refl} commutes with every automorphism of the lattice $\Gamma^{4,4}$, we have that $\calR$ is also an invertible topological defect of the orbifold model $\calC$. 

\bigskip
The elements $\theta$ that are stabilised by $g$ are all the solutions $x\in (\Gamma^{4,4}\otimes\RR)/\Gamma^{4,4}$ of the equation
\be (1-g)(x)\equiv 0\mod  \Gamma^{4,4}\ .
\ee Because by hypothesis $g$ has no fixed vectors on $\Gamma^{4,4}\otimes\RR$ then
\be 1+g+\ldots g^{N-1}=0\ ,
\ee and $(1-g)$ is invertible, with inverse
\be (1-g)^{-1}=-\frac{1}{N}(g+2g^2+\ldots+(N-1)g^{N-1})\ ,
\ee as can be easily verified. Thus, the fixed vectors $x$ are all elements of 
\be ((1-g)^{-1}\Gamma^{4,4})/\Gamma^{4,4}\subseteq (\frac{1}{N}\Gamma^{4,4})/\Gamma^{4,4}\ .
\ee The number of distinct points in this quotient is given by $\det(1-g)$ and can be easily computed once the eigenvalues of $g$ are known.

Consider for example the case when $N$ is prime. The analysis in \cite{Volpato:2014zla} shows that the possible values are $N=2,3,5$ (for $N=5$, the symmetry $g$ has no geometric interpretation as an automorphism of the target $T^4$ torus.). In this case, for any $\theta\in (\Gamma^{4,4}\otimes\RR)/\Gamma^{4,4}$, either the orbit of $\theta$ contains $N$ distinct elements,  or $\theta$ is fixed by the full group $\langle g\rangle$. For each $N$, the eigenvalues of $g$ are all primitive $N$-th roots of unity with the same multiplicity \cite{Volpato:2014zla}.  Therefore, the number $\det(1-g)$ of distinct points in the quotient $((1-g)^{-1}\Gamma^{4,4})/\Gamma^{4,4}$ equals $2^8$ for $N=2$, $3^4$ for $N=3$ and $5^2$ for $N=5$.
The corresponding simple defects $\eta_{x}$ are invertible, and together with $Q$ they generate some non-abelian central extension of $\ZZ_2^8$, $\ZZ_3^4$ and $\ZZ_5^2$ called extraspecial groups $2^{1+8}$, $3^{1+4}$ and $5^{1+2}$ for $N=2$, $N=3$ and $N=5$, respectively:
\be 1\longrightarrow \langle Q\rangle\cong \ZZ_N\longrightarrow N^{1+k}\longrightarrow \ZZ_N^k\to 1\ . 
\ee In particular, two lines $\eta_{x}$ and $\eta_{x'}$ do not necessarily commute 
\be\label{notabelian} \eta_{x}\eta_{x'}=Q^{c_g(x,x')}\eta_{x'}\eta_{x}\ .
\ee
The $2$-cocycle $c_g(x,x')$ \eqref{notabelian} characterizing the central extension $N^{1+k}$ are determined by the 't Hooft anomaly in the $U(1)^8$ symmetry of the torus model, see \cite{Bantay:1990yr,Roche:1990hs,Dijkgraaf:1989pz,Tachikawa:2017gyf}.

Similar arguments hold when $N$ is not prime. The possible values of $N$, in this case, are $4,6,8,10,12$ (they are the values $o(\pm g_0)$ in table 2 of \cite{Volpato:2014zla}; the eigenvalues of $g$ when acting on the space $\Gamma^{4,4}\otimes \RR$ are denoted by $\pm\zeta_L,\pm\zeta_R,\pm\zeta_L^{-1},\pm\zeta_R^{-1}$ in the same table). For each $d|N$, consider the element $g^d$. If $1$ is an eigenvalue of $g^d$, then the points $x\in (\Gamma^{4,4}\otimes\RR)/\Gamma^{4,4}$ stabilized by $g^d$ form a continuum, corresponding to the $g^d$-fixed subspace of $\Gamma^{4,4}\otimes \RR$. On the other hand, suppose that none of the eigenvalues of $g^d$ equals $1$, so that $(1-g^d)$ is invertible. In this case, the $g^d$ fixed vectors $x\in \Gamma^{4,4}\otimes\RR/\Gamma^{4,4}$ are the $\det(1-g^d)$ elements in the quotient $((1-g^d)^{-1}\Gamma^{4,4})/\Gamma^{4,4}$. When $d=1$, the corresponding defects $\eta_x$ are invertible, and form a group that is a central extension of some $\ZZ_N^k$ by the quantum symmetry $Q$.

\bigskip

Consider the case where $g$ has  \emph{prime} order $N$, and $(1-g)$ is invertible.
Let $x\in (\Gamma^{4,4}\otimes \RR)/\Gamma^{4,4}$, $x\neq 0$, be a non-trivial fixed point of $g$, i.e. such that $x\equiv g(x)\mod \Gamma^{4,4}$; this is the image $x=(1-g)^{-1}\lambda$ of some vector $\lambda\in\Gamma^{4,4}$. If we apply the operator $(1-g)^{-1}$ once again to $x$, we obtain a vector
\be v:=(1-g)^{-1}x\ ,
\ee with some special properties. In particular, we get
\be (1-g^k)v=(1+g+\ldots+g^{k-1})(1-g)(1-g)^{-1}x=kx\mod \Gamma^{4,4}\ .
\ee Let us define the simple defect
\be \CN_x:=\calR T_{v}=\calR \sum_{k=0}^{N-1} W_{g^k(v)}=\calR\sum_{k=0}^{N-1} W_{v-kx}\ .
\ee We get
\begin{align} \CN_x^2&=\calR T_{v}\calR T_{v}=T_{-v}T_v=\sum_{k=0}^{N-1}T_{v-g^k(v)}=\sum_{k=0}^{N-1}T_{kx}\notag\\&=(1+\eta_x+\ldots+\eta_{(N-1)x})(1+Q+\ldots+Q^{N-1})\ ,
\end{align} and
\be \CN_x\eta_x=\calR T_{v}\eta_x=\calR T_{v+x}=\calR \sum_{k=0}^{N-1} W_{v-kx+x}=\calR \sum_{k=0}^{N-1} W_{v-(k-1)x}=\CN_x
\ee
Therefore, $\CN_x$ is a duality defect for the abelian group of order $N^2$ generated by $\eta_x$ and $Q$. This group could be isomorphic to either $\ZZ_N\times\ZZ_N$ or $\ZZ_{N^2}$, depending on the norm of $x$. This argument shows that the torus orbifold $\CT/\langle g\rangle$ is self-orbifold with respect to any such group of symmetries.

\bigskip

More generally, one expects a continuum of defects in $\cTop_\calC$ whenever $\calC$ can be described as an orbifold of a torus model $\CT$ by some (possibly non-abelian) symmetry group $G$ of $\CT$. Such defects are induced by superpositions $\sum_{g\in G} W_{g(\theta)}$ of topological defects of torus models, and are simple for generic values of $\theta$.

\subsection{An example of torus orbifold $T^4/\ZZ_4$}\label{s:T4Z4}

Let us consider an example of the general theory described in the previous section. We consider a torus model $\CT$ with a large group $U(1)^8\rtimes G^0_\CT$ of symmetries commuting with a small $\CN=(4,4)$ superconformal algebra, where $G^0_\CT$ has order $192$. This is the torus model described in section 4.4.1 of \cite{Volpato:2014zla}, and contains a $\widehat{so}(8)_1$ chiral and anti-chiral algebra. The $\ZZ_2$ orbifold by $\calR$ (the centre of $G^0_\CT$) gives the K3 model with the `largest symmetry group' \cite{Gaberdiel:2013psa,Harvey:2020jvu} that we will consider in section \ref{s:Z28M20}.

The winding-momentum lattice $\Gamma^{4,4}$ of this model is given by the columns of the following matrix
\be\frac{1}{\sqrt{2}}\left(
\begin{array}{cccccccc}
 2 & 0 & 0 & -1 & 1 & 1 & 1 & 1 \\
 0 & 2 & 0 & -1 & 1 & 0 & 0 & 1 \\
 0 & 0 & 2 & -1 & 0 & 1 & 0 & 0 \\
 0 & 0 & 0 & -1 & 0 & 0 & 1 & 0 \\
 0 & 0 & 0 & 0 & 1 & 1 & 1 & 1 \\
 0 & 0 & 0 & 0 & 1 & 0 & 0 & -1 \\
 0 & 0 & 0 & 0 & 0 & 1 & 0 & 0 \\
 0 & 0 & 0 & 0 & 0 & 0 & 1 & 0 \\
\end{array}
\right)
\ee In each column, the first four coordinates are the eigenvalues $\lambda_L^1,\ldots,\lambda_L^4$ of the zero modes of $\partial X^1,\ldots, \partial X^4$, while the last four are the eigenvalues $\lambda_R^1,\ldots,\lambda_R^4$ of the zero modes of $\bar\partial X^1,\ldots, \bar\partial X^4$.

Notice that the vertex operators $V_\lambda(z)$ corresponding to the first four columns are holomorphic currents ($h=\frac{\lambda_L\cdot\lambda_L}{2}=1$, $\bar h=\frac{\lambda_R\cdot\lambda_R}{2}=0$), which together with $\partial X^k$ generate the $\widehat{so}(8)_1$ current algebra.

Let us focus on a symmetry $g\in G^0_\CT$ of order $N=4$, acting by
\be \left(
\begin{array}{cccccccc}
 -1 & 0 & 0 & 0 & 0 & 0 & 0 & 0 \\
 0 & -1 & 0 & 0 & 0 & 0 & 0 & 0 \\
 0 & 0 & -1 & 0 & 0 & 0 & 0 & 0 \\
 0 & 0 & 0 & -1 & 0 & 0 & 0 & 0 \\
 0 & 0 & 0 & 0 & 0 & 1 & 0 & 0 \\
 0 & 0 & 0 & 0 & -1 & 0 & 0 & 0 \\
 0 & 0 & 0 & 0 & 0 & 0 & 0 & 1 \\
 0 & 0 & 0 & 0 & 0 & 0 & -1 & 0 \\
\end{array}
\right)
\ee on $\Gamma^{4,4}\otimes \RR$. One can easily verify that this is a lattice automorphism, $g\in O(\Gamma^{4,4})$, and that $\det(1-g)=64\neq 0$. This is a symmetry in the class $-4A$, in the notation of \cite{Volpato:2014zla}, and the orbifold $\CT/\langle g\rangle$ is again the K3 model with largest symmetry group that we will describe in section \ref{s:Z28M20} (see section $6$ of \cite{Gaberdiel:2013psa}). 

Let us consider the $g$-fixed points in $\Gamma^{4,4}\otimes\RR/\Gamma^{4,4}$. As explained in the previous section they correspond to the points
\be ((1-g)^{-1}\Gamma^{4,4})/\Gamma^{4,4}\cong \ZZ_2\times \ZZ_2\times \ZZ_4\times \ZZ_4
\ee with generators
\be y_1=\begin{pmatrix}
    0 \\ 1 \\ 0 \\ 0 \\ 0 \\ 0 \\ 0 \\ 0
\end{pmatrix}\qquad y_2=\begin{pmatrix}
 -{1}/{2} \\ -{1}/{2} \\ -{1}/{2} \\ -{1}/{2} \\ 0 \\ 0 \\ 0 \\ 0 
\end{pmatrix}\qquad 
u_1=\begin{pmatrix}
 {3}/{2} \\ {3}/{2} \\ 0 \\ 0 \\ 1 \\ 0 \\ 0 \\ 0 
\end{pmatrix}\qquad 
u_2=\begin{pmatrix}
 {3}/{2} \\ 0 \\ 1 \\ {1}/{2} \\ {3}/{2} \\ {1}/{2} \\ {1}/{2} \\
   {1}/{2} 
\end{pmatrix}\ ,
\ee where $y_1$ and $y_2$ have order $2$, while $u_1$ and $u_2$ have order $4$ (modulo $\Gamma^{4,4}$).

The invertible topological defects $W_x$ of the torus model $\CT$, where $x\in ((1-g)^{-1}\Gamma^{4,4})/\Gamma^{4,4}$ are $g$-fixed points, induce invertible defects $\eta_x,Q\eta_x,Q^2\eta_x,Q^3\eta_x$ in the K3 model $\CT/\langle g\rangle$, where $Q$ is the quantum symmetry of order $4$. The group they generate is a central extension of $((1-g)^{-1}\Gamma^{4,4})/\Gamma^{4,4}\cong \ZZ_2\times \ZZ_2\times \ZZ_4\times \ZZ_4$ by $\langle Q\rangle\cong \ZZ_4$. In order to determine the central extension in detail, one needs to know the 't Hooft anomaly for the abelian group
\be H=\langle g\rangle \times ((1-g)^{-1}\Gamma^{4,4})/\Gamma^{4,4}\cong \ZZ_4\times \ZZ_2\times \ZZ_2\times \ZZ_4\times \ZZ_4\ .
\ee The anomaly is encoded in a cohomology class $[\omega]\in H^3(H,U(1))$, with representative a $3$-cocycle $\omega:H\times H\times H\to U(1)$. It is known that 
\be H^3(\ZZ_2^2\times \ZZ_4^3,U(1))\cong \ZZ_2^{18}\times\ZZ_4^{7}\ ,\ee and a basis of generators can be found, for example, in \cite{deWildPropitius:1995cf}. In order to determine which class in $H^3(H,U(1))$ is relevant in this case, one can just consider, for each $k\in H$, the failure of the level matching condition for the $k$-twisted sector of $\CT$. In particular, if $k\in H$ has order $o(k)$, the level matching condition is satisfied if the spin ${spin}(k)$ (i.e. the difference $h-\bar h$ of conformal weights) of the $k$-twisted states take values in $\frac{1}{o(k)}\ZZ$; in this case, the restriction of the anomaly class $[\omega]$ to the cyclic group $\langle k\rangle$ is trivial. More generally, if the restriction of $[\omega]$ to $\langle k\rangle$ is non-trivial, one has the following relation \cite{Bantay:1990yr,Coste:2000tq,Roche:1990hs,Dijkgraaf:1989pz,Frohlich:2009gb}
\be\label{spincocy} e^{2\pi i o(k)spin(k)}=\prod_{i=1}^{o(k)-1}\omega(k,k^i,k)\ ,
\ee between the spin and the cocycle $\omega$. 
If we know the spin of the $k$-twisted sectors for all $k\in H$, we can use \eqref{spincocy} to determine the class $[\omega]$. In fact, for elements $x\in ((1-g)^{-1}\Gamma^{4,4})/\Gamma^{4,4}\subset H$ the spin of the $x$-twisted ground state is
\be {\rm spin}(x)=\frac{(x,x)}{2}\ .
\ee On the other hand, the cohomology class $[\omega]$ is trivial when restricted to the cyclic subgroup $\langle g\rangle\subset H$ (see \cite{Volpato:2014zla}). More generally, the restriction of $[\omega]$ to any cyclic group of the form $\langle gx\rangle\subset H$ is trivial for all $x\in ((1-g)^{-1}\Gamma^{4,4})/\Gamma^{4,4}$, because $g$ and $gx$ are conjugate within the group $G$. These data are sufficient to determine the class $[\omega]$ uniquely.

Once the cocycle $\omega$ representing the 't Hooft anomaly of $H$ is known, one can determine the central extension of $((1-g)^{-1}\Gamma^{4,4})/\Gamma^{4,4}$  by the quantum symmetry $Q$ that is induced on the orbifold theory $\CT^4/\langle g\rangle$. In particular, the commutation relations are
$$
    \eta_{y_1}\eta_{y_2}=Q^2\eta_{y_2}\eta_{y_1}\ , \qquad\qquad \eta_{u_1}\eta_{u_2}=Q\eta_{u_2}\eta_{u_1}\ ,\qquad\qquad
    \eta_{y_j}\eta_{u_k}=Q^2\eta_{u_k}\eta_{y_j}\ ,\qquad j,k=1,2\ .
$$ Furthermore, while $\eta_{y_1}$ and $\eta_{y_2}$ can be chosen of order $2$
\be \eta_{y_1}^2=1=\eta_{y_2}^2\ ,
\ee the symmetries $\eta_{u_1}$ and $\eta_{u_2}$ have order $8$
\be \eta_{u_1}^4=Q^2=\eta_{u_2}^4\ ,\qquad\qquad Q^4=1\ .
\ee

The orbifold theory $\CT/\langle g\rangle$ contains a continuum of topological defects $T_\theta$ of dimension $4$, with $\theta\in (\Gamma^{4,4}\otimes\RR)/\Gamma^{4,4}$, that are simple for generic $\theta$. These defects, as well as the defects of the form $Q^kT_\theta$, $k=0,1,2,3$, are induced by the superpositions $\sum_{j=0}^3 W_{g^j(\theta)}$ in the torus model
\be\label{WWWWT} W_\theta+W_{g(\theta)}+W_{g^2(\theta)}+W_{g^3(\theta)} \longrightarrow T_\theta,\ QT_\theta,\ Q^2T_\theta,\ Q^3T_\theta\ . 
\ee When $\theta\equiv x\in  ((1-g)^{-1}\Gamma^{4,4})/\Gamma^{4,4}$ is one of the $g$-fixed points, $T_\theta\equiv T_x$ is not simple, but becomes a superposition of four invertible defects
\be T_x=\eta_x+Q\eta_x+Q^2\eta_x+Q^3\eta_x\ , 
\ee and satisfies $QT_x=T_x$.

An intermediate case occurs when $\theta$ is fixed by $g^2$, but not by $g$. This happens when $\theta$ is of the form $\theta=(\theta_L,0)$, so that $g(\theta)=-\theta$. In this case, $T_\theta$ decomposes as a sum
\be T_\theta=\xi_{\theta}+Q\xi_{\theta}\ ,
\ee where $\xi_{\theta}$ and $Q\xi_\theta$ are dimension two defects induced by $W_\theta+W_{-\theta}$ on $\CT$
\be\label{WWxi} W_\theta+W_{-\theta} \longrightarrow \xi_\theta,Q\xi_\theta\ .\ee
The defects $\xi_\theta$, where $\theta=(\theta_L,0)$, are simple for generic $\theta_L$, while they decompose as $\xi_x=\eta_x+Q\eta_x$ when $x\equiv (\theta_L,0)$ is fixed by $g$. Furthermore, the $\xi_\theta$ are always unoriented and satisfy
\be Q^2\xi_\theta=\xi_\theta Q^2=\xi_\theta\ ,\qquad\qquad \xi_\theta\xi_{\theta'}=\xi_{\theta+\theta'}+\xi_{\theta-\theta'}\ .
\ee

As described in the previous section, for each $g$-fixed vector $x\in ((1-g)^{-1}\Gamma^{4,4})/\Gamma^{4,4}$, we can consider $v=(1-g)^{-1}x$. There are two cases to be considered, depending on whether $x$ is fixed by $g^2$ or not. In the first case, $x$ is a linear combination
\be x=a_1y_1+a_2y_2+2b_1u_1+2b_2u_2\ ,
\ee with $a_1,a_2,b_1,b_2\in \{0,1\}$. Notice the $x$ has always order $2$ in this case, $2x\in \Gamma^{4,4}$, and $g(x)=-x\mod \Gamma^{4,4}$. In this case, we have
\be v=a_1\begin{pmatrix}
    0 \\ {1}{2} \\ 0 \\ 0 \\ 0 \\ 0 \\ 0 \\ 0\end{pmatrix}+a_2\begin{pmatrix}
       -{1}/{4} \\ -{1}/{4} \\ -{1}/{4} \\ -{1}/{4} \\ 0 \\ 0 \\ 0 \\ 0 
    \end{pmatrix} + b_1\begin{pmatrix}
       {3}/{4} \\ {3}/{4} \\ 0 \\ 0 \\ 0 \\ 0 \\ 0 \\ 0
    \end{pmatrix}+b_2\begin{pmatrix}
       -1/2 \\ 0 \\ 0\\ -1/2  \\ 0 \\ 0 \\ 0 \\ 0 
    \end{pmatrix}
\ee so that $g(v)=-v$ and $g^2(v)=v$. Thus, in the orbifold theory $\CT/\langle g\rangle$ there are two dimension $2$ unoriented simple topological defects, $\xi_v$ and $Q\xi_v$, induced by the defect $W_{v}+W_{-v}$ of $\CT$, as in \eqref{WWxi}.
 The fusion of $\xi_v$ and $Q\xi_v$ with themselves gives
\be \xi_v^2=(Q\xi_v)^2=\xi_0+\xi_x=1+Q^2+\eta_x+Q^2\eta_x\ .
\ee Furthermore, because $v+x\equiv -v\mod \Gamma^{4,4}$, one gets
\be \xi_v\eta_x=\eta_x\xi_v=\xi_v\ .
\ee Therefore, both $\xi_v$ and $Q\xi_v$ are duality defects for the $\ZZ_2\times \ZZ_2$ group generated by $\eta_x$ and $Q^2$.

In the second case, where $x$ is not fixed by $g^2$, we have
\be x=a_1y_1+a_2y_2+c_1u_1+c_2u_2\ ,
\ee where $a_1,a_2,c_1,c_2\in \{0,1\}$, with at least one among $c_1$ and $c_2$ being odd. In this case, $x$ has order $4$, i.e. $4x\in \Gamma^{4,4}$, and $v:=(1-g)^{-1}x$ is not fixed by either $g$ or $g^2$. Thus, the orbifold theory $\CT/\langle g\rangle$ contains four dimension $4$ defects $Q^kT_v$, $k=0,1,2,3$, induced by $\sum_{k=0}^3W_{g^k(v)}$, as in \eqref{WWWWT}. The treatment is analogous to the case described in the previous section, where the order $N$ of $g$ was a prime number. One defines
\be \CN_x:=\calR T_v\ ,
\ee and the resulting topological defect $\CN_x$ is unoriented and satisfies
\be \CN_x^2= (1+Q+Q^2+Q^3)(1+\eta_x+\eta_{2x}+\eta_{3x})
\ee Thus, it is a duality defect for the $\ZZ_2\times \ZZ_8$ group with generators $\eta_x$ of order $8$ and $Q\eta_x^2$ of order $2$.

\subsection{K3 models with continuous defects: a conjecture}\label{s:conjecture}

In the previous sections, we have seen that all K3 models $\calC$ that can be described as torus orbifolds contain a continuous family of topological defects $\CL_\lambda\in \cTop_\calC$, that preserve the $\CN=(4,4)$ algebra and spectral flow, and that are simple for generic values of $\lambda$. 
On the other hand, in claim \ref{th:generic} we argued that for a generic K3 model $\cTop_\calC$ is trivial, and in particular has no such family of defects. Thus, it is natural to look for a characterisation of the K3 models admitting such continuous families.
We propose the following conjecture:

\begin{conjecture}\label{th:onlytori} Let $\calC$ be a K3 model that admits a continuous family of topological defects $\CL_\lambda\in \cTop_\calC$, preserving the $\CN=(4,4)$ superconformal algebra and spectral flow, and simple for all $\lambda$ except possibly a zero measure set. Then, $\calC$ is the (generalised) orbifold of a torus model $\CT$.
\end{conjecture}

Let us recall the notion of generalised orbifold \cite{Brunner:2013xna,Carqueville:2012dk,Frohlich:2009gb}. In a `standard' orbifold of a CFT $\calC$ by a finite group of symmetries $G$, one can consider the superposition $A=\sum_{g\in G} \CL_g$ of all invertible defects in the group. 
Then, the Hilbert space of local operators in the orbifold theory $\calC/G$ is just the subspace of $G$-invariant operators in $\CH_{A}$, and their correlation functions can be obtained as correlation functions in the original theory with an insertion of a fine enough network of defects of type $A$. 
This construction can be generalised to any topological defect $A$, for which there exist some `multiplication' and `co-multiplication' maps, i.e. topological junction operators $\mu:\CH_A\otimes \CH_A\to \CH_A$ and  $\tilde \mu:\CH_A\to\CH_A\otimes \CH_A$ satisfying some suitable conditions -- essentially, that the associator of the multiplication map is trivial, and that $\tilde \mu\circ\mu$ is the identity. 
When such properties are satisfied, one can define a new consistent CFT, the generalised orbifold $\calC/A$, whose space of local operators is a suitable subspace of $\CH_A$, and whose correlation functions are obtained by inserting networks of defects $A$ in the $\calC$ correlators, with the operators $\mu$ and $\tilde\mu$ inserted at trivalent junctions. 
Furthermore, the generalised orbifold procedure is always reversible: if $\calC'=\calC/A$ is obtained as a generalised orbifold of some CFT $\calC$, then $\calC=\calC'/A'$ is a generalised orbifold of $\calC'$. 

\medskip

Let us now sketch an argument for a possible proof of conjecture \ref{th:onlytori}. According to \cite{Thorngren:2021yso}, the presence of a continuum of operators $\CL_\lambda$ is related to the presence of a conserved current $j\in \CH_{\CL_\lambda\CL_\lambda^\dual}$:
\be j:=J(z)dz+\tilde J(\bar z)d\bar z\ ,
\ee where $\bar \partial J(z)=0=\partial \tilde J(\bar z)$. This current is such that an infinitesimal deformation $\CL_{\lambda+\delta\lambda}$ of the defect $\CL_\lambda$ can be obtained as
\be \CL_{\lambda+\delta\lambda}=\CL_\lambda e^{i\delta\lambda\int j}\ ,
\ee where $j\in \CH_{\CL_\lambda\CL_\lambda^\dual}\cong \CH_{\CL_\lambda}\otimes\CH^*_{\CL_\lambda}$ is interpreted as a linear operator from $\CH_{\CL_\lambda}$ to itself, and the integration $\int j$ is over the support of the defect $\CL_\lambda$.

In the case we are interested in, since $\CL_\lambda$ and $\CL_{\lambda+\delta\lambda}$ are defects in $\cTop_\calC$, then the current $j$ must preserve the $\CN=(4,4)$ superconformal algebra and the spectral flow. In particular, $J(z)$ and $\tilde J(\bar z)$ must be neutral with respect to the holomorphic and anti-holomorphic $\widehat{su}(2)_1$ R-symmetry, and must be the supersymmetric descendants of holomorphic and anti-holomorphic spin $1/2$ fields. This implies that the NS-NS sector of $\CH_{\CL_\lambda\CL_\lambda^\dual}$ contains two $\CN=(4,4)$ BPS representation $(\frac{1}{2},\frac{1}{2};0,0)$ and $(0,0;\frac{1}{2},\frac{1}{2})$ containing the currents $J(z)$ and $\tilde J(\bar z)$. As a matter of fact, because the ground states of these representations are a $\widehat{su}(2)_1$ doublet on \emph{complex} spin $1/2$ fields, and because $(\CL_\lambda\CL_\lambda^\dual)^\dual=\CL_\lambda\CL_\lambda^\dual$, there must be $4$ holomorphic and $4$ anti-holomorphic Majorana spin $1/2$ fields (in other words, one needs to include also the CPT conjugates of these BPS representations). 
The OPE of the four holomorphic spin $1/2$ fields is very constrained, and only the vacuum operator can appear in a singular term. 
We can take $\CL_{orb}\subseteq \CL_\lambda\CL_\lambda^\dual$ to be the smallest unoriented ($\CL_{orb}^\dual=\CL_{orb}$) defect containing $I$, with $\CL_{orb}\subset (\CL_{orb})^2$, and such that $\CH_{\CL_{orb}}$ contains the four holomorphic and antiholomorphic spin $1/2$ fields and is closed under OPE. 
Then $\CH_{\CL_{orb}}$ contains a holomorphic and a antiholomorphic copy of the algebra generated by $\CN=4$ and the free fermions. 
This is just the chiral algebra of a generic $T^4$ model, i.e. the algebra of four free bosons and four free fermions. Furthermore, $\CH_{\CL_{orb}}$ decomposes into (possibly twisted) representations of this chiral and anti-chiral algebra. 
What remains to prove is that $A:=\CL_{orb}$ (or possibly some extension of $\CL_{orb}$), satisfies all properties such that the generalised orbifold $\calC/A$ is well-defined. It might be possible to prove this fact using properties of the category of (twisted) representations of the chiral algebra of $T^4$. If the orbifold $\calC/A$ is consistent, then it contains four free fermions and four free bosons and has central charges $c=6=\tilde c$, so that it is necessarily a sigma model $\CT$ on $T^4$. 
By reversibility of the orbifold procedure, we conclude that $\calC$ is a generalised orbifold of the torus model $\CT$.

In the examples of continuous defects discussed in the previous sections, $\CL_{orb}$ was given by $I+Q+\ldots+Q^{N-1}$, and indeed the orbifold of $\calC$ by $\CL_{orb}$ gives back $\CT$.

\section{An example: a K3 model with $\ZZ_2^8:M_{20}$ symmetry group}\label{s:Z28M20}

In this section, we discuss the topological defects $\CL\in \cTop_\calC$ preserving the $\CN=(4,4)$ superconformal algebra (SCA) in a particular K3 sigma model $\calC_{GTVW}$, the `most symmetric model' of \cite{Gaberdiel:2013psa} (see also \cite{Harvey:2020jvu}). The bosonic parts of the holomorphic and anti-holomorphic chiral algebras of this model are isomorphic to $\CA:=(\widehat{su}(2)_1)^{6}$ (six copies of the $su(2)$ current algebra at level $1$). Recall that $\widehat{su}(2)_1$ has two irreducible representations with conformal weights $0$ and $1/4$, and the fusion ring of  $\widehat{su}(2)_1$ is the group ring of the cyclic group $\ZZ_2$ of order two, which we represent as the additive group with elements $0,1$. Each irreducible $\CA$-module is labeled by an element in $\ZZ_2^{6}$, as $M_{a_1,\ldots, a_6}$, $a_i\in\{0,1\}$.

The NS-NS sector of $\calC_{GTVW}$ is given by the sum
\be \bigoplus_{[a_1,\ldots,a_6;b_1,\ldots,b_6]\in A_{NS-NS}} M_{a_1,\ldots, a_6}\otimes \bar M_{b_1,\ldots, b_6}\ ,
\ee where the set $A_{NS-NS}\subset \ZZ_2^6\times \ZZ_2^6$ is given by the disjoint union
\be A_{NS-NS}=A_{(NS-NS)^+}\sqcup A_{(NS-NS)^-}
\ee of a `bosonic' subset (corresponding to the subsector with positive fermion number)
\be A_{(NS-NS)^+}=\{[a_1,\ldots,a_6;b_1,\ldots,b_6]\in \ZZ_2^6\times \ZZ_2^6\mid a_i=b_i,\quad \sum_i a_i\equiv 0\mod 2\}
\ee
and a `fermionic' one
\be A_{(NS-NS)^-}=\{[a_1,\ldots,a_6;b_1,\ldots,b_6]\in \ZZ_2^6\times \ZZ_2^6\mid a_i=b_i+1,\quad \sum_i a_i\equiv 0\mod 2\}\ .
\ee
Similarly, the R-R sector is a sum over the modules whose labels take values in the set $A_{R-R}=(A_{R-R})^+\sqcup (A_{R-R})^-$ with
\be A_{(R-R)^+}=\{[a_1,\ldots,a_6;b_1,\ldots,b_6]\in\ZZ_2^6\times \ZZ_2^6\mid a_i=b_i,\quad \sum_i a_i\equiv 1\mod 2\}
\ee
and
\be A_{(R-R)^-}=\{[a_1,\ldots,a_6;b_1,\ldots,b_6]\in\ZZ_2^6\times \ZZ_2^6\mid a_i=b_i+1,\quad \sum_i a_i\equiv 1\mod 2\}\ .
\ee
The bosonic CFT $\calC^{bos}_{GTVW}$ that contains only the states in the $(NS-NS)^+$ and $(R-R)^+$ sectors of the SCFT $\calC_{GTVW}$ is just the diagonal modular invariant of the $(\widehat{su}(2)_1)^{\oplus 6}$ algebra.

The K3  model $\calC_{GTVW}$ can also be defined as the orbifold $\CT/\ZZ_2$ of a particular torus model $\CT$ by the $\ZZ_2$ symmetry $\calR$ reflecting all torus coordinates, see section \ref{s:torusOrbs}. Alternatively, it can be described as a $\CT/\ZZ_4$ orbifold of the same torus model $\CT$ by a symmetry of order $4$, as described in section \ref{s:T4Z4}. See also \cite{Gaberdiel:2013psa} for more details on this construction.

\bigskip

This model contains several different copies of the $\CN=(4,4)$ superconformal algebra, related to each other by symmetries of the CFT. We focus on one of them, such that the $\widehat{su}(2)_1$ subalgebra of the (anti-)holomorphic $\CN=4$ is identified with the first factor in the (anti-)chiral algebra $(\widehat{su}(2)_1)^{6}$. The four supercurrents of the holomorphic $\CN=(4,4)$ are suitably chosen ground states in the $M_{1,1,1,1,1,1}\otimes \bar{M}_{0,0,0,0,0,0}$ module of $\CA\times\bar\CA$. The space of ground states of this module is isomorphic to a tensor product $(\CC^2)^{\otimes 6}$, and the $\CN=4$ supercurrents can be nicely described in terms of a quantum error correcting code \cite{Harvey:2020jvu}. The full (bosonic and fermionic) chiral algebra of the theory is generated by the bosonic $(\widehat{su}(2)_1)^{6}$ and one of the supercurrents.

The $24$ (real) Ramond-Ramond states with conformal weight $h=\bar h=\frac{1}{4}$ generating the space $V$ are the ground states in the six $\CA\times\bar\CA$ representations
\begin{align}
&[1,0,0,0,0,0;1,0,0,0,0,0]\, ,\ [0,1,0,0,0,0;0,1,0,0,0,0]\, ,\ [0,0,1,0,0,0;0,0,1,0,0,0]\notag\\
&[0,0,0,1,0,0;0,0,0,1,0,0]\, ,\ [0,0,0,0,1,0;0,0,0,0,1,0]\, ,\ [0,0,0,0,0,1;0,0,0,0,0,1]\label{RRreps}
\end{align} each one containing $4$ ground states; we will call any such a set of four states a \emph{tetrad}. In particular, the first tetrad, i.e. the ground states in $[1,0,0,0,0,0;1,0,0,0,0,0]$, are the spectral flow generators, spanning $\Pi\subset V$.

\bigskip

The full group of symmetry of $\calC_{GTVW}$ is $(SU(2)^6\times SU(2)^6)\rtimes S^6$, where the two $SU(2)^6$ subgroups are generated by the zero modes of the holomorphic and anti-holomorphic currents, while $S^6$ is the group of permutations of the six $\widehat{su}(2)_1$ factors in $(\widehat{su}(2)_1)^{6}$ acting diagonally on the holomorphic and anti-holomorphic currents. To be precise, such a group does not act faithfully on $\CH$, because a certain subgroup $Z_0\subset SU(2)^6\times SU(2)^6$ acts trivially on all the states of the theory. Indeed, let \be t_i,\tilde t_i\ , \qquad i=1,\ldots,6,\ee
be the generators of the centre $\ZZ_2^6\times \ZZ_2^6$ of $SU(2)^6\times SU(2)^6$, where $t_i$ and $\tilde t_i$  act on the representation $M_{a_1,\ldots, a_6}\otimes \bar M_{b_1,\ldots, b_6}$ by multiplication by $(-1)^{a_i}$ and $(-1)^{b_i}$, respectively. Then, the subgroup
\be \label{Z0_subgroup_trivial_action} Z_0=\{ \prod_{i=1}^6 (t_i\tilde t_i)^{r_i}\ \mid \sum_{i=1}^6 r_i\in 2\ZZ \}\cong \ZZ_2^5
\ee
acts trivially on all states of the theory. Thus, the group acting faithfully is $((SU(2)^6\times SU(2)^6)/Z_0)\rtimes S^6$. In particular, the quotient $(\ZZ_2^6\times \ZZ_2^6)/Z_0 \cong \ZZ_2^7$ is the group of symmetries commuting with the full bosonic chiral algebra $\CA$, and we can take 
\be\label{gensz27} t_1,\ldots, t_6, \tilde t_1\ ,
\ee as representatives of the generators (modulo $Z_0$).

In \cite{Gaberdiel:2013psa} it was shown that the subgroup $G_{GTVW}\subset ((SU(2)^6\times SU(2)^6)/Z_0)\rtimes S^6$ fixing the $\CN=(4,4)$ superconformal algebra and the spectral flow generators is finite and isomorphic to a split extension $\ZZ_2^8:M_{20}$ of the Mathieu group $M_{20}$ by $\ZZ_2^8$. In turn, $M_{20}$ is a split extension $\ZZ_2^4:A_5$ of the alternating group $A_5$ by $\ZZ_2^4$. See appendix \ref{a:gensZ28M20} for a description of the generators. The intersection
\be\label{centreGGTVW} G_{GTVW}\cap (\ZZ_2^6\times \ZZ_2^6)/Z_0=\{ \prod_{i=2}^6 t_i^{r_i}\ \mid \sum_{i=2}^6 r_i\in 2\ZZ\}\cong\ZZ_2^4\ ,
 \ee is the subgroup of symmetries commuting with both $\CN=(4,4)$ superconformal algebra and with the bosonic current algebra $(\widehat{su}(2)_1)^{6}\times (\widehat{su}(2)_1)^{6}$.

\subsection{Topological defects}\label{s:topdefsGTVW}

Let us now consider the possible topological defects $\CL\in \cTop_{GTVW}\equiv \cTop_{\calC_{GTVW}}$ preserving the $\CN=(4,4)$ algebra and spectral flow. The following result was used in section \ref{s:topdefsK3} to prove that an analogous statement holds for the category of defects $\cTop_\Pi$ of any K3 models $\calC_\Pi$ (see claim \ref{th:onlyid}).

\begin{claim}\label{th:nootherdef} The only defects $\CL\in \cTop_{GTVW}$ that are transparent to all $24$ R-R operators in $V$ are integral multiples of the identity defect.
    \end{claim}
To prove this statement, we notice that if $\CL$ is transparent to the $\CN=(4,4)$ algebra and to all R-R operators in $V$, then it is also transparent to all the operators that can be obtained from their OPE. Thus, we just need to prove that one can obtain \emph{all} operators in $\CH$ by OPE of R-R operators in $V$ and their $\CN=(4,4)$ descendants. In fact, by taking the OPE of the four operators the $i$-th tetrad, $i=1,\ldots, 6$, one can obtain the currents in the holomorphic and anti-holomorphic $i$-th $\widehat{su}(2)_1$ factor. Thus, in this way we can generate the full chiral and antichiral algebra $\CA\times\bar\CA$. Furthermore, the fusion rules between $\CA$ modules imply that the OPE of the $\CA\times\bar\CA$ primary fields in the six modules \eqref{RRreps}, together with the ones in the $[1,1,1,1,1,1;0,0,0,0,0,0]$ and $[0,0,0,0,0,0; 1,1,1,1,1,1]$ modules, where the supercurrents live, generate every other $\CA\times\bar\CA$ primary operator in the spectrum, and we conclude.

\bigskip

Some examples of defects in $\cTop_{GTVW}$ are as follows:
\begin{enumerate}
	\item \emph{Invertible defects.} We have one simple defect $\CL_g$ for each $g\in G_{GTVW}\cong \ZZ_2^8:M_{20}\cong \ZZ_2^8:(\ZZ_2^4:A_5)$.
 \item \emph{Defects induced by the torus orbifold description.} Since the model $\calC_{GTVW}$ can be defined as a torus orbifold $\CT/\ZZ_2$, it inherits all defects from the torus model $\CT$. In particular, as described in section \ref{s:torusOrbs}, $\cTop_{GTVW}$ contains a continuum of non-invertible defects $T_\theta$, $\theta\in (\Gamma^{4,4}\otimes\RR)/\Gamma^{4,4}$ of dimension $2$ that are simple for generic values of $\theta$. One can identify the first two tetrads in \eqref{RRreps} as the RR ground states in the untwisted sector of the orbifold, and the last four tetrads as the twisted sector. Then, all operators $\hat T_\theta$ act on the first two tetrads by multiplication by $2$, while they annihilate the last four tetrads. Notice that symmetries in $G_{GTVW}$ permute the second tetrad with the last four, while keeping the first tetrad of spectral flow generators fixed; therefore, they change the identification of the tetrads with the twisted or untwisted sector of the orbifold. This means that fusion of $T_\theta$ with invertible defects in $\cTop_{GTVW}$ provide further sets of continuous defects of dimension $2$. Any such defect acts by multiplication by $2$ on the spectral flow generators (the first tetrad) and on one more tetrad (the untwisted sector RR ground states), while annihilating the remaining four tetrads (the twisted sector).\\
 The model $\calC_{GTVW}$ can also be described as a $T^4/\ZZ_4$ orbifold. In fact, this is exactly the model described in section \ref{s:T4Z4}. We denote by $Q_4$ the quantum symmetry of order $4$ in this torus orbifold description. Among the elements of the group $G_{GTVW}$, the quantum symmetries $Q_4$ can be characterised by their eigenvalues on the $24$-dimensional space $V$ of RR ground states ($\pm 1$ with multiplicity $4$ each, and $\pm i$ with multiplicity $8$ each).  From the $T^4/\ZZ_4$ description, one can deduce that $\cTop_{GTVW}$ must contain a continuum of topological defects $T^{(4)}_\theta$ of order $4$ and a continuum of topological defects $\xi_\theta$ of order $2$. The former act on the first tetrad of RR ground states by multiplication by $4$, while they annihilate all other $5$ tetrads. The latter act by multiplication by $2$ on two tetrads (including the first), while they annihilate the remaining four tetrads.  
 \item \emph{Verlinde lines for $(\widehat{su}(2)_1)^{\oplus 6}$.} Since this SCFT is rational with respect to the chiral algebra $\CA$, it is useful to consider the topological defects that preserve the whole algebra $\CA$. For rational CFTs, there is a finite number of such defects. In particular, for a bosonic CFT corresponding to the diagonal invariant of $\CA$, simple topological defects preserving $\CA$ are completely classified, and are given by the Verlinde line defects. In general, Verlinde lines are in one-to-one correspondence with representations of the algebra $\CA$, and their fusion ring is the same as the fusion ring of $\CA$ representations. Let us consider the simple topological defects of the bosonic CFT $\calC_{GTVW}^{bos}$ that preserve the $(\widehat{su}(2)_1)^{\oplus 6}$ chiral and anti-chiral algebras. They are $2^6$ Verlinde lines, one for each $(\widehat{su(2)}_1)^{\oplus 6}$ representation, and obey group-like $\ZZ_2^6$ fusion rules; this means that they are all invertible, so that they form a $\ZZ_2^6$ group of symmetries. The lifts of such symmetries from the bosonic CFT $\calC_{GTVW}^{bos}$ to the SCFT $\calC_{GTVW}$, together with the fermion number, generate the full group $(\ZZ_2^6\times \ZZ_2^6)/Z_0\cong \ZZ_2^7$ of symmetries fixing all $(\widehat{su}(2)_1)^{\oplus 6}$. Notice that not all Verlinde defects preserve the $\CN=(4,4)$ superconformal algebra and the spectral flow generators -- $\CA$ only contains the even subalgebra of $\CN=4$, but not the supercurrent. In fact, only  the subgroup \eqref{centreGGTVW}
  of $(\ZZ_2^6\times \ZZ_2^6)/Z_0$ is contained in $G_{GTVW}$ and therefore gives rise to defects in $\cTop_{GTVW}$. In particular, $t_1$, $ \tilde t_1$, and $t_1\cdots t_6$ can be identified, respectively, with the left-moving and right-moving fermion number, and with the $\ZZ_2$ symmetry acting by $-1$ on all R-R states. In conclusion, Verlinde lines only provide topological defects that are generated by ordinary symmetries.
 \item \emph{Lines preserving some `large' chiral algebra.} There are two ways we can generalize the construction of Verlinde lines in the previous point. One is to consider topological defects that act on the $(\widehat{su}(2)_1)^{\oplus 6}$ chiral and anti-chiral algebra by some non-trivial automorphism. The full analysis is performed in appendix \ref{s:automA}. The final outcome is that, besides the invertible defects, there are a few more duality defects (see the next point), related to the fact that the SCFT $\cTop_{GTVW}$ is self-orbifold $\cTop_{GTVW}=\cTop_{GTVW}/H$ with respect to certain subgroups $H$ of the subgroup $G_{GTVW}\cap (\ZZ_2^6\times \ZZ_2^6)/Z_0$ of symmetries fixing all the $(\widehat{su}(2)_1)^{\oplus 6}$ currents. Actually, it turns out that the duality defects obtained in this way are elements in the family of continuous defects $T_\theta$, at special values of $\theta$. The possibility that some $T_\theta$ are duality defects was discussed in section \ref{s:torusOrbs}.\\
 The other possibility to discover new objects in $\cTop_{GTVW}$ is to consider the defects that preserve some chiral algebra $\CB$ that is smaller than the full chiral algebra of the theory, but such that the theory is still rational with respect to $\CB$. 
 There are many simple subalgebras of the bosonic $(\widehat{su}(2)_1)^{\oplus 6}$ whose representation theory are well known, such as, for example, products of affine subalgebras $\hat h\subset (\widehat{su}(2)_1)^{\oplus 6}$ and cosets $(\widehat{su}(2)_1)^{\oplus 6}/h$. The main difficulty with this approach is to find the defects that preserve the $\CN=4$ supercurrents. We were not able to find any new defects in $\cTop_{GTVW}$ using this technique.
	\item \emph{Duality defects.} By definition, a duality defect $N$ is such that the fusion with the reversed orientation defect $N^\dual$ is a superposition of invertible defects
 \be N^\dual N=\sum_{h\in H} \CL_h\ ,\qquad N\CL_h=N\ ,
 \ee for some group $H$ of symmetries. Duality defects occur when a model $\calC$ is self-orbifold with respect to the group $H$, i.e. if the orbifold theory $\calC/H$ is a consistent CFT isomorphic to $\calC$. This means that there is an isomorphism between the space $\CH$ of local operators of $\calC$ and the one of  $\calC/H$ such that correlation functions are the same. In this case, moving a local operator of $\calC$ through $N$ gives the corresponding operator in $\calC/H$.
In particular, when $N =N^\dual$ is unoriented, and $H$ is an abelian group, the category generated by $N$ and $\CL_h$, $h\in H$, is called a Tambara-Yamagami (TY) category, with fusions
\be N^2=\sum_{h\in H} \CL_h\ ,\qquad \CL_h N=N=N\CL_h\ .
\ee
If $N\in \cTop_{GTVW}$ is a duality defect, then the corresponding $H$ must be a subgroup of $G_{GTVW}$, the group of symmetries acting trivially on the $\CN=(4,4)$ superconformal algebra and spectral flow. Notice, however, that the converse is not true in general: even if a certain subgroup $H\subset G_{GTVW}$ is such that $\calC_{GTVW}/H\cong \calC_{GTVW}$, the corresponding duality defect $N$ is not necessarily transparent to the $\CN=(4,4)$ algebra and spectral flow, and in this case it is not in $\cTop_{GTVW}$ (see appendix \ref{a:NnotinTop} for an example). \\
Let us discuss the possible abelian groups $H$ for which $\calC_{GTVW}/H$ is isomorphic to $\calC_{GTVW}$. One necessary condition is that the orbifold model contains bosonic chiral and antichiral algebras isomorphic to $(\widehat{su}(2)_1)^{\oplus 6}$. We can separate the subgrops $H\subset G_{GTVW}$ into two classes, depending on whether they act trivially on the chiral and anti-chiral $\CA\times \bar \CA$ or not.\\
In the first case, $H$ must be a subgroup of \eqref{centreGGTVW}, i.e. the intersection of $G_{GTVW}$ and the centre $(\ZZ_2^6\times\ZZ_2^6)/Z_0$. Furthermore, $\frac{\hat{N}}{\sqrt{|H|}}$ must act by an automorphism on $(\widehat{su}(2)_1)^{\oplus 6}$. All such defects $N\in \cTop$ are discussed in appendix \ref{s:automA}. The results are as follows. We find three kind of duality defects (up to conjugation in $G_{GTVW}$): $\CN_{ijk}$ for all $2\le i<j<k\le 6$, with $\CN_{ijk}^2=\CI+\CL_{t_it_j}+\CL_{t_jt_k}+\CL_{t_it_k}$; $\CN_{ij,kl}$ for all pairwise distinct $i,j,k,l\in\{2,\ldots,6\}$, with $\CN_{ij,kl}^2=\CI+\CL_{t_it_j}+\CL_{t_kt_l}+\CL_{t_it_jt_kt_l}$; and $\CN_{23456}$ where $\CN_{23456}^2$ equals to the superposition of the $16$ invertible defects in $G_{GTVW}\cap (\ZZ_2^6\times\ZZ_2^6)/Z_0\cong \ZZ_2^4$.  \\
Let us now consider the case where $H$ acts non-trivially on the chiral algebra. In order for $\calC/H$ to be isomorphic to $\calC$, there must be holomorophic currents in some twisted sector $\CH_h$, for some $h\in H$. 
 
By a direct calculation of the $h$-twisted partition function, we found that the group $G_{GTVW}$ admits only three conjugacy classes of elements $h$ such that $\Hh_h$ contains holomorphic currents. One class is given by symmmetries of the form $t_it_jt_kt_l$, $2\le i<j<k<l\le 6$, that are contained in the group \eqref{centreGGTVW}. Each of them can be interpreted as the quantum symmetry in a description of the model as a torus orbifold $T^4/\ZZ_2$. This means that the orbifold of the CFT $\calC_{GTVW}$ by $\ZZ_2$ group generated by any such symmetries is a torus model, and in particular cannot be isomorphic to $\calC_{GTVW}$ itself. However, there could be larger abelian groups containing symmetries of the form $t_it_jt_kt_l$, under which $\calC$ is self-orbifold. In fact, we already found some of these abelian groups among the ones acting trivially on currents.\\ 
The second class are given by symmetries $Q_4$ of order four, such as, for example
\be (\one \, \mone \, \mi \, \ii \, \ii\, \ii\, ;\  \one\,\one\,\one\,\one\,\one\,\one)(1)^6
\ee in the notation of appendix \ref{a:gensZ28M20}; the third class is related to the second by exchanging the action on the holomorphic and anti-holomorphic sectors. The elements $Q_4$ in either the second or the third conjugacy class are quantum symmetries in a description of the model $\calC_{GTVW}$ as a $\ZZ_4$ torus orbifold $T^4/\ZZ_4$.  
As in the previous case, this means that the orbifold of $\calC$ by $\langle Q_4\rangle \cong \ZZ_4$ is a torus model (more precisely, the model described in section \ref{s:T4Z4}), and therefore is not isomorphic to $\calC_{GTVW}$. On the other hand, from the results in section \ref{s:T4Z4}, we know that there are duality defects $\CN_x\in \cTop_{GTVW}$ of order $4$ for groups of  the form $\ZZ_2\times \ZZ_8$ generated by $Q_4$ and by some other symmetry $\eta_x$ of order $8$. In principle, there could be even larger abelian groups with respect to which $\calC$ is self-orbifold, but we did not attempt a full classification. 
\end{enumerate}

\subsection{D-branes and RR charges}\label{s:GTVWbranes}

Let us describe the lattice $\Gamma^{4,20}$ of RR charges of the K3 model $\calC_{GTVW}$, expressed in the orthonormal basis $\{|1,i\rangle,|2,i\rangle,|3,i\rangle,|4,i\rangle\}_{i=1,\ldots,6}$ of the space $V$ of R-R ground fields that is described in appendix \ref{a:gensZ28M20}. In particular, for each fixed $i=1,\ldots,6$, the subset $\{|1,i\rangle,|2,i\rangle,|3,i\rangle,|4,i\rangle\}$ corresponds to a `tetrad' of states, belonging to one of the six representations of $\CA\times\bar\CA$ listed in eq. \eqref{RRreps}. The first tetrad ($i=1$) is given  by the spectral flow generators and spans the subspace $\Pi\subset V$, while $\Pi^\perp$ is spanned by the remaining tetrads with $i=2,\ldots,6$. 

In order to find $\Gamma^{4,20}$ the most direct way would be to find an explicit description of a suitable set of $24$ boundary states. However, we will consider a simpler method, that uses the group of symmetries $G_{GTVW}$. The outcome of this analysis is that the only even unimodular lattice $\Gamma\subset V$ with signature $(4,20)$ that is invariant under the action of the $G_{GTVW}$ on $V$, as described in appendix \ref{a:gensZ28M20}, is the one spanned by the columns of the following matrix:
\be \arraycolsep=3pt \def\arraystretch{0.8}
\frac{1}{\sqrt{8}}
\left(
\begin{array}{cccccccccccccccccccccccc}
     8 & 4 & 4 & 4 & 4 & 4 & 4 & 2 & 4 & 4 & 4 & 2 & 4 & 2 & 2 & 2 & 4 & 2 & 2 & 2 & 2 & 0 &
   0 & 1 \\
 0 & 4 & 0 & 0 & 0 & 0 & 0 & 2 & 0 & 0 & 0 & 2 & 0 & 2 & 0 & 0 & 0 & 0 & 0 & 2 & 0 & 0 &
   0 & 1 \\
 0 & 0 & 4 & 0 & 0 & 0 & 0 & 2 & 0 & 0 & 0 & 2 & 0 & 0 & 2 & 0 & 0 & 2 & 0 & 0 & 0 & 0 &
   0 & 1 \\
 0 & 0 & 0 & 4 & 0 & 0 & 0 & 2 & 0 & 0 & 0 & 2 & 0 & 0 & 0 & 2 & 0 & 0 & 2 & 0 & 0 & 0 &
   0 & 1 \\
 0 & 0 & 0 & 0 & 4 & 0 & 0 & 2 & 0 & 0 & 0 & 0 & 0 & 2 & 2 & 2 & 0 & 2 & 2 & 2 & 2 & 0 &
   0 & 1 \\
 0 & 0 & 0 & 0 & 0 & 4 & 0 & 2 & 0 & 0 & 0 & 0 & 0 & 2 & 0 & 0 & 0 & 0 & 2 & 0 & 0 & 0 &
   0 & 1 \\
 0 & 0 & 0 & 0 & 0 & 0 & 4 & 2 & 0 & 0 & 0 & 0 & 0 & 0 & 2 & 0 & 0 & 0 & 0 & 2 & 0 & 0 &
   0 & 1 \\
 0 & 0 & 0 & 0 & 0 & 0 & 0 & 2 & 0 & 0 & 0 & 0 & 0 & 0 & 0 & 2 & 0 & 2 & 0 & 0 & 0 & 0 &
   0 & 1 \\
 0 & 0 & 0 & 0 & 0 & 0 & 0 & 0 & 4 & 0 & 0 & 2 & 0 & 2 & 2 & 2 & 0 & 2 & 2 & 2 & 2 & 2 &
   2 & 1 \\
 0 & 0 & 0 & 0 & 0 & 0 & 0 & 0 & 0 & 4 & 0 & 2 & 0 & 2 & 0 & 0 & 0 & 2 & 0 & 0 & 0 & 2 &
   0 & 1 \\
 0 & 0 & 0 & 0 & 0 & 0 & 0 & 0 & 0 & 0 & 4 & 2 & 0 & 0 & 2 & 0 & 0 & 0 & 2 & 0 & 0 & 0 &
   2 & 1 \\
 0 & 0 & 0 & 0 & 0 & 0 & 0 & 0 & 0 & 0 & 0 & 2 & 0 & 0 & 0 & 2 & 0 & 0 & 0 & 2 & 0 & 0 &
   0 & 1 \\
 0 & 0 & 0 & 0 & 0 & 0 & 0 & 0 & 0 & 0 & 0 & 0 & 4 & 2 & 2 & 2 & 0 & 0 & 0 & 0 & 2 & 2 &
   2 & 1 \\
 0 & 0 & 0 & 0 & 0 & 0 & 0 & 0 & 0 & 0 & 0 & 0 & 0 & 2 & 0 & 0 & 0 & 0 & 0 & 0 & 0 & 2 &
   0 & 1 \\
 0 & 0 & 0 & 0 & 0 & 0 & 0 & 0 & 0 & 0 & 0 & 0 & 0 & 0 & 2 & 0 & 0 & 0 & 0 & 0 & 0 & 0 &
   2 & 1 \\
 0 & 0 & 0 & 0 & 0 & 0 & 0 & 0 & 0 & 0 & 0 & 0 & 0 & 0 & 0 & 2 & 0 & 0 & 0 & 0 & 0 & 0 &
   0 & 1 \\
 0 & 0 & 0 & 0 & 0 & 0 & 0 & 0 & 0 & 0 & 0 & 0 & 0 & 0 & 0 & 0 & 4 & 2 & 2 & 2 & 2 & 2 &
   2 & 1 \\
 0 & 0 & 0 & 0 & 0 & 0 & 0 & 0 & 0 & 0 & 0 & 0 & 0 & 0 & 0 & 0 & 0 & 2 & 0 & 0 & 0 & 2 &
   0 & 1 \\
 0 & 0 & 0 & 0 & 0 & 0 & 0 & 0 & 0 & 0 & 0 & 0 & 0 & 0 & 0 & 0 & 0 & 0 & 2 & 0 & 0 & 0 &
   2 & 1 \\
 0 & 0 & 0 & 0 & 0 & 0 & 0 & 0 & 0 & 0 & 0 & 0 & 0 & 0 & 0 & 0 & 0 & 0 & 0 & 2 & 0 & 0 &
   0 & 1 \\
 0 & 0 & 0 & 0 & 0 & 0 & 0 & 0 & 0 & 0 & 0 & 0 & 0 & 0 & 0 & 0 & 0 & 0 & 0 & 0 & 2 & 2 &
   2 & 1 \\
 0 & 0 & 0 & 0 & 0 & 0 & 0 & 0 & 0 & 0 & 0 & 0 & 0 & 0 & 0 & 0 & 0 & 0 & 0 & 0 & 0 & 2 &
   0 & 1 \\
 0 & 0 & 0 & 0 & 0 & 0 & 0 & 0 & 0 & 0 & 0 & 0 & 0 & 0 & 0 & 0 & 0 & 0 & 0 & 0 & 0 & 0 &
   2 & 1 \\
 0 & 0 & 0 & 0 & 0 & 0 & 0 & 0 & 0 & 0 & 0 & 0 & 0 & 0 & 0 & 0 & 0 & 0 & 0 & 0 & 0 & 0 &
   0 & 1 
   \end{array}
   \right)
\ee
This is a slight modification (necessary because of the different signature) of the basis of vectors of the Leech lattice in \cite{ConwaySloane}.
It is easy to verify by a direct calculation that the lattice generated by these vectors is invariant under the group $G_{GTVW}$, and that it is even and unimodular with respect to the diagonal metric $\eta={\rm diag}(1,1,1,1,-1,\ldots,-1)$ of signature $(4,20)$. 

The proof that the lattice with these properties is unique (up to $O(4)\times O(20)$ transformations that do not affect the splitting $V=\Pi\oplus\Pi^\perp$) is as follows.
It is known from \cite{Gaberdiel:2011fg} that the group of symmetries $G_{\calC}$ of any K3 model $\calC$ is isomorphic to a subgroup $G_\Lambda$ of the Conway group $O(\Lambda)\cong Co_0$, the group of automorphisms of the Leech lattice $\Lambda$. Furthermore, if $\Gamma^G\subset \Gamma^{4,20}$ and $\Lambda^G\subset \Lambda$ denote the sublattices of $G$-fixed vectors and $\Gamma_G:=\Gamma^{4,20}\cap (\Gamma^G)^\perp$ and $\Lambda_G:=\Lambda\cap (\Lambda^G)^\perp$ their orthogonal complement, then there is an isomorphism $\Gamma_G\cong \Lambda_G$ (reversing the sign of the quadratic form) that is compatible with the action of $G_\calC\cong G_\Lambda$. 
The relevant subgroups of $Co_0$ were classified by H\"ohn and Mason in \cite{HohnMason2016}, and $G_{GTVW}$ appears as group $99$ in their table. In particular, up to conjugation in $Co_0$, \cite{HohnMason2016} shows that there is a unique sublattice $\Lambda_G\cong \Gamma_G$ for this group; this means that $\Gamma\cap \Pi^\perp\cong \Gamma_G$ is unique up to $O(20)$ transformations of $\Pi^\perp$.
As described, for example, in \cite{nikulin} (see also \cite{HohnMason2016}), the genus of the primitive sublattice $\Gamma_G\subset \Gamma^{4,20}$ determines the genus of its orthogonal complement $\Gamma^G$ in the even unimodular lattice $\Gamma^{4,20}$. 
We checked, with some computer aid, that such a genus contains a unique isomorphism class of lattices; this means that $\Gamma\cap \Pi\cong \Gamma^G$ is also uniquely determined up to $O(4)$ transformations. Therefore, for any choice of even unimodular lattice $\Gamma\subset V$ invariant under $G_{GTVW}$, there is a primitive embedding of $\Gamma^G\oplus \Gamma_G$ into $\Gamma$. Different choice of the lattice $\Gamma$ would lead to primitive embedddings that are not related by either $O(\Gamma^G)\times O(\Gamma_G)$ or $O(\Gamma)$ automorphisms (the latter can thought of as changes of basis in a given lattice $\Gamma)$. The possible equivalence classes of such primitive embeddings, up to $O(\Gamma^G)\times O(\Gamma_G)$ and $O(\Gamma)$ transformations, are described by proposition 2.1 of \cite{HohnMason2016}, which is a reformulation of propositions 1.4.1 and 1.6.1 of \cite{nikulin}. In particular, the fact that for group $99$ in \cite{HohnMason2016} the index $\bar{i}_G$ equals $1$, implies that there is a unique class of such embeddings, and this concludes the proof. We refer to \cite{HohnMason2016, nikulin} for more information about embeddings of primitive sublattices into even unimodular lattices, and in particular for the meaning of the index $\bar{i}_G$.

\bigskip

Notice that the first four vectors in the lattice basis 
\be \frac{8}{\sqrt{8}}|1,1\rangle\ ,\qquad \frac{4}{\sqrt{8}}(|1,1\rangle+|2,1\rangle)\ ,\qquad \frac{4}{\sqrt{8}}(|1,1\rangle+|3,1\rangle)\ ,\qquad \frac{4}{\sqrt{8}}(|1,1\rangle+|4,1\rangle)\ ,
\ee are contained in the subspace $\Pi\subset V$ of spectral flow operators. By Claim \eqref{th:qdim}, this implies that all topological defects $\CL\in \cTop_{GTVW}$ in this model have integral quantum dimension. 

As described in section \ref{s:topdefsK3}, with each $\CL\in \cTop_{GTVW}$ is associated a lattice endomorphism $\nL:\Gamma^{4,20}\to\Gamma^{4,20}$, which is contained in the intersection $B^{4,20}_\Pi(\ZZ)={\rm End}(\Gamma^{4,20})\cap B^{4,20}(\RR)$ of ${\rm End}(\Gamma^{4,20})$ with the $401$-dimensional real space $B^{4,20}(\RR)$ of block diagonal matrices defined in eq.\eqref{blkspace}.

Using some computer aid \cite{GAP4}, we found that for the model $\calC_{GTVW}$ the $\ZZ$-module $B^{4,20}_\Pi(\ZZ)$ has maximal rank, and computed a set of $401$ generators $\{\nL_i\}_{i=1,\ldots,401}$, so that any such map $\nL$ can be written as
\be \nL=\sum_i k_i \nL_i\ ,\qquad k_i\in \ZZ\ .
\ee We stress that not all elements $\nL\in B^{4,20}_\Pi(\ZZ)$ are expected to correspond to actual topological defects in $\cTop_{GTVW}$. The $\ZZ$-module $B_{4,20}(\ZZ)$ contains a submodule,
\be B^{4,20}_{\Pi,inv}(\ZZ)\subset B^{4,20}_\Pi(\ZZ)
\ee
generated by all maps $\nL_g$ induced by invertible defects $\CL_g$, $g\in G_{GTVW}$. Notice that, since the $24$-dimensional representation of $G_{GTVW}$ is faithful, $B^{4,20}_{\Pi,inv}(\ZZ)$ is isomorphic to the group ring $\ZZ[G_{GTVW}]$. We found that $B^{4,20}_{\Pi,inv}(\ZZ)$ has also maximal rank, has index $2^{168}$ in $B^{4,20}_{\Pi}(\ZZ)$, and that for all $\nL\in B^{4,20}_{\Pi}(\ZZ)$ one has $4\cdot \nL\in  B^{4,20}_{\Pi,inv}(\ZZ)$. In particular,
\be B^{4,20}_{\Pi}(\ZZ)/B^{4,20}_{\Pi,inv}(\ZZ)\cong \ZZ_4^{8}\oplus \ZZ_2^{152}\ .
\ee

In fact, all the examples of topological defects $\CL\in \cTop_{GTVW}$ described in this section, including the simple non-invertible ones, correspond to elements $\nL$ in $B^{4,20}_{\Pi,inv}(\ZZ)$ (recall that the map $\CL\to \nL$ is not injective).  Unfortunately, we were not able to determine which elements (if any) $\nL\in B^{4,20}_{\Pi}(\ZZ)$, $\nL\notin B^{4,20}_{\Pi,inv}(\ZZ)$, are actually induced by topological defects in $\cTop_{GTVW}$. 

It is interesting to consider the action on RR charges of the duality defects $\CN_{ijk}$, described in appendix \ref{s:automA}. For example, the map $\mathsf{N}_{256}$ related to the defect $\CN_{256}$ is
\be\arraycolsep=3pt \def\arraystretch{0.8}\mathsf{N}_{256}=
\left(
\begin{array}{cccccccccccccccccccccccc}
 2 & 0 & 0 & 0 & 0 & 0 & 0 & 0 & 0 & 0 & 0 & 0 & 0 & 0 & 0 & 0 & 0 & 0 & 0 & 0 & 0 & 0 & 0 & 0 \\
 0 & 2 & 0 & 0 & 0 & 0 & 0 & 0 & 0 & 0 & 0 & 0 & 0 & 0 & 0 & 0 & 0 & 0 & 0 & 0 & 0 & 0 & 0 & 0 \\
 0 & 0 & 2 & 0 & 0 & 0 & 0 & 0 & 0 & 0 & 0 & 0 & 0 & 0 & 0 & 0 & 0 & 0 & 0 & 0 & 0 & 0 & 0 & 0 \\
 0 & 0 & 0 & 2 & 0 & 0 & 0 & 0 & 0 & 0 & 0 & 0 & 0 & 0 & 0 & 0 & 0 & 0 & 0 & 0 & 0 & 0 & 0 & 0 \\
 0 & 0 & 0 & 0 & 0 & 0 & 0 & 0 & 0 & 0 & 0 & 0 & 0 & 0 & 0 & 0 & 0 & 0 & 0 & 0 & 0 & 0 & 0 & 0 \\
 0 & 0 & 0 & 0 & 0 & 0 & 0 & 0 & 0 & 0 & 0 & 0 & 0 & 0 & 0 & 0 & 0 & 0 & 0 & 0 & 0 & 0 & 0 & 0 \\
 0 & 0 & 0 & 0 & 0 & 0 & 0 & 0 & 0 & 0 & 0 & 0 & 0 & 0 & 0 & 0 & 0 & 0 & 0 & 0 & 0 & 0 & 0 & 0 \\
 0 & 0 & 0 & 0 & 0 & 0 & 0 & 0 & 0 & 0 & 0 & 0 & 0 & 0 & 0 & 0 & 0 & 0 & 0 & 0 & 0 & 0 & 0 & 0 \\
 0 & 0 & 0 & 0 & 0 & 0 & 0 & 0 & 0 & 0 & 0 & 0 & 1 & -1 & -1 & -1 & 0 & 0 & 0 & 0 & 0 & 0 & 0 & 0
   \\
 0 & 0 & 0 & 0 & 0 & 0 & 0 & 0 & 0 & 0 & 0 & 0 & 1 & 1 & -1 & 1 & 0 & 0 & 0 & 0 & 0 & 0 & 0 & 0 \\
 0 & 0 & 0 & 0 & 0 & 0 & 0 & 0 & 0 & 0 & 0 & 0 & 1 & 1 & 1 & -1 & 0 & 0 & 0 & 0 & 0 & 0 & 0 & 0 \\
 0 & 0 & 0 & 0 & 0 & 0 & 0 & 0 & 0 & 0 & 0 & 0 & 1 & -1 & 1 & 1 & 0 & 0 & 0 & 0 & 0 & 0 & 0 & 0 \\
 0 & 0 & 0 & 0 & 0 & 0 & 0 & 0 & 1 & 1 & 1 & 1 & 0 & 0 & 0 & 0 & 0 & 0 & 0 & 0 & 0 & 0 & 0 & 0 \\
 0 & 0 & 0 & 0 & 0 & 0 & 0 & 0 & -1 & 1 & 1 & -1 & 0 & 0 & 0 & 0 & 0 & 0 & 0 & 0 & 0 & 0 & 0 & 0 \\
 0 & 0 & 0 & 0 & 0 & 0 & 0 & 0 & -1 & -1 & 1 & 1 & 0 & 0 & 0 & 0 & 0 & 0 & 0 & 0 & 0 & 0 & 0 & 0 \\
 0 & 0 & 0 & 0 & 0 & 0 & 0 & 0 & -1 & 1 & -1 & 1 & 0 & 0 & 0 & 0 & 0 & 0 & 0 & 0 & 0 & 0 & 0 & 0 \\
 0 & 0 & 0 & 0 & 0 & 0 & 0 & 0 & 0 & 0 & 0 & 0 & 0 & 0 & 0 & 0 & 0 & 0 & 0 & 0 & 0 & 0 & 0 & 0 \\
 0 & 0 & 0 & 0 & 0 & 0 & 0 & 0 & 0 & 0 & 0 & 0 & 0 & 0 & 0 & 0 & 0 & 0 & 0 & 0 & 0 & 0 & 0 & 0 \\
 0 & 0 & 0 & 0 & 0 & 0 & 0 & 0 & 0 & 0 & 0 & 0 & 0 & 0 & 0 & 0 & 0 & 0 & 0 & 0 & 0 & 0 & 0 & 0 \\
 0 & 0 & 0 & 0 & 0 & 0 & 0 & 0 & 0 & 0 & 0 & 0 & 0 & 0 & 0 & 0 & 0 & 0 & 0 & 0 & 0 & 0 & 0 & 0 \\
 0 & 0 & 0 & 0 & 0 & 0 & 0 & 0 & 0 & 0 & 0 & 0 & 0 & 0 & 0 & 0 & 0 & 0 & 0 & 0 & 0 & 0 & 0 & 0 \\
 0 & 0 & 0 & 0 & 0 & 0 & 0 & 0 & 0 & 0 & 0 & 0 & 0 & 0 & 0 & 0 & 0 & 0 & 0 & 0 & 0 & 0 & 0 & 0 \\
 0 & 0 & 0 & 0 & 0 & 0 & 0 & 0 & 0 & 0 & 0 & 0 & 0 & 0 & 0 & 0 & 0 & 0 & 0 & 0 & 0 & 0 & 0 & 0 \\
 0 & 0 & 0 & 0 & 0 & 0 & 0 & 0 & 0 & 0 & 0 & 0 & 0 & 0 & 0 & 0 & 0 & 0 & 0 & 0 & 0 & 0 & 0 & 0 \\
\end{array}
\right)
\ee While $\mathsf{N}_{256}\in B^{4,20}_{\Pi,inv}(\ZZ)$, it cannot be written as a sum of the form $\nL_g+\nL_h$,  for some $g,h\in G_{GTVW}$: one always needs an integral linear combination of (at least) four $\nL_g$, $g\in G_{GTVW}$ with (at least) one \emph{negative} coefficient. This can be proved, for example, by noticing that each $\nL_g$ acts by a signed permutation in this basis, so that each row or column in $\nL_g$ has always a single non-zero entry, equal to $\pm 1$. Thus, any sum $\nL_g+\nL_h$ has at most two non-zero entries in each row or column, so it cannot be equal to $\mathsf{N}_{256}$. This example shows that the image of $\cTop_{GTVW}$ in ${\rm End}(\Gamma^{4,20})$ by  $\CL\to \nL$ is strictly larger than the image of the subcategory generated only by invertible defects.

Finally, let us show an example of a deformation of this model lifting all topological defects in $\cTop_{GTVW}$. Let us relabel $\{|\psi_1\rangle,\ldots,|\psi_{24}\rangle\}$ the ordered basis of R-R ground states, where $|\psi_{4k-3}\rangle,\ldots,|\psi_{4k}\rangle$ are the states in the $k$-th tetrad 
$$ |\psi_{4k-3}\rangle=|1,k\rangle,\qquad |\psi_{4k-2}\rangle=|2,k\rangle,\qquad |\psi_{4k-1}\rangle=|3,k\rangle,\qquad |\psi_{4k}\rangle=|4,k\rangle.
$$ We denote by $\psi_1,\ldots,\psi_{24}$ the corresponding operators.
The exactly marginal operators are the supersymmetric descendants of NS-NS operators of weights $h=\bar h=\frac{1}{2}$, obtained from OPE of one of the spectral flow operators $\psi_1,\ldots,\psi_4$ and one of the other $20$ RR ground fields $\psi_5,\ldots,\psi_{24}$. Thus, we can label a basis of such NS-NS operators by $\chi_{i,j}$, $i=1,\ldots,4$, $j=5,\ldots,24$. Each infinitesimal deformation is generated by a linear combination $\chi$ of the $\chi_{i,j}$. A defect $\CL$ is preserved by the deformation generated by $\chi$ if and only if $\chi$ is transparent with respect to the defect $\CL$. This is equivalent to saying that the operator $\hat\CL$ acts on the state $|\chi\rangle$ in the same way as on the vacuum, i.e.
\be\label{Linvar} \hat\CL|\chi\rangle=\langle\CL\rangle |\chi\rangle\ .
\ee Let us consider the following linear combination
\be\label{breakdef} |\chi\rangle=|\chi_{1,5}\rangle+\pi |\chi_{1,6}\rangle+\pi^2|\chi_{1,7}\rangle+\ldots+\pi^{19}|\chi_{1,24}\rangle=\sum_{j=5}^{24}\pi^{j-5}|\chi_{1,j}\rangle\ ,
\ee where $\pi=3.14\ldots$. Because $\CL$ is transparent to the spectral flow operators $\psi_1,\ldots,\psi_4$, the action of $\hat\CL$ on $\chi_{i,j}$ is the same as on $\psi_j$
$$ \hat\CL|\chi_{i,j}\rangle =\sum_{k=5}^{24} \nL_{jk}|\chi_{i,k}\rangle\ .
$$ Thus, for the deformation \eqref{breakdef},  the condition \eqref{Linvar} reads
$$ \sum_{k,j=5}^{24}\pi^{j-5}\nL_{jk}|\chi_{1,k}\rangle=\langle \CL\rangle\sum_{k=5}^{24}\pi^{k-5} |\chi_{1,k}\rangle,
$$ which is equivalent to
\be \sum_{j=5}^{24}\pi^{j-k}\frac{\nL_{jk}}{\langle \CL\rangle}=1,\qquad\qquad \forall k=5,\ldots,24\ .
\ee It is now sufficient to observe that, in this basis, the matrix $\frac{1}{\langle\CL\rangle}\nL_{jk}$ has rational entries, so that the only way this relation can be satisfied is
\be \frac{\nL_{jk}}{\langle \CL\rangle}=\begin{cases}
    1 & \text{for }j=k\\
    0 & \text{for }j\neq k\ .
\end{cases}\ee Thus, the only topological defects that are preserved by the deformation $\chi$ are the ones that are proportional to the identity, in agreement with the general statements in section \ref{s:genK3models}. The same result would hold for any linear combination $|\chi\rangle=\sum_{j=5}^{24}\alpha_j|\chi_{i,j}\rangle$ such that the ratios $\alpha_j/\alpha_k$  are irrational for all $j\neq k$.

\section{Another example: the Gepner model $(1)^6$}\label{s:onetosix}

The second example of K3 model we will analyse is the Gepner model $(1)^6$, that was considered in \cite{Gaberdiel:2011fg}. In general, Gepner models are obtained by considering products of $\CN=2$ minimal models, each one having central charge $c_k=\frac{3k}{k+2}$ with $k\in \NN$, and then taking an orbifold that projects on the states carrying integral total $U(1)$ R-charge. When the total central charge is $c=6$ and the spectrum is invariant under spectral flow, then the $\CN=(2,2)$ superconformal algebra is enhanced to $\CN=(4,4)$ and the CFT is always a non-linear sigma model on K3.

In particular, as the name suggests, the $(1)^6$ model $\calC_{(1)^6}$ is obtained by taking six copies of the $k=1$ minimal model with $c_{k=1}=1$. We refer to \cite{Gaberdiel:2011fg} for a description of this CFT as a Gepner model.  
Here, we notice that this particular Gepner model admits different descriptions that might be more useful to study symmetries, defects and D-branes.

\subsection{The $(1)^6$ Gepner model as a free scalar CFT}

Let us describe the K3 model $\calC_{(1)^6}$ as a rational CFT. Both chiral and anti-chiral superalgebras $\CA$ and $\bar\CA$ contain\footnote{Notice that this is not  the full chiral algebra of this CFT.} the product of six copies of the $\CN=2$ superconformal algebra at $c=1$. It is known (see appendix \ref{a:Neq2freeboson}) that the bosonic subalgebra of $\CN=2$ at $c=1$ can be completely described in terms of a chiral free boson on a circle of suitable radius. This means that the K3 model $\calC_{(1)^6}$ is generated by $6$ holomorphic and six antiholomorphic chiral free bosons
$$ i\partial X^k(z)\ ,\qquad\qquad i\bar \partial X^k(z)\ ,\qquad\qquad k=1,\ldots,6\ ,
$$ together with the vertex operators $V_\lambda(z,\bar z)\sim :e^{i(\vec{\lambda}_L\cdot \vec{X}_L(z)+\vec{\lambda}_R\cdot \vec{X}_R(\bar z))}:$ of conformal weights $(h_L,h_R)=(\frac{\vec\lambda_L^2}{2},\frac{\vec\lambda_L^2}{2})$, where $\lambda\equiv (\vec{\lambda_L},\vec{\lambda}_R)$ takes values in a suitable integral (odd) lattice $\Upsilon\in \RR^{6,6}$. Let $\{e_1,\ldots,e_6,\tilde e_1,\ldots,\tilde e_6\}$ denote an orthonormal basis of $\RR^{6,6}$, so that $\lambda_L=\sum_i^6\lambda_ie_i$ and $\lambda_R=\sum_k^6 \tilde \lambda_k\tilde e_k$. The lattice $\Upsilon$ is given by the union $\Upsilon=\Upsilon^{NS-NS}\cup \Upsilon^{R-R}$ of a NS-NS and a R-R component.

The chiral and anti-chiral algebras $\CA$ and $\bar\CA$ are generated by the currents $i\partial X^k$ and $i\bar\partial X^k$, and by purely holomorphic (respectively, anti-holomorphic) vertex operators $V_{(\vec \lambda_L,0)}(z)$ (respectively, $V_{(0,\vec \lambda_R)}(\bar z)$)  with $\vec \lambda_L$ takes values in a suitable six dimensional positive definite lattice $\Lambda_\CA\subset\RR^6$, such that  
\be  (\Upsilon^{NS-NS}\cap \RR^{6,0})\cong \Lambda_\CA\ ,\qquad (\Upsilon^{NS-NS}\cap \RR^{0,6})\cong \Lambda_\CA(-1) \ ,\ee where $\Lambda_\CA(-1)$ denotes the lattices with opposite quadratic form. In particular, the $(\CN=2)^6\subset \CA$ algebra can be described as the lattice super vertex operator algebra (SVOA) associated with the odd lattice $(\sqrt{3}\ZZ)^6\equiv (\sqrt{3}\ZZ)\oplus \ldots \oplus (\sqrt{3}\ZZ)$ (six copies), so that 
$$ (\sqrt{3}\ZZ)^6\subset \Lambda_\CA\ .
$$
The $(\CN=2)^6$ algebra is not the full chiral algebra $\CA$ of the model $\calC_{(1)^6}$: one needs to extend this SVOA by two additional holomorphic currents $V_{\vec\pm (\lambda_L,0)}(z)$, corresponding to the  lattice vectors 
$$ \pm \vec\lambda_L=\pm (\frac{1}{\sqrt{3}},\frac{1}{\sqrt{3}},\frac{1}{\sqrt{3}},\frac{1}{\sqrt{3}},\frac{1}{\sqrt{3}},\frac{1}{\sqrt{3}})\subset \RR^6\ .$$
The full chiral and antichiral algebras $\CA$ and $\bar\CA$  of the theory are therefore both isomorphic to the lattice SVOA associated 
with the integral odd lattice
\be \Lambda_{\CA} =\{ \frac{n}{\sqrt{3}}(1,1,1,1,1,1)+\sqrt{3}(x_1,\ldots,x_6)\mid n,x_1,\ldots,x_6\in \ZZ\}\subset \RR^6\ .
\ee The bosonic chiral algebra is the lattice VOA associated with the even lattice
\be \Lambda^{bos}_{\CA} =\{ \frac{n}{\sqrt{3}}(1,1,1,1,1,1)+\sqrt{3}(x_1,\ldots,x_6)\mid n,x_1,\ldots,x_6\in \ZZ,\ \sum_ix_i\in 2\ZZ\}\subset \Lambda_{\CA}\ .
\ee
Therefore, the NS-NS lattice $\Upsilon^{NS-NS}$ contains $\Lambda\oplus \Lambda(-1)$ as a sublattice, corresponding to the chiral and anti-chiral algebra $\CA\times \bar\CA$.

 The representations of the algebra $\CA\times \bar\CA$ are labeled by cosets $[v]:v+(\Lambda_{\CA}\oplus \Lambda_{\CA}(-1))$, for suitable $v\in \RR^{6,6}$ so that $\Upsilon^{NS-NS}$ and $\Upsilon^{R-R}$ decompose as unions of cosets
\be
\Upsilon^{NS-NS}=\bigcup_{[v]\in A_{(NS-NS)}}  v+(\Lambda_{\CA}\oplus \Lambda_{\CA}(-1))
\ee and 
\be
\Upsilon^{R-R}=\bigcup_{[v]\in A_{(R-R)}}  v+(\Lambda_{\CA}\oplus \Lambda_{\CA}(-))
\ee
In particular, the NS-NS sector contains the representations
$$ A_{(NS-NS)} =\Bigl\{\bigl[\tfrac{1}{\sqrt{3}}(a_1,\ldots,a_6;a_1,\ldots,a_6)\bigr]\ \mid\ a_i\in \ZZ/3\ZZ,\ ,\ \sum_{i=1}^6 a_i\in 3\ZZ\Bigr\}
$$ while the R-R sector includes
$$ A_{(R-R)} =\Bigl\{\bigl[\tfrac{1}{\sqrt{3}}(a_1,\ldots,a_6;a_1,\ldots,a_6)\bigr]\ \mid\ a_i\in \frac{1}{2}+\ZZ/3\ZZ,\ ,\ \sum_{i=1}^6 a_i\in 3\ZZ\Bigr\}\ .
$$ Notice that we have the identifications
$$ [v]\sim [v+\frac{1}{\sqrt{3}}(1,1,1,1,1,1;1,1,1,1,1,1)]
$$ and the fusion rules
$$ [v]\times [v']=[v+v']
$$

In this description, the $\CN=4$ superconformal algebra at $c=6$ is given as follows. The $\hat{su}(2)_1$ R-symmetry algebra is generated by
\be\label{Rsymmu16} J^3(z)=\frac{1}{2\sqrt{3}}\sum_{k=1}^6 i\partial X^k(z)\ ,\qquad J^\pm(z)=
V_{(\pm\frac{1}{\sqrt{3}},\ldots,\pm\frac{1}{\sqrt{3}};0,\ldots,0)}(z) 
\ee while the supercurrents are given by
\be G^\pm(z)=\sqrt{\frac{2}{3}}\sum_{k=1}^6 V_{\pm\sqrt{3}e_k}(z)\ ,\qquad \tilde G^\pm(z)=\sqrt{\frac{2}{3}}\sum_{k=1}^6 V_{\mp\sqrt{3}e_k\pm (\frac{1}{\sqrt{3}},\ldots,\frac{1}{\sqrt{3}};0,\ldots,0)}(z)
\ee
The $24$ R-R ground states are given by $V_\lambda$ with the following vectors $\lambda\in \Upsilon^{R-R}$:
\begin{align*}
    &\pm\frac{1}{2\sqrt{3}}(1,\ldots,1;1,\ldots,1)\qquad \pm\frac{1}{2\sqrt{3}}(1,\ldots,1;-1,\ldots,-1)\qquad 4\text{ spectral flow generators}\\
    &\frac{1}{2\sqrt{3}}(1,1,1,-1,-1,-1;1,1,1,-1,-1,-1)\text{ and permutations}\qquad 20\text{ operators}\ .
\end{align*}

The symmetries of the model preserving the $\CN=(4,4)$ algebra form a group \be G_{(1)^6}\cong \ZZ_3^4\rtimes A_6\ ,\ee i.e. the extension by $\ZZ_3^4$ of the alternating group $A_6$ of even permutations of six objects \cite{Gaberdiel:2011fg}. Here, the normal subgroup $\ZZ_3^4$ is the group of CFT symmetries acting trivially on the chiral and anti-chiral algebras $\CA\times \bar\CA$. It is generated by elements $t_1^{m_1}\cdots t_6^{m_6}$ acting by
\be\label{Z34} \prod_{k=1}^6t_k^{m_k}(V_\lambda) = 
e^{\frac{2\pi i}{3} \sum_{k=1}^6m_ka_k}V_\lambda\ ,\qquad \lambda\in \frac{1}{\sqrt{3}}(a_1,\ldots,a_6;a_1,\ldots,a_6)+(\Lambda_\CA\oplus\Lambda_{\CA}(-1))\ ,
\ee where $m_1,\ldots,m_6\in \ZZ/3\ZZ$ satisfy the condition $\sum_{k=1}^6 m_k\equiv 0\mod 3$. Notice that, because of the condition $\sum_{k=1}^6 a_k\equiv 0\mod 3$, we have that the action is trivial whenever $m_1=\ldots=m_6$
\be
t_1t_2t_3t_4t_5t_6=t_1^2t_2^2t_3^2t_4^2t_5^2t_6^2=1\ .
\ee 

The alternating group $A_6$ acts by even permutations $\sigma$ on the left- and right-moving $\hat{u}(1)$ currents, and transforms the vertex operators accordingly:
$$ i\partial X^k\mapsto i\partial X^{\sigma(k)}\ ,\qquad i\bar\partial X^k\mapsto i\bar\partial X^{\sigma(k)}\ ,\qquad V_{\sum_k(\lambda_ke_k+\tilde\lambda_k\tilde e_k)}\mapsto \pm V_{\sum_k(\lambda_ke_{\sigma(k)}+\tilde\lambda_k\tilde e_{\sigma(k)})}\ ,
$$ where the sign in front of each $V_\lambda$ must be chosen in such a way that the OPE between vertex operators is preserved.

The full group of symmetries of the CFT is $[(SU(2)\times U(1)^5)\times (SU(2)\times U(1)^5)]\rtimes S_6$, where the two $SU(2)\times U(1)^5$ factors are generated by the zero modes of the holomorphic and anti-holomorphic currents. The permutation group $S_6$ acts simultaneously on each of the $U(1)^5$ factors as in the $5$-dimensional irreducible representation.

\subsection{The $(1)^6$ model as a torus orbifold $T^4/\ZZ_3$}\label{s:OnetosixIsT4z3}

The $\calC_{(1)^6}$ Gepner model can also be described in terms of a torus orbifold $T^4/\ZZ_3$. In fact, there are actually many different ways to obtain this model from an orbifold. 
One simple way to prove that $\calC_{(1)^6}$ is a torus orbifold is to notice that the symmetry group $\ZZ_3^4:A_6$ contains some symmetry $Q$ of order $3$ (e.g. the element of $\ZZ_3^4$ with $m_1=m_2=m_3=1$ and $m_4=m_5=m_6=0$) whose trace on the $24$ dimensional representation of RR ground states $V$ is $-3$. The orbifold of $\calC_{(1)^6}$ by $Q$ gives a sigma model $\CT:=\calC_{(1)^6}/\langle Q\rangle$ on $T^4$, as can be checked by verifying that the elliptic genus of $\calC_{(1)^6}/\langle Q\rangle$ is $0$ \cite{Gaberdiel:2012um}. By reversibility of the orbifold procedure, because $\CT=\calC_{(1)^6}/\langle Q\rangle$ with $\langle Q\rangle\cong\ZZ_3$, we must conclude that the K3 model $\calC_{(1)^6}$ is an orbifold $\calC_{(1)^6}=\CT/\langle g\rangle$ of the $T^4$ model $\CT$ by a cyclic group $\langle g\rangle\cong \ZZ_3$, and $Q$ is the `quantum symmetry' acting trivially on the untwisted sector and multiplying by $e^{\frac{2\pi i k}{3}}$ the $g^k$-twisted sector. 

Knowing the action of $Q$ on the states of $\calC_{(1)^6}$ allows us to identify the untwisted sector and the $g^k$ twisted sector in the bosonic description of the previous section:
$$ \text{untwisted sector: } x_1+x_2+x_3\equiv 0\mod 3\ ,$$
$$ g^k\text{twisted sector: } x_1+x_2+x_3\equiv k\mod 3\ .$$

It is easy to construct the sigma model $\CT$ on $T^4$ corresponding to the orbifold $\CT=\calC_{(1)^6}/\langle Q\rangle$.  We know that the symmetry $Q$ can be written as
\be V_\lambda \mapsto e^{2\pi i(\delta_L\cdot\lambda_L-\delta_R\cdot\lambda_R)} V_\lambda=V_\lambda
\ee where
\be \delta =\frac{1}{2\sqrt{3}}(4,4,4,0,0,0;0,0,0,0,0,0)
\ee is a vector such that $3\delta\in \Gamma$ and $\delta_L^2-\delta_R^2\in 2\ZZ$. The untwisted sector of $\calC_{(1)^6}/\langle Q\rangle$ contains the currents $i\partial X^i(z),i\bar\partial X^i(\bar z)$, as well as the vertex operators $V_\lambda$, $\lambda\in \Upsilon^{Q}$, where
\be \Upsilon^Q=\{\lambda\in\Upsilon\mid \delta_L\cdot\lambda_L-\delta_R\cdot\lambda_R\in \ZZ\}\ .
\ee The full orbifold $\calC_{(1)^6}/\langle Q\rangle$ is generated by the $i\partial X^i(z),i\bar\partial X^i(\bar z)$ currents together with the vertex operators $V_\lambda$ with $\lambda$ taking values in the extended lattice
\be \Upsilon'=\Upsilon^Q\cup(\delta+\Upsilon^Q)\cup (2\delta+\Upsilon^Q)\ .
\ee Let us consider the holomorphic fields of this orbifold. As expected for a supersymmetric sigma model on $T^4$, there are four holomorphic fields of weight $1/2$, namely 
$$ \chi_1^{\pm}(z)\sim V_{\pm (\frac{1}{\sqrt{3}},\frac{1}{\sqrt{3}},\frac{1}{\sqrt{3}},0,0,0;0,\ldots,0)}\qquad\qquad \chi_2^{\pm}(z)\sim V_{\pm (0,0,0,\frac{1}{\sqrt{3}},\frac{1}{\sqrt{3}},\frac{1}{\sqrt{3}};0,\ldots,0)}\ .
$$ By taking normal ordered products of pairs of these four fermions, one obtains the $6$ currents of a `fermionic' $so(4)_1=su(2)_1\oplus su(2)_1$ algebra, where one of the $su(2)_1$ is the R-symmetry of the $\CN=4$ superconformal algebra with currents \eqref{Rsymmu16}, and the second commuting $su(2)_1$ algebra is generated by
$$ i\sum_{k=1}^3(\partial X^k-\partial X^{k+3})\ ,\qquad V_{\pm (\frac{1}{\sqrt{3}},\frac{1}{\sqrt{3}},\frac{1}{\sqrt{3}},-\frac{1}{\sqrt{3}},-\frac{1}{\sqrt{3}},-\frac{1}{\sqrt{3}};0,\ldots,0)}(z)\ .
$$ Besides the $so(4)_1$ currents that can be obtained from the OPE of the spin $1/2$ fields, the orbifold $\calC_{(1)^6}/\langle Q\rangle$ contains $16$ additional spin $1$ fields, that can be arranged into two commuting algebras
$$ \frac{i}{\sqrt{2}}(\partial X^1-\partial X^2)\ ,\qquad \frac{i}{\sqrt{2}}(\partial X^2-\partial X^3)
$$
$$ V_{\pm (\frac{-2}{\sqrt{3}},\frac{1}{\sqrt{3}},\frac{1}{\sqrt{3}},0,0,0;0,\ldots,0)}(z)\,\quad V_{\pm (\frac{1}{\sqrt{3}},\frac{-2}{\sqrt{3}},\frac{1}{\sqrt{3}},0,0,0;0,\ldots,0)}(z)\ ,\quad
V_{\pm (\frac{1}{\sqrt{3}},\frac{1}{\sqrt{3}},\frac{-2}{\sqrt{3}},0,0,0;0,\ldots,0)}(z)
$$
and
$$ \frac{i}{\sqrt{2}}(\partial X^4-\partial X^5)\ ,\qquad \frac{i}{\sqrt{2}}(\partial X^5-\partial X^6)
$$
$$ V_{\pm (0,0,0,\frac{-2}{\sqrt{3}},\frac{1}{\sqrt{3}},\frac{1}{\sqrt{3}};0,\ldots,0)}(z)\,\quad V_{\pm (0,0,0,\frac{1}{\sqrt{3}},\frac{-2}{\sqrt{3}},\frac{1}{\sqrt{3}};0,\ldots,0)}(z)\ ,\quad
V_{\pm (0,0,0,\frac{1}{\sqrt{3}},\frac{1}{\sqrt{3}},\frac{-2}{\sqrt{3}};0,\ldots,0)}(z)\ .
$$ These two commuting algebras have dimension $8$ and rank $2$, and this is enough to conclude that they are both isomorphic to $su(3)_1$. Thus, the orbifold model $\calC_{(1)^6}/\langle Q\rangle$ has both a chiral and anti-chiral algebras isomorphic to $(su(3)_1)^2$, besides the $so(4)_1$ from the free fermions. For this reason, we denote this torus model as $\CT_{A_2^2}$, so that
\be \calC_{(1)^6}/\langle Q\rangle=\CT_{A_2^2}\ .
\ee
The $(su(3)_1)^2$ currents include the superconformal descendants $i\partial Z_i^\pm$ of the free fermions $\chi^\pm_i$. The scalar fields $Z_1^+,Z_2^+$ can be identified with the complex coordinates on the target torus $T^4$, with $Z_1^-,Z_2^-$ their complex conjugate. Notice that the currents $i\partial Z_i^\pm$ are linear combinations of the $i\partial X^k$ and of suitable vertex operators $V_\lambda$; this means that one cannot give the $X^k$ any geometric interpretation as the coordinates on the target torus.

As a CFT, any supersymmetric non-linear sigma model $\CT_{A_2^2}$ on $T^4$ can be described as the product $\CT_{A_2^2}=F(4)\otimes \bar{F}(4)\times \CT^{bos}_{A_2^2}$ of a theory $F(4)\otimes \bar{F}(4)$ of four chiral and four antichiral free fermions with central charges $(2,2)$, times the bosonic sigma model $\CT^{bos}_{A_2^2}$ on the same torus $T^4$, with central charges $(4,4)$. Let us focus on the bosonic sigma model factor, whose fields are characterised by having non-singular OPE with all the chiral and anti-chiral free fermions (i.e. they commute with all the free fermion modes). This is an honest bosonic CFT, with modular invariant partition function.  It contains the chiral algebra $(su(3)_1)^2$ -- the free fermions have zero charge with respect to this algebra, so their OPE with the currents is non-singular. More generally, every $V_\lambda$ with $\lambda$ orthogonal to the vectors
$$ \pm (\frac{1}{\sqrt{3}},\frac{1}{\sqrt{3}},\frac{1}{\sqrt{3}},0,0,0;0,\ldots,0), \qquad \pm (0,0,0,\frac{1}{\sqrt{3}},\frac{1}{\sqrt{3}},\frac{1}{\sqrt{3}};0,\ldots,0)
$$ and their anti-holomorphic counterparts, belong to the bosonic sigma model.

As we show below, the bosonic sigma model $\CT^{bos}_{A_2^2}$ on $T^4$ is just the product $\CT^{bos}_{A_2^2}=\CT^{bos}_{A_2}\otimes \CT^{bos}_{A_2}$ of two isomorphic bosonic sigma models $\CT^{bos}_{A_2}$ on $T^2$, each one corresponding to the diagonal modular invariant for the algebra $su(3)_1$.
Recall that $su(3)_1$ admits three different modules, which we denote by $M_{\bf 1}$ (the vacuum), $M_{\bf 3}$ and $M_{\bar{\bf 3}}$; the latter modules are charge conjugate of each other, and their highest weight vector has weight $2/3$. Therefore, the spectrum of the diagonal modular invariant is
$$ (M_{\bf 1}\otimes \overline{M}_{\bf 1})\oplus (M_{\bf 3}\otimes \overline{M}_{\bf 3})\oplus (M_{\bar{\bf 3}}\otimes \overline{M}_{\bar{\bf 3}})\ .
$$ More precisely, the first $\CT^{bos}_{A_2}$ is generated by the first (holomorphic and anti-holomorphic) $su(3)_1$ factor, by the $9$ non-holomorphic vertex operators of conformal weights $(2/3,2/3)$
$$ V_{(\frac{1}{\sqrt{3}},\frac{-1}{\sqrt{3}},0,0,0,0;\frac{1}{\sqrt{3}},\frac{-1}{\sqrt{3}},0,0,0,0)}\quad V_{(0,\frac{1}{\sqrt{3}},\frac{-1}{\sqrt{3}},0,0,0;0,\frac{1}{\sqrt{3}},\frac{-1}{\sqrt{3}},0,0,0)}\quad V_{(\frac{-1}{\sqrt{3}},0,\frac{1}{\sqrt{3}},0,0,0;\frac{-1}{\sqrt{3}},0,\frac{1}{\sqrt{3}},0,0,0)}
$$ 
$$ V_{(\frac{1}{\sqrt{3}},\frac{-1}{\sqrt{3}},0,0,0,0;\frac{-1}{\sqrt{3}},0,\frac{1}{\sqrt{3}},0,0,0)}\quad V_{(0,\frac{1}{\sqrt{3}},\frac{-1}{\sqrt{3}},0,0,0;\frac{1}{\sqrt{3}},\frac{-1}{\sqrt{3}},0,0,0,0)}\quad V_{(\frac{-1}{\sqrt{3}},0,\frac{1}{\sqrt{3}},0,0,0;0,\frac{1}{\sqrt{3}},\frac{-1}{\sqrt{3}},0,0,0)}
$$ 
$$ V_{(\frac{1}{\sqrt{3}},\frac{-1}{\sqrt{3}},0,0,0,0;0,\frac{1}{\sqrt{3}},\frac{-1}{\sqrt{3}},0,0,0)}\quad V_{(0,\frac{1}{\sqrt{3}},\frac{-1}{\sqrt{3}},0,0,0;\frac{-1}{\sqrt{3}},0,\frac{1}{\sqrt{3}},0,0,0)}\quad V_{(\frac{-1}{\sqrt{3}},0,\frac{1}{\sqrt{3}},0,0,0;\frac{1}{\sqrt{3}},\frac{-1}{\sqrt{3}},0,0,0,0)}
$$ that are the ground states of the $M_{\bf 3}\otimes \overline{M}_{\bf 3}$ modules, and by the $9$ vertex operators $V_\lambda$ with the opposite signs for $\lambda$, namely
$$ V_{-(\frac{1}{\sqrt{3}},\frac{-1}{\sqrt{3}},0,0,0,0;\frac{1}{\sqrt{3}},\frac{-1}{\sqrt{3}},0,0,0,0)}\quad  V_{-(0,\frac{1}{\sqrt{3}},\frac{-1}{\sqrt{3}},0,0,0;0,\frac{1}{\sqrt{3}},\frac{-1}{\sqrt{3}},0,0,0)}\quad V_{-(\frac{-1}{\sqrt{3}},0,\frac{1}{\sqrt{3}},0,0,0;\frac{-1}{\sqrt{3}},0,\frac{1}{\sqrt{3}},0,0,0)}
$$ 
$$ V_{-(\frac{1}{\sqrt{3}},\frac{-1}{\sqrt{3}},0,0,0,0;\frac{-1}{\sqrt{3}},0,\frac{1}{\sqrt{3}},0,0,0)}\quad V_{-(0,\frac{1}{\sqrt{3}},\frac{-1}{\sqrt{3}},0,0,0;\frac{1}{\sqrt{3}},\frac{-1}{\sqrt{3}},0,0,0,0)}\quad V_{-(\frac{-1}{\sqrt{3}},0,\frac{1}{\sqrt{3}},0,0,0;0,\frac{1}{\sqrt{3}},\frac{-1}{\sqrt{3}},0,0,0)}
$$ 
$$ V_{-(\frac{1}{\sqrt{3}},\frac{-1}{\sqrt{3}},0,0,0,0;0,\frac{1}{\sqrt{3}},\frac{-1}{\sqrt{3}},0,0,0)}\quad V_{-(0,\frac{1}{\sqrt{3}},\frac{-1}{\sqrt{3}},0,0,0;\frac{-1}{\sqrt{3}},0,\frac{1}{\sqrt{3}},0,0,0)}\quad V_{-(\frac{-1}{\sqrt{3}},0,\frac{1}{\sqrt{3}},0,0,0;\frac{1}{\sqrt{3}},\frac{-1}{\sqrt{3}},0,0,0,0)}
$$ that are the ground states of the charge conjugate module $M_{\bar{\bf 3}}\otimes \overline{M}_{\bar{\bf 3}}$. The second $\CT^{bos}_{A_2}$ torus model is generated in a similar way, by exchanging the first free set of bosonic coordinates $X^1,X^2,X^3$ with the second set $X^4,X^5,X^6$.

The analysis in this section is sufficient to identify $\CT_{A_2^2}$ with the supersymmetric $T^4$ sigma model denoted by $A_2^2$ in section 4.4.4 of \cite{Volpato:2014zla}. Geometrically, the target space is the product $T^2\times T^2$ of two tori $T^2=\CC/(\ZZ+e^{\frac{2\pi i}{3}}\ZZ)$ with a suitable B-field.

The K3 model $\calC_{(1)^6}$ can be obtained by the orbifold of $\CT_{A_2^2}$ by a symmetry $g$ of order $3$. More precisely, $g$ belongs the class denoted by 3A in \cite{Volpato:2014zla}, and corresponds to a geometric rotation obtained by multiplying the complex coordinates $(z_1,z_2)$ parametrizing the first and second torus $T^2=\CC/(\ZZ+e^{\frac{2\pi i}{3}}\ZZ)$ by $(z_1,z_2)\mapsto (e^{\frac{2\pi i}{3}}z_1,e^{-\frac{2\pi i}{3}}z_2)$. Indeed, only for this class of symmetries the orbifold $\CT_{A_2^2}/\langle g\rangle$ contains $6$ R-R ground states in the untwisted sector and $9+9$ in the two twisted sectors, as expected for the model we are considering. The $g$-invariant untwisted sector contains a subalgebra $\widehat{su}(2)_1\oplus \widehat{u}(1)$ from the `fermionic' $\widehat{so}(4)_1$ algebra of $\CT_{A_2^2}$, as well as subalgebras $\widehat{u}(1)^2\subset \widehat{su}(3)_1$ for each of the factors in the bosonic $\widehat{su}(3)_1$. Altogether, we get $\widehat{su}(2)_1\oplus \widehat{u}(1)^5$, which is the current algebra of $\calC_{(1)^6}=\CT_{A_2^2}/\langle g\rangle$ -- no further currents come from the twisted sectors.

The full symmetry group of the $\CT_{A_2^2}$ torus model preserving the small $\CN=(4,4)$ algebra and the spectral flow generators is $U(1)^8\rtimes G_0$ where $G_0$ is a group of order $36$ acting by automorphisms on the winding-momentum lattice. 
The subgroup commuting with $g$ is $\langle g\rangle\times (\ZZ_3^4\rtimes \ZZ_6)$, where $\ZZ_3^4\subset U(1)^8$ and $\ZZ_6$ is generated $\calR:(z_1,z_2)\mapsto (-z_1,-z_2)$ and by a non-geometric symmetry of order $3$ that rotates only the holomorphic currents $(\partial Z_1,\partial Z_2)\mapsto (e^{\frac{2\pi i}{3}}\partial Z_1,e^{-\frac{2\pi i}{3}}\partial Z_2)$ while keeping the anti-holomorphic currents $(\bar\partial Z_1,\bar\partial Z_2)$ fixed. These torus symmetries induce a group of symmetries of the orbifold $\calC_{(1)^6}=\CT_{A_2^2}/\langle g\rangle$ isomorphic to $3^{1+4}\rtimes \ZZ_6$, obtained from $\langle g\rangle\times (\ZZ_3^4\rtimes \ZZ_6)$ by first quotienting out $\langle g\rangle$, and then taking a non-trivial central extension by the quantum symmetry $\langle Q\rangle\cong\ZZ_3$, as described in section \ref{s:torusOrbs}.

\subsection{Symmetries and topological defects}

After describing in some detail the model $\calC_{(1)^6}$, we are now ready to discuss the topological defects in $\cTop_{(1)^6}$ that preserve the $\CN=(4,4)$ algebra and the spectral flow. Such a category includes:
\begin{enumerate}
    \item \emph{Invertible defects.} There is one defect $\CL_g\in \cTop_{(1)^6}$ for each $g\in G_{(1)^6}\cong \ZZ_3^4\rtimes A_6$. 
    \item \emph{Defects acting by automorphisms on the chiral algebra.} The Verlinde lines preserving the full chiral and anti-chiral algebras $\CA\times \bar\CA$ must be in one-to-one correspondence with the representations of $\CA\times \bar\CA$ and with the same fusion rules. This implies that the Verlinde lines are all invertible, and generate the subgroup $\ZZ_2^4\subset G_{(1)^6}$.\\
    It is more interesting to consider topological lines acting on $\CA\times \bar\CA$ by algebra automorphisms that fix the $\CN=(4,4)$ subalgebra, and do \emph{not} lift to CFT symmetries. The analysis is very similar to the one in appendix \ref{s:automA}, so we just summarize the main points. Recall that the full group of symmetries of the CFT is $((SU(2)\times U(1)^5)^2)\rtimes S_6$). Up to conjugation by such CFT symmetries, we can focus on defects that act on the chiral algebra $\CA$ by an \emph{outer} automorphism $\rho_L$, while acting trivially on $\bar\CA$. 
    Outer automorphisms preserving the $\CN=4$ subalgebra correspond to \emph{even} permutations in $A_6$; it is sufficient to consider a representative $\rho_L$ for each even conjugacy class in $S_6$, i.e. for each even cycle shape. 
    There are five possible cycle shapes, corresponding to the partitions $1+1+2+2$, $1+1+1+3$, $3+3$, $2+4$, $1+5$. Let $\CL_{(\rho_L,1)}$ denote one of these defects. 
    Then, the  $\CA\times \bar \CA$ representation labeled by $a_1,\ldots,a_6\in\ZZ/3\ZZ$ is annihilated by the operator $\hat\CL_{(\rho_L,1)}$ unless it satisfies \be\label{rhoLinv} (a_{\rho_L(1)},\ldots,a_{\rho_L(6)}) = (a_{1},\ldots,a_{6})+n(1,\ldots,1)\ ,\ee for some $n\in \ZZ/3\ZZ$. For a given $(\rho_L,1)$, the number of simple defects $\CL_{(\rho_L,1)}$ equals the number of representations satisfying \eqref{rhoLinv}, and can be obtained from each other by conjugation $\CL_g \CL_{(\rho_L,1)}\CL_g^\dual$ by invertible defects $\CL_g$ with $g\in \ZZ_3^4\subset G_{(1)^6}$. 
    Let us also notice that for every invertible defect $\CL_g$ with $g\in \ZZ_3^4\subset G_{(1)^6}$ and acting trivially on all representations satisfying \eqref{rhoLinv}, one has $\CL_g\CL_{(\rho_L,1)}=\CL_{(\rho_L,1)}\CL_g=\CL_{(\rho_L,1)}$.\\ 
    For all $\CL_{(\rho_L,1)}$ we can find an invertible $\CL_g$ such that $\CN=\CL_g\CL_\rho$ is unoriented -- this implies that $\CN$ acts on the chiral and antichiral algebras by  involutions $(\rho'_L,\rho'_R)$, \be\label{involutions}
    {\rho'_L}^2=1={\rho'_R}^2\ ,\qquad\rho'_L\rho'_R=\rho_L\ .\ee Furthermore, because $\CN^2$ acts trivially on the chiral and antichiral algebra, it must be a superposition of Verlinde lines, i.e. invertible defects in $\ZZ_3^4\subset G_{(1)^6}$. This implies that all the $\CN$ obtained in this way are duality defects for some subgroup $H\subseteq \ZZ_3^4$
    \be\label{Hduality} \CN^2=\sum_{h\in H} \CL_h\ ,
    \ee of order $|H|=\langle \CN\rangle ^2$. In the following table, for each partition of $\{1,\ldots,6\}$, we report the cycle shape of a generic automorphism $\rho_L\in A_6$, the dimension $\langle \CN\rangle$ of the duality defect $\CN$ acting on $\CA\times\bar\CA$ by automorphisms $(\rho'_L,\rho'_R)$ as in \eqref{involutions}, the group $H\subseteq \ZZ_3^4$ appearing in \eqref{Hduality}, and the generators of $H$ in the form \eqref{Z34}:
    \begin{center}
    \begin{tabular}{|c|c|c|c|c|}
    \hline
     Partition & $\rho_L$   & $\langle \CN\rangle$ & $H$  & Generators of $H$ \\
     \hline
      $1+1+2+2$   & $(ij)(kl)$ & $3$ & $\ZZ_3\times \ZZ_3$ & $t_it_j^{-1},\quad  t_kt_l^{-1}$\\
      $1+1+1+3$   & $(ijk)$ & $3$ & $\ZZ_3\times \ZZ_3$ & $t_it_j^{-1},\quad  t_it_jt_k$\\
      $3+3$   & $(ijk)(lmn)$ & $3$ & $\ZZ_3\times \ZZ_3$ & $t_it_j^{-1}t_l^{-1}t_m,\quad  t_it_jt_k$\\
      $2+4$   & $(ij)(klmn)$ & $9$ & $\ZZ_3^4$ & All $\prod_k t_k^{m_k}$\\
      $1+5$   & $(ijklm)$ & $9$ & $\ZZ_3^4$ & All $\prod_k t_k^{m_k}$\\
      \hline
    \end{tabular}\end{center}
    This analysis implies that the K3 model $\calC_{(1)^6}$ is self-orbifold with respect to the groups $H\subseteq \ZZ_3^4$ in this table.
    \item \emph{Defects induced by the description as a torus orbifold $T^4/\ZZ_3$.} We have seen in section \ref{s:OnetosixIsT4z3} that $\calC_{(1)^6}=\CT_{A_2^2}/\langle g\rangle$ for a suitable torus symmetry $g$ of order $3$ preserving the small $\CN=(4,4)$ superconformal algebra. From the general analysis of section \ref{s:torusOrbs}, this implies $\cTop_{(1)^6}$ contains a continuum of topological defects $T_\theta$ of quantum dimension $3$, that are induced by superpositions $W_\theta+W_{g(\theta)}+W_{g^2(\theta)}$ of torus model defects, and that are simple for generic $\theta\in \Gamma^{4,4}/(\Gamma^{4,4}\otimes \RR)$. When $x\in \Gamma^{4,4}/(\Gamma^{4,4}\otimes \RR)$ is a $g$-fixed point, i.e. $g(x)=x\mod \Gamma^{4,4}$, the defect $T_x$ decomposes as a superposition $\eta_x+Q\eta_x+Q^2\eta_x$, where $Q$ is a quantum symmetry of order $3$ and $\eta_x$ are invertible defects. The invertible defects $\eta_x$ obtained in this way, together with $Q$, generate an extraspecial group $3^{1+4}$, that must be a subgroup of $G_{(1)^6}$. Indeed, the quantum symmetry $Q$ can be identified, for example, with a symmetry $t_1t_2t_3=t_4t_5t_6$ acting as in \eqref{Z34}, and  the subgroup of $G_{(1)^6}$ commuting with $Q$ is isomorphic to $3^{1+4}\rtimes \ZZ_6$. The generators of the latter group are induced by the symmetries of the torus model $\CT$ that commute with $g$. The group $3^{1+4}\rtimes \ZZ_6$ contains the subgroup $\ZZ_3^4$ fixing the chiral and anti-chiral algebra $\CA\times \bar\CA$ -- indeed, the index of $3^{1+4}\rtimes \ZZ_6$ in $G_{(1)^6}$ is $20$, which is not divisible by $3$. However, none of the elements $\eta_x$ are contained in this $\ZZ_3^4$ group, because they all act non-trivially on some holomorphic or antiholomorphic operator.
    \item \emph{Duality defects.} We do not attempt a classification of all abelian groups $H\subset G_{(1)^6}$ with respect to which $\calC_{(1)^6}$ is self-orbifold, i.e $\calC/\langle H\rangle\cong \calC_{(1)^6}$, but just mention a few examples. When $H$ leaves invariant the full chiral and anti-chiral algebra, then the corresponding duality defect $\CN$ must act on $\CA\times\bar\CA$ by automorphisms. All cases where $\CN$ preserves the $\CN=(4,4)$ algebra were considered in point 2 above. As discussed in \ref{s:torusOrbs}, the continuum of defects induced by the torus orbifold construction contains also the duality defects $\CN$ for the order $9$ groups $H$ generated by $\eta_x$ and $Q$. Notice that the element $\eta_x$, and therefore the group $H$, does not act trivially on $\CA\times \bar\CA$, so that $\CN$ is not one of the defects considered in point 2 above.
\end{enumerate}

Let us describe the lattice $\Gamma^{4,20}$ of R-R charges in this model. A set of boundary states generating the full lattice was determined in \cite{Gaberdiel:2011fg}, using the Gepner model description. The same lattice can be also obtained in a way similar as section \ref{s:GTVWbranes}, using arguments based purely on the symmetry group $G_{(1)^6}$ of the model and on lattice theory. In particular, $G_{(1)^6}$ correspond to group $101$ in the list of \cite{HohnMason2016}, and their results imply that there is a unique point in the moduli space of K3 models with this symmetry group. It is easier to describe $\Gamma^{4,20}$ as a complex lattice in $\CC^{2,10}$, the complex vector space with indefinite sesquilinear form $$\langle z,w\rangle=\bar z_1w_1+\bar z_2w_2-\bar z_3w_3-\ldots-\bar z_{12}w_{12}\ ,$$ of signature $(2,10)$, where $z\equiv (z_1,\ldots z_{12})\in \CC^{2,10}$ and $w\equiv (w_1,\ldots w_{12})\in \CC^{2,10}$.  Using the notation $\omega:=e^{\frac{2\pi i}{3}}$ and $\theta=\omega-\bar\omega=i\sqrt{3}$, we find that $\Gamma^{4,20}\subset \CC^{2,10}$ is generated by the rows of the following matrix
\be\label{cplxRRlattice} \arraycolsep=3pt \def\arraystretch{0.8}
\frac{\sqrt{2}}{3}
\left(
\begin{array}{cccccccccccc}
	3 & 0 & 0 & 0 & 0 & 0 & 0 & 0 & 0 & 0 & 0 & 0 \\
	3 \omega  & 0 & 0 & 0 & 0 & 0 & 0 & 0 & 0 & 0 & 0 & 0 \\
	0 & 3 \theta & 0 & 0 & 0 & 0 & 0 & 0 & 0 & 0 & 0 & 0 \\
	0 & 3 \theta \omega  & 0 & 0 & 0 & 0 & 0 & 0 & 0 & 0 & 0 & 0 \\
	0 & -3 & -3 & 0 & 0 & 0 & 0 & 0 & 0 & 0 & 0 & 0 \\
	0 & 3 \bar\omega  & 3 \bar\omega  & 0 & 0 & 0 & 0 & 0 & 0 & 0 & 0 & 0 \\
	0 & -3 & 0 & -3 & 0 & 0 & 0 & 0 & 0 & 0 & 0 & 0 \\
	0 & 3 \bar\omega  & 0 & 3 \bar\omega  & 0 & 0 & 0 & 0 & 0 & 0 & 0 & 0 \\
	0 & -3 & 0 & 0 & -3 & 0 & 0 & 0 & 0 & 0 & 0 & 0 \\
	0 & 3 \bar\omega  & 0 & 0 & 3 \bar\omega  & 0 & 0 & 0 & 0 & 0 & 0 & 0 \\
	0 & -3 & 0 & 0 & 0 & -3 & 0 & 0 & 0 & 0 & 0 & 0 \\
	0 & 3 \bar\omega  & 0 & 0 & 0 & 3 \bar\omega  & 0 & 0 & 0 & 0 & 0 & 0 \\
	\theta \omega  & -\theta \omega  & -\theta \omega  & -\theta
	\omega  & \theta \omega  & -\theta \omega  & \theta \omega  & 0 &
	0 & 0 & 0 & 0 \\
	-\theta \omega  & -\theta \omega  & \theta \omega  & -\theta
	\omega  & -\theta \omega  & -\theta \omega  & 0 & \theta \omega  &
	0 & 0 & 0 & 0 \\
	-\theta \omega  & 0 & 0 & \theta \omega  & -\theta \omega  &
	-\theta \omega  & 0 & 0 & \theta \omega  & 0 & 0 & 0 \\
	0 & -\theta \omega  & \theta \omega  & 0 & \theta \omega  &
	-\theta \omega  & 0 & 0 & 0 & \theta \omega  & 0 & 0 \\
	\theta \omega  & 0 & -\theta \omega  & \theta \omega  & 0 &
	-\theta \omega  & 0 & 0 & 0 & 0 & \theta \omega  & 0 \\
	0 & \theta \omega  & \theta \omega  & \theta \omega  & \theta
	\omega  & 0 & 0 & 0 & 0 & 0 & 0 & \theta \omega  \\
	\theta & -\theta & -\theta & -\theta & \theta & -\theta &
	\theta & 0 & 0 & 0 & 0 & 0 \\
	-\theta & -\theta & \theta & -\theta & -\theta & -\theta & 0 &
	\theta & 0 & 0 & 0 & 0 \\
	-\theta & 0 & 0 & \theta & -\theta & -\theta & 0 & 0 & \theta & 0 &
	0 & 0 \\
	0 & -\theta & \theta & 0 & \theta & -\theta & 0 & 0 & 0 & \theta & 0
	& 0 \\
	0 & \theta & \theta & \theta & \theta & 0 & 0 & 0 & 0 & 0 & 0 &
	\theta \\
	0 & 2 \omega -\bar\omega  & 1 & 1 & 1 & 1 & 1 & 1 & 1 & 1 & 1 & 4 \\
\end{array}
\right)
\ee
The image of this complex lattice with respect to the standard map
$$ \CC^{2,10}\in (z_1,\ldots,z_{12}) \mapsto (\Re(z_1),\Im(z_1),\ldots, \Re(z_{12}),\Im(z_{12}))\in \RR^{4,20},
$$ defines a real even unimodular lattice of signature $(4,20)$, as can be checked by a direct computation. The construction is analogous to the definition of the complex Leech lattice (see for example chapter 7 section 8 in \cite{ConwaySloane}), with some modifications due to the different signature. The first two columns in \eqref{cplxRRlattice} correspond to complex coordinates $z_1,z_2$ for the space $\Pi\subset V$ of spectral flow generators, while the remaining $10$ columns are complex coordinates $z_3,\ldots,z_{12}$ for its orthogonal complement. Notice that first four vectors (rows) in the lattice basis, namely
$$
\begin{array}{cccccccccccc}
	(3 & 0 & 0 & 0 & 0 & 0 & 0 & 0 & 0 & 0 & 0 & 0) \\
	(3 \omega  & 0 & 0 & 0 & 0 & 0 & 0 & 0 & 0 & 0 & 0 & 0) \\
	(0 & 3 \theta & 0 & 0 & 0 & 0 & 0 & 0 & 0 & 0 & 0 & 0) \\
	(0 & 3 \theta \omega  & 0 & 0 & 0 & 0 & 0 & 0 & 0 & 0 & 0 & 0)\\
 \end{array}
$$
are contained in the subspace $\Pi$. By claim \ref{th:qdim}, the existence of such charges imply that all defects $\CL\in \cTop_{(1)^6}$ have integral quantum dimension.

\medskip

The coordinates $z_3,\ldots,z_{12}$ are a complex basis for the space $V\cap \Pi^\perp$ of R-R ground states in the $(\frac{1}{4},0;\frac{1}{4},0)$ representation on $\CN=(4,4)$. In particular, the basis consists of the $20$ highest weight vectors for the chiral and anti-chiral algebra $\CA\times \bar\CA$ labeled by
$$ [a_1,\ldots, a_{6}]=[+\frac{1}{2},+\frac{1}{2},+\frac{1}{2},-\frac{1}{2},-\frac{1}{2},-\frac{1}{2}]\qquad \text{and permutations.}
$$ Let us denote by $|i,j,k\rangle$, $1\le i<j<k\le 6$, the highest weight vector with the $+$ signs in positions $i,j,k$, i.e. in the $\CA\times\bar\CA$ representation $[a_1,\ldots,a_6]$ with
$$ a_n=\begin{cases}
    +1/2 & \text{if } n\in \{i,j,k\}\\
    -1/2 & \text{if } n\notin \{i,j,k\}\ .
\end{cases}
$$ This is a complex basis, since the CPT conjugate of $[a_1,\ldots, a_6]$ is the opposite $[-a_1,\ldots,-a_6]$. We identify the states $|i,j,k\rangle$ with the complex coordinates $z_3,z_3^*,\ldots,z_{12},z_{12}^*$ as follows:
\be\label{threetuples}
\begin{array}{cc|cc}
z_3\rightarrow |1,2,3\rangle & z_3^*\rightarrow  |4,5,6\rangle \qquad &\qquad   z_8\rightarrow |1,3,4\rangle & z_8^*\rightarrow  |2,5,6\rangle\\ 
z_4\rightarrow |1,2,4\rangle & z_4^*\rightarrow  |3,5,6\rangle \qquad &\qquad 
z_9\rightarrow |1,5,6\rangle & z_9^*\rightarrow  |2,3,4\rangle \\
z_5\rightarrow |1,2,5\rangle & z_5^*\rightarrow  |3,4,6\rangle \qquad &\qquad z_{10}\rightarrow |1,3,5\rangle & z_{10}^*\rightarrow  |2,4,6\rangle\\ 
z_6\rightarrow |1,4,6\rangle & z_6^*\rightarrow  |2,3,5\rangle \qquad &\qquad z_{11}\rightarrow |1,3,6\rangle & z_{11}^*\rightarrow  |2,4,5\rangle\\
z_7\rightarrow |1,4,5\rangle & z_7^*\rightarrow  |2,3,6\rangle\qquad &\qquad z_{12}\rightarrow |1,2,6\rangle & z_{12}^*\rightarrow  |3,4,5\rangle\\
\end{array}
\ee
The group of invertible defect in $\cTop_{(1)^6}$ is $G_{(1)^6}\cong \ZZ_3^4\rtimes A_6$. All such symmetries act trivially on the spectral flow operators, and therefore on the complex coordinates $z_1,z_2$. A permutation $\sigma\in A_6$ acts in the obvious way on the states $|i,j,k\rangle$
\be |i,j,k\rangle \mapsto |\sigma(i),\sigma(j),\sigma(k)\rangle\ .
\ee Using the correspondence \eqref{threetuples}, it is immediate to derive the action of $\sigma$ on the coordinates $z_3,z_3^*,\ldots,z_{12},z_{12}^*$. 

The elements of the subgroup $\ZZ_3^4$ are labeled by $(m_1,\ldots,m_6)\in (\ZZ/3\ZZ)^6$, with the condition $\sum_i m_i\equiv 0\mod 3$ and modulo $(1,1,1,1,1,1)$. They multiply the states $|i,j,k\rangle$, and therefore complex coordinates $z_3,\ldots, z_{12}$, by some cubic roots of unity, determined by
\be |i,j,k\rangle \mapsto \exp\Bigl(\frac{2\pi i}{3} [\sum_{n\in \{i,j,k\}} m_n - \sum_{n\notin\{i,j,k\}} m_n]\Bigr)|i,j,k\rangle\ .
\ee
It is easy to check that these transformations of the space $\CC^{2,10}$ correspond to automorphisms of the lattice spanned by the rows of \ref{cplxRRlattice}.

\medskip

As for the model $\calC_{GTVW}$, we computed a set of generators for the intersection
$$ B^{4,20}_{(1)^6}(\ZZ):= {\rm End}(\Gamma^{4,20})\cap B^{4,20}(\RR)\ ,
$$ between the $\ZZ$-linear endomorphisms of the lattice $\Gamma^{4,20}$ and the $401$-dimensional real space $B^{4,20}(\RR)$ of block diagonal $24\times 24$ matrices, with a upper left $4\times 4$  block proportional to the identity, and an unconstrained lower right $20\times 20$ block, see \eqref{blkspace}. As in the previous example, the $\ZZ$-module $B^{4,20}_{(1)^6}(\ZZ)$ has maximal rank $401$, i.e. equal to the dimension of $B^{4,20}(\RR)$. 

We also computed the submodule
$$ B^{4,20}_{(1)^6,inv}(\ZZ)\subset B^{4,20}_{(1)^6}(\ZZ)
$$ generated by the $\nL$ where $\CL$ is a superposition of \emph{invertible} defects in $\cTop_{(1)^6}$. As in the previous example, we find that the submodule $B^{4,20}_{(1)^6,inv}(\ZZ)$ has full rank $401$, with quotient
\be B^{4,20}_{(1)^6}(\ZZ)/B^{4,20}_{(1)^6,inv}(\ZZ)\cong \ZZ_3^{104}\ .
\ee 

For all the defects $\CL\in \cTop_{(1)^6}$ that are discussed above, the corresponding map $\nL$ is contained in $B^{4,20}_{(1)^6}(\ZZ)$, i.e. it can be written as an integral linear combination $\nL=\sum_g n_g \nL_g$ of maps $\nL_g$ related to invertible defects $\CL_g$, $g\in G_{(1)^6}$, for some $n_g\in \ZZ$. However, some $\nL$ cannot be written as linear combinations $\sum_g n_g \nL_g$ with \emph{non-negative} integral coefficients $n_g\in \ZZ_{\ge 0}$. An example is given by the duality defect $\CN_{(12)(34)}$ considered in point $2$ above, that acts on the chiral algebra $\CA$ by a permutation with cycle shape $(12)(34)$, while preserving the anti-chiral algebra $\bar\CA$. The square of $\CN_{(12)(34)}$ is a superposition of $9$ invertible defects in the $\ZZ_3\times \ZZ_3$ group generated by $t_1t_2^2$ and $t_3t_4^2$
\be \CN_{(12)(34)}^2=(\CI+\CL_{t_1t_2^2}+\CL_{t_1^2t_2})(\CI+\CL_{t_3t_4^2}+\CL_{t_3^2t_4})\ .
\ee This determines (up to a sign) the action of $\CN_{(12)(34)}$ on the basis \eqref{threetuples} of RR ground states in the $(\frac{1}{4},0;\frac{1}{4},0)$ representation of $\CN=(4,4)$, and therefore on the space $\CC^{2,10}$ with coordinates $z_1,\ldots,z_{12}$. In particular, the action is via the diagonal matrix
\be \arraycolsep=3pt \def\arraystretch{0.8}
\mathsf{N}_{(12)(34)}=
\left(
\begin{array}{cccccccccccc}
	3 & 0 & 0 & 0 & 0 & 0 & 0 & 0 & 0 & 0 & 0 & 0 \\ 
 0 & 3 & 0 & 0 & 0 & 0 & 0 & 0 & 0 & 0 & 0 & 0 \\
 0 & 0 & -3 & 0 & 0 & 0 & 0 & 0 & 0 & 0 & 0 & 0 \\
 0 & 0 & 0 & 0 & 0 & 0 & 0 & 0 & 0 & 0 & 0 & 0 \\
 0 & 0 & 0 & 0 & 0 & 0 & 0 & 0 & 0 & 0 & 0 & 0 \\
 0 & 0 & 0 & 0 & 0 & 0 & 0 & 0 & 0 & 0 & 0 & 0 \\
 0 & 0 & 0 & 0 & 0 & 0 & 0 & 0 & 0 & 0 & 0 & 0 \\
 0 & 0 & 0 & 0 & 0 & 0 & 0 & 0 & 0 & 0 & 0 & 0 \\
 0 & 0 & 0 & 0 & 0 & 0 & 0 & 0 & 0 & 0 & 0 & 0 \\
 0 & 0 & 0 & 0 & 0 & 0 & 0 & 0 & 0 & 0 & 0 & 0 \\
 0 & 0 & 0 & 0 & 0 & 0 & 0 & 0 & 0 & 0 & 0 & 0 \\
 0 & 0 & 0 & 0 & 0 & 0 & 0 & 0 & 0 & 0 & 0 & 0 \\\end{array}\right)
\ee where the minus sign is fixed by the requirement that this linear map induces an endomorphisms of the lattice of D-brane charges. On the other hand, all invertible defects $\CL_g$ act on the basis \eqref{threetuples} by multiplication by some cubic root of unity followed by a permutation. This means that in the corresponding matrix $\nL_g$ the sum of the entries in each row or column is a cubic root of unity. Because it is impossible to obtain $-3$ as a sum over three cubic roots of unity, it follows that $\mathsf{N}_{(12)(34)}$ cannot be written as a sum of three matrices $\nL_{g_1}+\nL_{g_2}+\nL_{g_3}$ for any $g_1,g_2,g_3\in G_{(1)^6}$.

\bigskip

As for the GTVW model, it is easy to exhibit a deformation of the model that lifts all non-trivial topological defects $\CL\in \cTop_{(1)^6}$. It is sufficient to consider a basis $\chi_{i,j}$, $i=1,\ldots,4$, $j=5,\ldots,24$ of the $80$-dimensional space of exactly marginal operators with respect to which all the linear operators $\hat\CL$ are represented by matrices with rational entries -- or, more generally, entries in some algebraic number field. 
For example, one could take the $|\chi_{ij}\rangle$ to be the states related by spectral flow to the R-R $|i,j,k\rangle$ in \eqref{threetuples}. With respect to this basis, the operators $\hat\CL$ are represented by matrices with entries in the algebraic number field $\QQ[\omega]$, where $\omega=e^{2\pi i/3}$. 
Then, one can take a linear combination $\chi=\sum_{j=5}^{24}\alpha_j \chi_{1,j}$ whose coefficients have transcendental ratios $\alpha_j/\alpha_k$, for example $\alpha_j=\pi^{j-5}$ as in \eqref{breakdef}. 
The same argument as in section \ref{s:GTVWbranes} shows that the only defects that satisfy \eqref{Linvar}, and are therefore preserved by the deformation $\chi$, are the ones proportional to the identity.

\section{Conclusions}\label{s:conclusions}

In this article, we discuss some general properties of the categories $\cTop_{\calC}$ of topological defects preserving the $\CN=(4,4)$ superconformal algebra and the spectral flow in a supersymmetric non-linear sigma model on K3 $\calC$.  In particular, we focus on the fusion of such topological defects with boundary states corresponding to BPS D-branes, and define a homomorphism from the fusion ring of $\cTop$ to the ring of $\ZZ$-linear endomorphisms on the lattice of D-brane charges. This construction, together with some standard assumptions about the moduli space of K3 models, allows us to derive a number of general properties of the topological defects in $\cTop_\calC$. For example, we show that the set of K3 models where $\cTop$ is not trivial (i.e., where there are simple defects distinct from the identity) has zero measure in the moduli space. Furthermore, we provide some restrictions on the possible quantum dimensions of the defects in $\cTop$ and a sufficient condition on the model $\calC$ (or rather on the corresponding point in the moduli space) for all quantum dimensions to be integral. Finally, we apply our methods in a couple of well understood examples of K3 models.

There are many directions of investigation that would be interesting to pursue:
\begin{enumerate}
    \item One of the main tools in our analysis is the map $\cTop \to {\rm End}(\Gamma^{4,20})$ that associates with each topological defect $\CL$ a $\ZZ$-linear endomorphism $\nL$ of the lattice $\Gamma^{4,20}$ of RR charges of the model. In section \ref{s:topdefsK3} we provided some conditions that must be satisfied by the endomorphisms $\nL$ that are associated with some defect. It is clear, though, that in general these conditions are not sharp enough to identify the image of the $\cTop \to {\rm End}(\Gamma^{4,20})$. Can we refine them?
    \item Similarly, we know that the map $\cTop \to {\rm End}(\Gamma^{4,20})$ is in general not injective. In fact, in the two K3 models $\calC$ that we studied, we found some families of defects $\CL_\theta\in \cTop_{\calC}$ depending on real parameters $\theta$, corresponding to the same $\nL\in {\rm End}(\Gamma^{4,20})$. On the other hand, there are some cases (for example, when $\nL$ is a multiple of the identity) where we were able to prove that $\nL$ is the image of a single $\CL$. Can we determine for which $\nL$ the corresponding defect $\CL$ is unique?\\
    Let us give some more insight into this question. It is known that each topological defect $\CL$ in the CFT $\calC$ can be described as a boundary state in the double theory $\calC\times \bar \calC$, see section \ref{s:doubletheory}. In this description, it is natural to identify $\nL\in {\rm End}(\Gamma^{4,20})$ with the R-R charge $\nL\in \Gamma^{4,20}\otimes (\Gamma^{4,20})^*$ of the boundary state in the double theory. Thus, a continuum of defects for $\nL$ in the model $\calC$ corresponds to a moduli space of D-branes with a given charge $\nL$ in the double theory. In particular, the twisted conserved currents implementing the deformations of the topological defect (see the discussion in section \ref{s:contDefects} and references therein) should correspond to suitable massless modes for open strings with both ends on the D-brane in the double theory. It is natural to wonder whether there can be examples where there is a \emph{discrete} set of defects $\CL\in \cTop$ with the same image $\nL$ via the map $\cTop\to {\rm End}(\Gamma^{4,20})$. They would correspond to a discrete set of different BPS D-branes carrying the same R-R charge in some non-linear sigma model on $K3\times K3$.\\
    As a second observation, we notice that a given topological defect $\CL$ is preserved by a marginal deformation of the K3 model if and only if the defect is `transparent' for the corresponding exactly marginal operator. This property is completely encoded in the form of $\nL$. As a consequence, \emph{all} defects $\CL$ with the same $\nL$ are either all lifted or all preserved by a certain deformation. Thus, for each $\nL$, one can identify some connected families of K3 models $\calC$ in the moduli space, with the property that $\cTop_{\calC}$ contains a set of defects $\CL$ with action $\nL$ on the D-brane charges. These arguments suggest that the cardinality of this set of defects might only depend on the connected family, and not on the particular K3 model.
    \item In section \ref{s:contDefects} we conjectured that $\cTop_{\calC}$ contains a continuum of defects only if $\calC$ is a (possibly generalised) orbifold of a torus model. It would be very interesting to either prove this conjecture or to find a counterexample. Notice that both models studied in our article are actually torus models. In this respect, it would be very useful to study an example of K3 model that is \emph{not} a torus model. Unfortunately, while torus orbifolds are necessarily a zero measure set in the moduli space of K3 models (they are a discrete union of subloci of dimension at most $16$), finding such an example of K3 model that is reasonably under control (e.g. it is rational) is not an easy task.  In particular, we do not know of any simple criterion to determine whether a given K3 model admits or not a description as a generalised orbifold of a torus model.
    \item In \cite{Cordova:2023qei}, the authors study the topological defects in various CY manifolds, including some K3 models. The idea is to classify the topological defects for some rational models (Gepner models), and then consider deformations by exactly marginal operators that preserve some of the defects. In this way, one can obtain information about topological defects in models that are (probably) not rational. This method is very effective in establishing the existence of topological defects in a neighborhood of known rational CFTs. In contrast, our approach, based on the action of topological defects on D-brane charges, allows us to provide information about possible defects even at points in the moduli space that are far from any (known) rational point. On the other hand, while our methods are effective in putting restrictions on properties of putative defects, they are not sufficient to prove the existence of such defects without an explicit description of the model. It would be very interesting to combine the two approaches: one could study the topological defects in many Gepner models, and then use our methods to determine precisely the sublocus of the moduli space where the topological defect is preserved. Furthermore, by considering the points in the moduli space at the intersection of several such loci, one might be able to establish the existence of K3 models with large categories of topological defects.
    \item Given a symmetry $g$ of a K3 model $\calC$, one can define a twining genus $\phi_g(\tau,z)$ by computing the elliptic genus of the model $\calC$ with the insertion of the topological defect $\CL_g$ along one of the cycles of the torus. In the same spirit, one can define a twining genus $\phi_\CL(\tau,z)$ for any topological defect $\CL\in \cTop_{\calC}$. On general grounds, one expects such functions to be holomorphic and modular, in a suitable sense.\footnote{More precisely, we expect them to be weak Jacobi forms of weight $0$ and index $1$ with respect to a suitable subgroup of $SL(2,\ZZ)$.} Twining genera play a prominent role in the calculation of the microstates degeneracy for 1/4-BPS black hole in string compactifications on K3 and on orbifolds thereof (CHL models). See for example
    \cite{Dabholkar:2007vk,Dabholkar:2008zy,Dabholkar:2006xa,David:2006ud,David:2006ru,David:2006ji,David:2006yn,Dijkgraaf:1996it,Strominger:1996sh,Jatkar:2005bh,Shih:2005uc,Paquette:2017gmb}.
 They are also the main characters in the `moonshine' conjectures for string theory on K3
    \cite{Cheng:2010pq,Cheng:2014zpa,Cheng:2016org,Eguchi:2010ej,Eguchi:2010fg,Gaberdiel:2010ca,Gaberdiel:2010ch,Gaberdiel:2010ca}.

 The study of twining genera related to topological defects, rather than symmetries, opens a number of new paths for the research in these subjects.
    \item A mysterious correspondence has been observed \cite{Duncan:2015xoa} between symmetries of K3 sigma models and automorphisms of a certain $\CN=1$ supersymmetric vertex operator algebra  $V^{s\natural}$ (a holomorphic superconformal field theory, in physics parlance) with central charge $12$. In particular, most (though, probably, not all) of the twining genera $\phi_g$ of K3 models can be exactly reproduced by certain $g$-twisted supertraces in $V^{s\natural}$. It is reasonable to expect the SCFT $V^{s\natural}$ to admit suitable topological defects that preserve the $\CN=1$ superconformal symmetry -- to the best of our knowledge, such defects have not been studied yet. While one cannot define boundary states in a purely holomorphic theory, we do expect some modification of our methods to apply to this case as well. We are planning to describe such methods in a forthcoming paper \cite{Angius:2024xxx}. It is natural to wonder whether the K3-VOA correspondence extends to the case of defects. 
\end{enumerate}

\bigskip

{\bf Acknowledgements.} The research of S.G and R.V. has been supported by a BIRD-2021 project (PRD-2021) and by the  PRIN Project n. 2022ABPBEY, ``Understanding quantum field theory through its deformations''. The research of R.A. has been supported by the grants CEX2020001007-S and PID2021-123017NB-I00, funded by MCIN/AEI/10.13039/501100011033 and by ERDF A way of making Europe. S.G. and R.V. would like to thank the Pollica Physics Center for support and hospitality during the programme ``New connections between Physics and Number Theory'' where work on this paper was undertaken, and to thank the participants to the programme for stimulating discussions.

\bigskip


\appendix

\section{Generators of $\ZZ_2^8:M_{20}$}\label{a:gensZ28M20}

In this section, we provide an explicit description of the group of symmetries $G_{GTVW}\cong \ZZ_2^8:M_{20}$ of the K3 model $\calC_{GTVW}$ considered in section \ref{s:Z28M20}; see \cite{Gaberdiel:2013psa} and \cite{Harvey:2020jvu} for a derivation. Recall that the full symmetry group of the CFT is $(SU(2)^6\times SU(2)^6):S_6$, where the $SU(2)^6$ factors are generated by the zero modes of the holomorphic and antiholomorphic currents, while $G_{GTVW}\cong\ZZ_2^8:M_{20}$ is the subgroup that preserves the $\CN=(4,4)$ superconformal algebra and the spectral flow. Every element of $(SU(2)^6\times SU(2)^6):S_6$ can be written as $(A_1,\ldots,A_6\ B_1,\ldots,B_6)\pi$ where $A_i$ and $B_i$ are $SU(2)$ matrices, and $\pi\in S_6$ is a permutation. As discussed in section \ref{s:Z28M20}, $SU(2)^6\times SU(2)^6$ does not act faithfully on the operators of the CFT. There is a subgroup $Z_0\cong \ZZ_2^5$ of the center $\ZZ_2^6\times \ZZ_2^6$ of $SU(2)^6\times SU(2)^6$ that acts trivially on all states of the theory (see eq.\eqref{Z0_subgroup_trivial_action}). This means that there is an ambiguity in writing the symmetries  as $(A_1,\ldots,A_6\ B_1,\ldots,B_6)\pi$, as one can multiply by any element in $Z_0\cong \ZZ_2^5$.

The $SU(2)$ matrices that one needs to write all elements of $\ZZ_2^8:M_{20}$ are
\be \one=\begin{pmatrix}
    1 & 0 \\ 0 & 1
\end{pmatrix}\ ,\qquad \ii=\begin{pmatrix}
    0 & 1 \\ -1 & 0
\end{pmatrix}\ ,\qquad \jj=\begin{pmatrix}
    0 & i \\ i & 0
\end{pmatrix}\ ,\qquad \kk=\begin{pmatrix}
    i & 0 \\ 0 & -i
\end{pmatrix}
\ee and their opposite matrices
\be \mone=\begin{pmatrix}
    -1 & 0 \\ 0 & -1
\end{pmatrix}\ ,\qquad \mi=\begin{pmatrix}
    0 & -1 \\ 1 & 0
\end{pmatrix}\ ,\qquad \mj=\begin{pmatrix}
    0 & -i \\ -i & 0
\end{pmatrix}\ ,\qquad \mk=\begin{pmatrix}
    -i & 0 \\ 0 & i
\end{pmatrix}\ .
\ee Notice that $\ii,\jj,\kk$ obey the quaternionic relations $\ii^2=\jj^2=\kk^2=-1$, $\ii\jj=\kk=-\jj\ii$. Finally we need the matrices
\be \Omega=\begin{pmatrix}
    \frac{1-i}{2} & \frac{1+i}{2}\\
    -\frac{1-i}{2} & \frac{1+i}{2}
\end{pmatrix}\ ,\qquad \Omega^{\dag}=\begin{pmatrix}
    \frac{1+i}{2} & \frac{1+i}{2}\\
    \frac{1-i}{2} & \frac{1-i}{2}
\end{pmatrix}
\ee that satisfy $\Omega^3=-\one$, so that $\Omega^6=\one$ and $\Omega^\dag=-\Omega^2$. Furthermore, $\Omega \ii\Omega^\dag=\kk$, $\Omega \jj\Omega^\dag=\ii$, $\Omega \kk\Omega^\dag=\jj$. The permutations $\pi$ will be denoted by their non-trivial cycles.

In this notation, the generators of $\ZZ_2^8:M_{20}$ are as follows:
\begin{itemize}
    \item Generators in the $\ZZ_2^6\times \ZZ_2^6$ subgroup of $SU(2)^6\times SU(2)^6$:
    \begin{align}
        &t_{23}=( \one\mone\mone\one\one\one\  \one\one\one\one\one\one)(1)^6\\
        &t_{24}=(\one\mone\one\mone\one\one\  \one\one\one\one\one\one)(1)^6\\
        &t_{25}=(\one\mone\one\one\mone\one\  \one\one\one\one\one\one)(1)^6\\
        &t_{26}=(\one\mone\one\one\one\mone\  \one\one\one\one\one\one)(1)^6
    \end{align} These generators form a $\ZZ_2^4$ group, with $\ZZ_2^4\subset \ZZ_2^8\subset \ZZ_2^8:M_{20}$, that contains $10$ elements $t_{ij}$, $2\le i<j\le 6$ with non-trivial $SU(2)$ factors $A_i=A_j=\mone$, and $5$ elements $t_{ijkl}$, $2\le i<j<k<l\le 6$, with non-trivial $SU(2)$ factors $A_i=A_j=A_k=A_l=\mone$.
    \item Further generators of the $\ZZ_2^8$ normal subgroup of $\ZZ_2^8:M_{20}$
    \begin{align}
        &s_1=(\one \, \one \, \mi \, \mi \, \mi\, \mi\, ;\  \one\,\one\,\ii\,\ii\,\ii\,\ii)(1)^6\\
        &s_2=(\one\,\one\,\mj\,\mj\,\mj\,\mj\,;\  \one\,\one\,\jj\,\jj\,\jj\,\jj)(1)^6\\
        &s_3=(\one\,\ii\,\kk\,\mj\,\mi\,\one\,;\  \one\,\ii\,\kk\,\jj\,\ii\,\one)(1)^6\\
        &s_4=(\one\,\kk\,\jj\,\mi\,\mk\,\one\,;\  \one\,\kk\,\jj\,\ii\,\kk\,\one)(1)^6
    \end{align}
    \item Generators of the $\ZZ_2^4$ subgroup of $M_{20}\cong \ZZ_2^4.A_5$ (modulo elements of the form $t_{ijkl}$)
    \begin{align}
        &v_1=(\one\,\one\,\ii\,\ii\,\ii\,\ii\,;\ \one \, \one \, \kk \, \kk \, \kk\, \kk )(1)^6\\
        &v_2=(\one\,\one\,\jj\,\jj\,\jj\,\jj\,;\  \one\,\one\,\ii\,\ii\,\ii\,\ii)(1)^6\\
        &v_3=(\one\,\jj\,\jj\,\one\,\ii\,\kk\,;\  \one\,\ii\,\ii\,\one\,\kk\,\jj)(1)^6\\
        &v_4=(\one\,\ii\,\jj\,\kk\,\one\,\ii\,;\  \one\,\kk\,\ii\,\jj\,\one\,\kk)(1)^6
    \end{align}
    Altogether, the $t_{ij}$, $s_1,\ldots,s_4$, and $v_1,\ldots,v_4$ generate the subgroup of $\ZZ_2^8:M_{20}$ contained in $SU(2)^6\times SU(2)^6$.
    \item Generators of the quotient $A_5$ of $M_{20}\cong \ZZ_2^4.A_5$
    \begin{align}
        &p_1=(\one\,\one\,\one\,\one\,\one\,\one\, ;\  \one\,\one\,\one\,\one\,\one\,\one)(34)(56)\\
        &p_2=(\one\,\one\,\one\,\one\,\one\,\one\, ;\  \one\,\one\,\one\,\one\,\one\,\one)(35)(46)\\
        &p_3=(\one\,\one\,\Omega^{\dag}\,\Omega\,\one\,\one\, ;\  \one\,\one\,\Omega^{\dag}\,\Omega\,\one\,\one)(25)(34)
    \end{align}
\end{itemize}

{\bf Action on the currents.} The $18$ left-moving currents of $su(2)^6$ can be denoted by a pair $(\vec\sigma\cdot \vec n,i)_L$, where $i=1,\ldots, 6$ denotes the $su(2)_1$ factors, and $\vec\sigma\cdot \vec n$ is a traceless hermitian matrix, written a linear combination of Pauli matrices with coefficients $\vec n=(n_1,n_2,n_3)\in \RR^3$, $\vec n\cdot \vec n=1$. The action of $g=(A_,\ldots,A_6; B_1,\ldots,B_6)\pi$ on $(\vec\sigma\cdot \vec n,i)_L$ is
\be g\cdot (\vec\sigma\cdot \vec n,i)_L=(A_{\pi(i)}\vec\sigma\cdot \vec nA_{\pi(i)}^\dag,\pi(i))_L
\ee where we make the permutation act before the $SU(2)^6\times SU(2)^6$ adjoint action. Similarly, the action on the right-moving currents is
\be g\cdot (\vec\sigma\cdot \vec n,i)_R=(B_{\pi(i)}\vec\sigma\cdot \vec nB_{\pi(i)}^\dag,\pi(i))_R\ .
\ee

{\bf Action on the RR ground states.} An orthonormal basis of the $24$-dimensional space of RR ground states is given by six `tetrads' of states $|1,i\rangle$, $|2,i\rangle$, $|3,i\rangle$, $|4,i\rangle$, $i=1,\ldots,6$. One can assign with each element $|a,i\rangle$ in a tetrad a $2\times 2$ matrix $U_a$, $a=1,2,3,4$, as follows
\begin{align}
    |1,i\rangle &\longrightarrow U_1=\frac{1}{\sqrt{2}} \begin{pmatrix}
        0 & 1\\ -1 & 0
    \end{pmatrix}=\frac{\ii}{\sqrt{2}}\\
    |2,i\rangle &\longrightarrow U_2=\frac{1}{\sqrt{2}} \begin{pmatrix}
        1 & 0\\ 0 & 1
    \end{pmatrix}=\frac{\one}{\sqrt{2}}\\
    |3,i\rangle &\longrightarrow U_3=\frac{1}{\sqrt{2}} \begin{pmatrix}
        i & 0
        \\0 & -i
    \end{pmatrix}=-\frac{\kk}{\sqrt{2}}\\
    |4,i\rangle &\longrightarrow U_4=\frac{1}{\sqrt{2}} \begin{pmatrix}
        0 & i\\ i & 0
    \end{pmatrix}=\frac{\jj}{\sqrt{2}}\ ,
\end{align} in such a way that $\langle a,i|b,j\rangle =\delta_{ij}\Tr(U_a^\dag U_b)$. The action of $g=(A_1,\ldots,A_6;B_1,\ldots,B_6)\pi$ on $|a,i\rangle\equiv (U_a,i)$, $a=1,\ldots,4$, $i=1,\ldots,6$ is
\be g\cdot (U_a,i)=(A_{\pi(i)}U_aB_{\pi(i)}^t,\pi(i))\ .
\ee

\section{A duality defect not in $\cTop_{GTVW}$}\label{a:NnotinTop}

Consider the K3 model $\calC_{GTVW}$ described in section \ref{s:Z28M20} and let $H\subset G_{GTVW}$ the abelian group $H\cong \ZZ_2$ generated by the symmetry $g=t_5t_6$. 
By computing its partition function, one can show that the orbifold $\calC/H$ is a consistent K3 model with the same bosonic chiral and anti-chiral algebras and representations as $\calC$; therefore, the two CFTs must be isomorphic $\calC\cong \calC/H$. In particular, all holomorphic and anti-holomorphic currents are $g$-invariant; the $g$-twisted sector contains no further currents.
Thus, we expect a topological duality defect $\CN$ such that
\be \CN^2=\CI+\CL_g\ ,\qquad \CL_g\CN=\CN\CL_g=\CN\ .
\ee
In this section, we describe the duality defect $\CN$ and show that $\CN\not\in \cTop_{GTVW}$, and in particular that $\CN$ does not preserve the $\CN=(4,4)$  superconformal algebra.

The fusion rules for $\CN\equiv \CN_{56}$ imply that $\hat \CN$ must act with eigenvalues $\pm \sqrt{2}$ on the states that are fixed by $g$, while eigenstates of $g$ with eigenvalue $-1$ must be in the kernel. On the other hand, these topological defects $\CN$ cannot commute with the whole $(\widehat{su}(2)_1)^{\oplus 6}$ algebra -- we have already identified all such defects, and they are all invertible.

What is the possible action on the chiral algebra?
\begin{itemize}
	\item Consider moving a defect $\CN$ through local holomorphic operators in $\CA$, we obtain a map from  $\CA$ to itself in a way compatible with OPE. This means that the defect should act on $\CA$ by an automorphism of order $2$.
	\item The action of $\CN$ on $\CA$ should not lift to a symmetry of the whole CFT. Indeed, if this was the case, then for a suitable CFT symmetry $h$, the fusion $\CL_h\CN$ would act trivially on $\CA$. But the only simple defects with this property are the Verlinde lines for the algebra $\CA$ and all such defects are invertible. It would follow that $\CN$ is a superposition of invertible defects, which cannot be true for a duality defect.
	\item Because $\hat \CN^2$ acts trivially on the $g$-invariant space of states $\CH^g$, then $\CN$ acts by permutations on the set of $\CA\times\bar\CA$ representations contained in $\CH^g$. In particular, it cannot map any such representation into one that is not contained in $\CH^g$.
\end{itemize}
The group of outer automorphisms of the chiral and antichiral algebra is the $S_6\times S_6$ permutation group acting separately on the holomophic and anti-holomorphic  currents. Only the diagonal $S_6\subset S_6\times S_6$ lifts to a symmetry of the whole theory; the other elements of do not preserve the set of $\CA\times\bar\CA$ representations in the space of states $\CH$. If instead we consider the subgroup of $S_6\times S_6$ that preserve only the set of representations contained in the subspace $\CH^g\subset \CH$ of $g$-invariant states, then there is one non-trivial choice (modulo the diagonal $S_6$ symmetry): the involution exchanging the $5$-th and the $6$-th $\widehat{su}(2)_1$ components on the (say) holomorphic side, while keeping the anti-holomorphic side fixed. 

Therefore, modulo automorphisms that lift to CFT symmetries, this is the only possible action of $\CN$ on the algebra $(\widehat{su}(2)_1)^{ 6}$.

We can calculate the $\CN$-twined partition functions
\be Z^{NSNS}_{1,\CN}=\Tr_{NS-NS}(\hat\CN q^{L_0-\frac{c}{24}}\bar q^{\bar L_0-\frac{\bar c}{24}})\ ,\qquad \tilde{Z}^{NSNS}_{1,\CN}=\Tr_{NS-NS}((-1)^{F+\tilde F}\hat\CN q^{L_0-\frac{c}{24}}\bar q^{\bar L_0-\frac{\bar c}{24}})
\ee
\be Z^{RR}_{1,\CN}=\Tr_{R-R}(\hat\CN q^{L_0-\frac{c}{24}}\bar q^{\bar L_0-\frac{\bar c}{24}})\ ,\qquad \tilde{Z}^{RR}_{1,\CN}=\Tr_{R-R}((-1)^{F+\tilde F}\hat\CN q^{L_0-\frac{c}{24}}\bar q^{\bar L_0-\frac{\bar c}{24}})
\ee
using the characters of $\widehat{su}(2)_1$, namely
\be
{\rm ch}_{1,0}(\tau,z)= \frac{\theta_3(2\tau,2z)}{\eta(\tau)}\ ,\qquad {\rm ch}_{1,\frac{1}{2}}(\tau,z)= \frac{\theta_2(2\tau,2z)}{\eta(\tau)}\ ,
\ee
from which
\be
{\rm ch}_{1,0}(\tau,0)= \frac{\theta_3(2\tau,0)}{\eta(\tau)}=\frac{\eta(2\tau)^5}{\eta(\tau)^3\eta(4\tau)^2}\ ,\qquad {\rm ch}_{1,\frac{1}{2}}(\tau,0)= \frac{\theta_2(2\tau,0)}{\eta(\tau)}=2\frac{\eta(4\tau)^2}{\eta(2\tau)\eta(\tau)}\ .
\ee
The S-transformations are
\be {\rm ch}_{1,0}(-\tfrac{1}{\tau},0)=\frac{1}{\sqrt{2}}\bigl({\rm ch}_{1,0}(\tau,0)+{\rm ch}_{1,\frac{1}{2}}(\tau,0)\bigr)\ ,\quad {\rm ch}_{1,\frac{1}{2}}(-\tfrac{1}{\tau},0)=\frac{1}{\sqrt{2}}\bigl({\rm ch}_{1,0}(\tau,0)-{\rm ch}_{1,\frac{1}{2}}(\tau,0)\bigr)
\ee
Recall $\hat \CN$ acts on the $t_5t_6$ invariant states by $\sqrt{2}$ times the permutation of the fifth and sixth holomorphic $\widehat{su}(2)_1$ components.
It follows that the $\CN$-twined partition function are
\begin{align*} Z^{NSNS}_{1,\CN}=\sqrt{2}&\sum_{\substack{a_1,\ldots, a_4\in\{0,\frac{1}{2}\}\\ \sum a_i\in \ZZ}}\Bigl[\prod_{i=1}^4\left|{\rm ch}_{1,a_i}(\tau,0)\right|^2\bigl({\rm ch}_{1,0}(2\tau,0)\overline{{\rm ch}_{1,0}(\tau,0)}^2+{\rm ch}_{1,\frac{1}{2}}(2\tau,0)\overline{{\rm ch}_{1,\frac{1}{2}}(\tau,0)}^2\bigr)\\
&+\prod_{i=1}^4{\rm ch}_{1,a_i}(\tau,0)\overline{{\rm ch}_{1,\frac{1}{2}-a_i}(\tau,0)}\bigl({\rm ch}_{1,0}(2\tau,0)\overline{{\rm ch}_{1,\frac{1}{2}}(\tau,0)}^2+{\rm ch}_{1,\frac{1}{2}}(2\tau,0)\overline{{\rm ch}_{1,0}(\tau,0)}^2\bigr)\Bigr]
\end{align*}
\begin{align*} \tilde{Z}^{NSNS}_{1,\CN}=\sqrt{2}&\sum_{\substack{a_1,\ldots, a_4\in\{0,\frac{1}{2}\}\\ \sum a_i\in \ZZ}}\Bigl[\prod_{i=1}^4\left|{\rm ch}_{1,a_i}(\tau,0)\right|^2\bigl({\rm ch}_{1,0}(2\tau,0)\overline{{\rm ch}_{1,0}(\tau,0)}^2+{\rm ch}_{1,\frac{1}{2}}(2\tau,0)\overline{{\rm ch}_{1,\frac{1}{2}}(\tau,0)}^2\bigr)\\
&-\prod_{i=1}^4{\rm ch}_{1,a_i}(\tau,0)\overline{{\rm ch}_{1,\frac{1}{2}-a_i}(\tau,0)}\bigl({\rm ch}_{1,0}(2\tau,0)\overline{{\rm ch}_{1,\frac{1}{2}}(\tau,0)}^2+{\rm ch}_{1,\frac{1}{2}}(2\tau,0)\overline{{\rm ch}_{1,0}(\tau,0)}^2\bigr)\Bigr]
\end{align*}
\begin{align*} Z^{RR}_{1,\CN}=\sqrt{2}&\sum_{\substack{a_1,\ldots, a_4\in\{0,\frac{1}{2}\}\\ \sum a_i\in \frac{1}{2}+\ZZ}}\Bigl[\prod_{i=1}^4\left|{\rm ch}_{1,a_i}(\tau,0)\right|^2\bigl({\rm ch}_{1,0}(2\tau,0)\overline{{\rm ch}_{1,0}(\tau,0)}^2+{\rm ch}_{1,\frac{1}{2}}(2\tau,0)\overline{{\rm ch}_{1,\frac{1}{2}}(\tau,0)}^2\bigr)\\
&+\prod_{i=1}^4{\rm ch}_{1,a_i}(\tau,0)\overline{{\rm ch}_{1,\frac{1}{2}-a_i}(\tau,0)}\bigl({\rm ch}_{1,0}(2\tau,0)\overline{{\rm ch}_{1,\frac{1}{2}}(\tau,0)}^2+{\rm ch}_{1,\frac{1}{2}}(2\tau,0)\overline{{\rm ch}_{1,0}(\tau,0)}^2\bigr)\Bigr]
\end{align*}
\begin{align*} \tilde{Z}^{RR}_{1,\CN}=\sqrt{2}&\sum_{\substack{a_1,\ldots, a_4\in\{0,\frac{1}{2}\}\\ \sum a_i\in \frac{1}{2}+\ZZ}}\Bigl[\prod_{i=1}^4\left|{\rm ch}_{1,a_i}(\tau,0)\right|^2\bigl({\rm ch}_{1,0}(2\tau,0)\overline{{\rm ch}_{1,0}(\tau,0)}^2+{\rm ch}_{1,\frac{1}{2}}(2\tau,0)\overline{{\rm ch}_{1,\frac{1}{2}}(\tau,0)}^2\bigr)\\
&-\prod_{i=1}^4{\rm ch}_{1,a_i}(\tau,0)\overline{{\rm ch}_{1,\frac{1}{2}-a_i}(\tau,0)}\bigl({\rm ch}_{1,0}(2\tau,0)\overline{{\rm ch}_{1,\frac{1}{2}}(\tau,0)}^2+{\rm ch}_{1,\frac{1}{2}}(2\tau,0)\overline{{\rm ch}_{1,0}(\tau,0)}^2\bigr)\Bigr]
\end{align*}
By a direct calculation we obtain
\be \tilde{Z}^{RR}_{1,\CN}=(16+36q+96q^2+\ldots)+O(\bar q^1)\ .
\ee This shows that  $\CN$ does not preserve any holomorphic supercurrent. Indeed, if a holomorphic supercurrent were preserved by $\CN$, then $\tilde{Z}^{RR}_{1,\CN}$ would receive contributions only from the RR ground states, and it would be a constant in $q$.
 \section{Defects acting by automorphisms of the chiral algebra}\label{s:automA}

In this section we classify the topological defects  $\CL\in \cTop_{GTVW}$ of the model $\calC_{GTVW}$ of section \ref{s:Z28M20}, such that $\hat \CL$ acts on all holomorphic fields generating the bosonic chiral algebra $\CA\cong (\widehat {su}(2)_1)^6$ by an algebra automorphism $\rho_L$ times the quantum dimension $\langle \CL\rangle$. This implies that when a defect line $\CL$ is moved past the insertion point of one of the holomorphic currents $j(z)$, the latter gets simply transformed into the current $\rho_L(j(z))$.  Similarly, we require the defect to act on the anti-holomorphic fields by a (possibly different) automorphism $\rho_R$. The group of automorphisms of $(\widehat {su}(2)_1)^6$ is $SO(3)^6\rtimes S_6$, where $SO(3)^6$ are the inner automorphisms generated by the zero modes of the currents; note that the center $\ZZ_2^6$ of $SU(2)^6$ acts trivially on the algebra itself. 
We still require $\rho_L$ and $\rho_R$ to act trivially on the $\CN=(4,4)$ superconformal algebra; this condition constrains $\rho_L$ and $\rho_R$ to be in a finite subgroup of $SO(3)^6\rtimes S_6$. If the automorphisms $(\rho_L,\rho_R)$ extends to a symmetry of the CFT, then there is an invertible defect $\CL_g$, for some  $g\in G_{GTVW}\cong \ZZ_2^8:M_{20}$, such that the fusion $\CL_g\CL$ acts trivially on the whole $\CA\otimes\bar\CA$. The only simple defects with this property are the Verlinde lines discussed in section \ref{s:topdefsGTVW}, where it is shown that they are all invertible. This means that, in this case, $\CL$ is just a superposition of invertible defects.

\bigskip

Therefore, in order to get some new defects, we have to require that $(\rho_L,\rho_R)$ does \emph{not} lift to a symmetry of the CFT. Now, inner  automorphisms $SU(2)^6\times SU(2)^6$ are generated by the current zero modes, and therefore always define CFT symmetries. The group of outer automorphisms of $\CA\times \bar\CA$ is $S_6\times S_6$, and in this case, only the diagonal $S_6^{diag}\subset S_6\times S_6$, permuting the holomorphic and antiholomorphic $\widehat{su}(2)_1$ factors in the same way, lifts to a CFT symmetry. 
Thus, the defects acting by algebra automorphisms, modulo fusion with invertible defects from the left or from the right, correspond to non-trivial double cosets in \be \Bigl((SU(2)^6\times SU(2)^6)\rtimes S_6^{diag}\Bigr)\backslash \Bigl((SU(2)^6\times SU(2)^6)\rtimes (S_6\times S_6)\Bigr)/\Bigl((SU(2)^6\times SU(2)^6)\rtimes S_6^{diag}\Bigr)\ .\ee 

For each such coset, we can always choose a representative $(\rho_L,\rho_R)$ with $\rho_R=1$. Furthermore, we can choose one $\rho_L$ for each $S_6$-conjugacy class, i.e. for each possible cycle shape.

Finally, we require $\rho_L$ to preserve the $\CN=4$ superconformal algebra and the spectral flow generators. This implies that the induced permutation must be contained in the $A_5\subset S_6$ subgroup of \emph{even} permutations fixing the first $\widehat{su}(2)_1$ factor. There are only three possible non-trivial cycle shapes, corresponding to the partitions $1+1+1+3$, $1+1+2+2$ and $1+5$. Let us consider each of these three cases in detail.

\bigskip

\noindent{\bf Partition $1+1+1+3$.} It is known (see appendix \ref{a:gensZ28M20}), that $G_{GTVW}$ contains a symmetry 
\be g_3=p_1\circ p_3=(\one\,\one\,\Omega\,\Omega^{\dag}\,\one\,\one\, ;\  \one\,\one\,\Omega\,\Omega^{\dag}\,\one\,\one)(265)\ ,
\ee  where we denote the elements of $(SU(2)^6\times SU(2)^6)\rtimes S_6^{diag}$ by $(A_1,\ldots,A_6\, ; B_1,\ldots,B_6)\pi$, with $A_i,B_i\in SU(2)$ and $\pi\in S_6^{diag}$, and where
\be \Omega=\begin{pmatrix}
    \frac{1-i}{2} & \frac{1+i}{2}\\
    -\frac{1-i}{2} & \frac{1+i}{2}
\end{pmatrix}\in SU(2)\ .
\ee This means that if we define $(\rho_L,\rho_R)$ as
$$ \rho_L=(\one\,\one\,\Omega\,\Omega^{\dag}\,\one\,\one)(265)\ ,\qquad \rho_R=1\ .
$$ the left automorphism $\rho_L$ acts on the holomorphic algebra in the same way as the symmetry element $g_3$. This ensures that the holomorphic $\CN=4$ supercurrents are invariant under $\rho_L$. However, $\rho_R$ is not the same as for $g_3$, and in particular the left- and right-moving $\widehat{su}(2)_1$ factors are permuted in different ways. This means that the pair $(\rho_L,\rho_R)$ cannot be extended to a symmetry of the whole CFT. 

More explicitly, any such symmetry would have to map any field in the representation $[110000;110000]$, for example, to a field in a representation $[100010;110000]$; however, while the former field is in the spectrum of the theory, the latter is not. This implies that there cannot be any \emph{invertible} defect acting by $(\rho_L,\rho_R)$ on the chiral and antichiral algebras. However, there is no obstruction to having a non-invertible defect with such an action. Indeed, a non-invertible defect $\CL$, when circling a field in the $[110000;110000]$ representation, can simply annihilate it. More generally, the operator $\hat\CL$ associated with any such defect needs to annihilate any field in a representation $[a_1\ldots,a_6;b_1\ldots b_6]$ such that $[\rho_L(a_1\ldots,a_6);\rho_R(b_1\ldots b_6)]$ is not in the spectrum. The NS-NS representations $[a_1\ldots,a_6;b_1\ldots b_6]$ that are not necessarily annihilated by $\CL$ are the ones satisfying \be\label{nonzerocond} a_2=a_5=a_6\ ,\qquad b_2=b_5=b_6\ ,\ee and can be grouped into four sets:
\begin{align*}
    \Omega_1&=\{[000000;000000],\ [111111;000000],\ [000000;111111],\ [111111;111111] \}\ ,\\
    \Omega_2&=\{[101000;101000],\ [010111;101000],\ [101000;010111],\ [010111;010111] \}\ ,\\
    \Omega_3&=\{[100100;100100],\ [011011;100100],\ [100100;011011],\ [011011;011011] \}\ ,\\
    \Omega_4&=\{[001100;001100],\ [110011;001100],\ [001100;110011],\ [110011;110011] \}\ .
\end{align*} The four members in each set are related to each other by the action of the $\CN=(4,4)$ supercurrents. Therefore, the requirement that the $\CN=(4,4)$ algebra is invariant under $\hat\CL$ implies that the action of $\hat\CL$ is the same on all representations in the same set $\Omega_i$. We conclude that the action of $\hat\CL$ on the NS-NS sector depends only on four parameters $\alpha_1,\ldots,\alpha_4$ as
\be\label{Lauto} \hat\CL=\sum_{i=1}^4 \alpha_iP^i_{\rho_L,\rho_R}\ .
\ee Here, for each $i=1,\ldots,4$, $P^i_{\rho_L,\rho_R}$ acts by the automorphism $(\rho_L,\rho_R)$ on the representations in the set $\Omega_i$, while it annihilates all fields in the representations in $\Omega_j$ for $j\neq i$. 

Requiring $\CL$ to commute with the spectral flow operators fixes the action of $\hat\CL$ on the R-R sector to be of the same form \eqref{Lauto} with the same parameters $\alpha_i$, where now each $P^i_{\rho_L,\rho_R}$ is non-zero only on the set of representations $\tilde\Omega_i$ with
\begin{align*}
    \tilde\Omega_1&=\{[100000;100000],\ [011111;100000],\ [100000;011111],\ [011111;011111] \}\ ,\\
    \tilde\Omega_2&=\{[001000;001000],\ [110111;001000],\ [001000;110111],\ [110111;110111] \}\ ,\\
    \tilde\Omega_3&=\{[000100;000100],\ [111011;000100],\ [000100;111011],\ [111011;111011] \}\ ,\\
    \tilde\Omega_4&=\{[101100;101100],\ [010011;101100],\ [101100;010011],\ [010011;010011] \}\ .
\end{align*}
The four parameters $\alpha_1,\ldots,\alpha_4$ are constrained by the Cardy-like conditions that are obtained by considering the torus partition function with the defect line $\CL$ wrapping one of the cycles. Notice that if $(\alpha_1,\alpha_2,\alpha_3,\alpha_4)$ correspond to a consistent defect $\CL$ acting by $(\rho_L,\rho_R)$ on the chiral and antichiral algebras, then also $(\alpha_1,-\alpha_2,-\alpha_3,\alpha_4)$, $(\alpha_1,\alpha_2,-\alpha_3,-\alpha_4)$, and $(\alpha_1,-\alpha_2,\alpha_3,-\alpha_4)$ correspond to consistent defects, since they can be obtained by fusion of $\CL$ with the invertible defects acting trivially on all the currents and all the $\CN=(4,4)$ supercurrents (in particular, fusion with the symmetries $t_2t_3$, $t_2t_4$ and $t_3t_4$ gives all such defects; fusion with the symmetries $t_2t_5$, $t_2t_6$, and $t_5t_6$ leaves each of these defects invariant).

Let us consider the possible fusion products of the simple defect $\CL$ with $\alpha_1=\ldots=\alpha_4$. The dual defect $\CL^\dual$ acts on the space of states by the adjoint operator
\be \hat\CL^\dual = \sum_{i=1}^4 \alpha_i^*P^i_{\rho_L^{-1},\rho_R^{-1}}
\ee so that
the product $\CL\CL^\dual$ acts by
\be \hat\CL\hat\CL^\dual=\sum_{i=1}^4 |\alpha_i|^2 P^i\ ,
\ee where $P^i$ is the projector on the representations in the set $\Omega_i$. On the other hand, we know that $\CL\CL^\dual=\CI+\ldots $ acts trivially on the whole algebra $\widehat{su}(2)_1^6\oplus \widehat{su}(2)_1^6$ as well as on the $\CN=(4,4)$ algebra. This means that it must be a superposition of the invertible defects $\CL_{t_it_j}$ generating the $\ZZ_2^4$ subgroup of $\ZZ_2^8:M_{20}$ that is contained in the centre of $SU(2)^6\times SU(2)^6$. Furthermore, $\CL\CL^\dual$ must annihilate the representations $[a_1,\ldots,b_6]$ of the chiral algebra that do not satisfy \eqref{nonzerocond}, and must be a sum with positive coefficients of the projectors $P_i$. The only possibility is 
\be \CL\CL^\dual=\CI+\CL_{t_2t_5}+\CL_{t_2t_6}+\CL_{t_5t_6}\ ,
\ee so that
$$ |\alpha_i|=\langle \CL \rangle =\langle \CL^\dual \rangle=2\ .
$$

\noindent
The torus partition function with the insertion of $\CL$ is
\be Z_{\CL}(\tau)=\Tr(\hat\CL q^{L_0-\frac{c}{24}}\bar q^{\bar L_0-\frac{\bar c}{24}})=\sum_{i=1}^4 \alpha_i \Tr(P_{\rho_L,\rho_R}^i q^{L_0-\frac{c}{24}}\bar q^{\bar L_0-\frac{\bar c}{24}})\ ,
\ee
where 
\begin{align*}\Tr(P_{\rho_L,\rho_R}^i q^{L_0-\frac{c}{24}}\bar q^{\bar L_0-\frac{\bar c}{24}})=\sum_{[a_1,\ldots,b_6]\in\Omega_i} \Bigl[&\ch_{a_1}(\tau,0)\ch_{a_2}(3\tau,0)\ch_{a_3}(\tau,\frac{1}{6})\ch_{a_4}(\tau,\frac{1}{6})\\
&\overline{\bigl(\ch_{b_1}(\tau,0)\ch_{b_2}(\tau,0)^3\ch_{b_3}(\tau,0)\ch_{b_4}(\tau,0)\bigr)}\Bigr]
\end{align*} 
Let us calculate $Z_\CL$ with the ansatz
\be \alpha_1=\alpha_2=\alpha_3=\alpha_4= \langle \CL \rangle \label{ansatz_alpha} \ee
Using the $su(2)$ characters, $Z_{\mathcal{L}} (\tau)$ takes the form:
{\footnotesize{
\begin{equation*}
\begin{split}
      Z_{\mathcal{L}} (\tau) &  =  \frac{\langle \CL \rangle }{\overline{ \eta (\tau)^6} \eta (\tau)^3 \eta(3 \tau)} \left\lbrace \left( \overline{\theta_3 (2 \tau)^6} + \overline{\theta_2 (2 \tau)^6} \right)\left[ \theta_3 (2 \tau) \theta_3(6 \tau) \theta_3 \left( 2 \tau, \frac{1}{3} \right)^2 + \theta_2 (2 \tau) \theta_2(6 \tau) \theta_2 \left( 2 \tau, \frac{1}{3} \right)^2  \right] +\right. \\
    & + \left. \left( \overline{\theta_3 (2 \tau)^2} \overline{\theta_2 (2 \tau)^4}+ \overline{\theta_2 (2 \tau)^2} \overline{\theta_3 (2 \tau)^4} \right) \left[ 2 \left( \theta_3 \left( 2 \tau, \frac{1}{3} \right) \theta_2 \left( 2 \tau, \frac{1}{3} \right) \right) \left( \theta_3 (2 \tau) \theta_2 (6 \tau)  + \theta_2 (2 \tau) \theta_3 (6 \tau)  \right) +   \right. \right. \\
    & \left. \left.   +  \left( \theta_3 (2 \tau) \theta_3 (6 \tau) \theta_2 \left( 2 \tau, \frac{1}{3} \right)^2 + \theta_2 (2 \tau) \theta_2 (6 \tau) \theta_3 \left( 2 \tau, \frac{1}{3} \right)^2 \right) \right] \right\rbrace. \\
\end{split}
\end{equation*}
}}
Applying the transformation $\tau \mapsto - 1 / \tau$ we get:
{\footnotesize{
\begin{equation*}
\begin{split}
     & Z^{\mathcal{L}} (\tau)  =  \frac{ \langle \CL \rangle q^{1/18}}{32 \overline{\eta (\tau)^6} \eta(\tau)^3 \eta (\tau / 3)} \left\lbrace \left[ \left( \theta_3 \left( \frac{\tau}{2} \right) \theta_3 \left( \frac{\tau}{6} \right) \theta_4 \left( \frac{\tau}{2} ; \frac{\tau}{6} \right)^2 +  \theta_4 \left( \frac{\tau}{2} \right) \theta_4 \left( \frac{\tau}{6} \right) \theta_3 \left( \frac{\tau}{2} ; \frac{\tau}{6} \right)^2\right) \right. + \right. \\
        & \left. \left. + 2 \left( \theta_3 \left( \frac{\tau}{2} ; \frac{\tau}{6} \right) \theta_4 \left( \frac{\tau}{2} ; \frac{\tau}{6} \right) \right)\left( \theta_3 \left( \frac{\tau}{2} \right) \theta_4 \left( \frac{\tau}{6} \right) + \theta_4 \left( \frac{\tau}{2} \right) \theta_3 \left( \frac{\tau}{6} \right)\right) \right] \left( \overline{\theta_3 \left( \frac{\tau}{2}\right)^4} \overline{\theta_4 \left( \frac{\tau}{2}\right)^2} + \overline{\theta_4 \left( \frac{\tau}{2}\right)^4} \overline{\theta_3 \left( \frac{\tau}{2}\right)^2} \right) + \right. \\
     & \left. +  \left[ \theta_3 \left( \frac{\tau}{2} \right) \theta_3 \left( \frac{\tau}{6} \right) \theta_3 \left( \frac{\tau}{2} ; \frac{\tau}{6} \right)^2 + \theta_4 \left( \frac{\tau}{2} \right) \theta_4 \left( \frac{\tau}{6} \right) \theta_4 \left( \frac{\tau}{2} ; \frac{\tau}{6} \right)^2 \right] \left( \overline{\theta_3 \left(\frac{\tau}{2} \right)^6} + \overline{\theta_4 \left(\frac{\tau}{2} \right)^6} \right) \right\rbrace, \\
\end{split}
\end{equation*}
}}
which series expansion in $q$ and $\bar{q}$ is:
{\footnotesize{
\begin{equation*}
\begin{split}
    Z^{\mathcal{L}} = & \frac{\langle \CL \rangle }{ q^{1/4} \overline{q}^{1/4} } \left[ \left(\frac{ q^{1/6}}{2}+\frac{3 }{2}q^{1/2}+2  q^{2/3}+3  q^{5/6}+6  q+9  q^{7/6}+12  q^{4/3}+\frac{39}{2}  q^{3/2}+28  q^{5/3}+36  q^{11/6}+ \right. \right. \\
    & \left. +54  q^2+  \frac{149}{2}  q^{13/6} +96  q^{7/3}+\frac{273}{2}  q^{5/2}+182  q^{8/3}+234  q^{17/6}+O\left(q^{3}\right) \right) + \bar{q}^{1/2} \left(6  q^{1/6}+40  q^{1/3} +\right. \\
    & \left. +82  q^{1/2}+128  q^{2/3}+  196  q^{5/6}+320  q+524  q^{7/6}+776  q^{4/3}+1098  q^{3/2}+1616  q^{5/3}+2320  q^{11/6}+ \right. \\
    & \left. \left. + O\left(q^{2}\right)\right) + \bar{q} \left(33  q^{1/6} +160  q^{1/3} +355  q^{1/2}+ 548  q^{2/3} +838  q^{5/6}+O \left( q \right) \right) +O\left(\overline{q}^{3/2} \right) \right] . \\
\end{split}
\end{equation*}
}}
As predicted above, we find that such expansion produces only integer non-negative coefficients if $\langle \CL \rangle =2$. \\

Notice that the automorphism $(\rho_L,\rho_R)$ has order $3$, so that $\CL^2$ acts on the chiral algebra by $(\rho_L^2,\rho_R^2)=(\rho_L^{-1},\rho_R^{-1})$ (the same as $\CL^\dual$), while $\CL^3$ acts trivially on the chiral algebra. This suggests
\be \CL^2=2\CL^\dual\ ,
\ee and
\be \CL^3=2\CL\CL^\dual=2(\CI+\CL_{t_2t_5}+\CL_{t_2t_6}+\CL_{t_5t_6})\ ,
\ee which fits with the quantum dimensions.

We notice that if $\CL_{p_1}$ is the invertible defect corresponding to the symmetry $$ p_1=(\one\,\one\,\one\,\one\,\one\,\one\, ;\  \one\,\one\,\one\,\one\,\one\,\one)(34)(56)$$ of order $2$, then $\CL_{p_1}\CL\CL_{p_1}$ acts by $(\rho_L^{-1},\rho_R^{-1})$ on the chiral algebra. This suggests that
\be \CL_{p_1}\CL\CL_{p_1}=\CL^\dual\ .
\ee As a consequence, if we define
\be \CN_{256}:=\CL\CL_{p_1}\ ,
\ee then $\CN_{256}$ is unoriented $\CN_{256}=\CN_{256}^\dual$, and
\be \CN_{256}^2= \CI+\CL_{t_2t_5}+\CL_{t_2t_6}+\CL_{t_5t_6}\ .
\ee This means that $\CN_{256}$ is the duality defect related to the fact that the theory is self-orbifold with respect to the $\ZZ_2\times \ZZ_2$ group of symmetries with generators $t_2t_5$ and $t_2t_6$. Similarly, by conjugating by invertible defects, for every $1<i<j<k\le 6$ one can find duality defects $\CN_{ijk}$ of order $2$ such that
\be \CN_{ijk}^2= \CI+\CL_{t_it_j}+\CL_{t_jt_k}+\CL_{t_it_k}\ .
\ee 

\bigskip

\noindent{\bf Partition $1+1+2+2$.} Let us now consider the class of defects acting by
\be \rho_L=(\one\,\one\,\one\,\one\,\one\,\one)(34)(56)\ ,\qquad \rho_R=1\ ,
\ee on the chiral and anti-chiral algebra.  We denote such defects by $\CL^{(34)(56)}_a$, $a=1,2,\ldots$. Because $\rho_L$ act in the same way as $p_1\in G_{GTVW}$ on the chiral algebra, $\hat{\CL}^{(34)(56)}_a$ must preserve the $\CN=(4,4)$ SCA. The only representations where $\hat \CL^{(34)(56)}_a$ can be non-zero are the ones satisfying
\be\label{nonzerocond2} a_3=a_4\ ,\qquad a_5=a_6\ ,\qquad b_3=b_4\ ,\qquad b_5=b_6\ .
\ee Once again, we arrange such representations (in the NS-NS sector) in sets
\begin{align*}
    \Omega_1&=\{[000000;000000],\ [111111;000000],\ [000000;111111],\ [111111;111111] \}\ ,\\
    \Omega_2&=\{[110000;110000],\ [001111;110000],\ [110000;001111],\ [001111;001111] \}\ ,\\
    \Omega_3&=\{[001100;001100],\ [110011;001100],\ [001100;110011],\ [110011;110011] \}\ ,\\
    \Omega_4&=\{[000011;000011],\ [111100;000011],\ [000011;111100],\ [111100;111100] \}\ ,
\end{align*} so that
$$ \hat\CL^{(34)(56)}_a=\sum_{i=1}^4 \alpha_{i}P^i_{(34)(56)}\ .
$$ Because there are $4$ parameters $\alpha_i$, we expect (at most) $4$ simple defects of this kind. Let us denote by $\CL:=\CL^{(34)(56)}_1$ one of these defects. Then, by fusion of $\CL^{(34)(56)}$ with the invertible defects $\CL_{t_2t_3}$, $\CL_{t_2t_5}$, $\CL_{t_3t_5}$, we find three more simple defects $\CL^{(34)(56)}_2$, $\CL^{(34)(56)}_3$, $\CL^{(34)(56)}_4$ whose parameters  $\alpha_i$ differ only by signs. The product $\CL\CL^\dual=\CL^\dual\CL=\CI+\ldots$ acts trivially on the whole chiral and antichiral algebra, as well as on the $\CN=(4,4)$ superconformal algebra, and annihilate the representations that do not satisfy \eqref{nonzerocond2}. This leads to
\be \CL\CL^\dual=\CI+\CL_{t_3t_4}+\CL_{t_5t_6}+\CL_{t_3t_4t_5t_6}\ .
\ee It follows that
\be |\alpha_1|=|\alpha_2|=|\alpha_3|=|\alpha_4|=|\CL|=2\ .
\ee

\noindent
The torus partition function is 
$$Z_{\CL}(\tau)=\Tr(\hat\CL q^{L_0-\frac{c}{24}}\bar q^{\bar L_0-\frac{\bar c}{24}})=\sum_{i=1}^4 \alpha_i \Tr(P_{\rho_L,\rho_R}^i q^{L_0-\frac{c}{24}}\bar q^{\bar L_0-\frac{\bar c}{24}})$$
where 
\begin{align*}\Tr(P_{(34)(56)}^i q^{L_0-\frac{c}{24}}\bar q^{\bar L_0-\frac{\bar c}{24}})=\sum_{[a_1,\ldots,b_6]\in\Omega_i} \Bigl[&\ch_{a_1}(\tau,0)\ch_{a_2}(\tau,0)\ch_{a_3}(2\tau,0)\ch_{a_5}(2\tau,0)\\
&\overline{\bigl(\ch_{b_1}(\tau,0)\ch_{b_2}(\tau,0)\ch_{b_3}(\tau,0)^2\ch_{b_5}(\tau,0)^2\bigr)}\Bigr]
\end{align*} 
Using the ansatz $\alpha_i = \langle \mathcal{L}^{(34)(56)} \rangle$, the partition function with a $\CL^{(34)(56)}$ insertion can be written in terms of the Theta functions as:
{\footnotesize{
\begin{equation*}
\begin{split}
     Z_{\mathcal{L}^{(34)(56)}} (\tau) & =  \frac{\langle \CL^{(34)(56)} \rangle}{\overline{\eta (\tau)^6} \eta(\tau)^2 \eta (2 \tau)^2} \left\lbrace  \left( \overline{\theta_3 (2 \tau)^6} + \overline{\theta_2 (2 \tau)^6}\right) \left( \theta_3 (2 \tau)^2 \theta_3 (4 \tau)^2 + \theta_2 (2 \tau)^2 \theta_2 (4 \tau)^2 \right) +  \right. \\
    & \left. +  \left( \overline{\theta_2 (2 \tau)^2} \overline{\theta_3 (2 \tau)^4} + \overline{\theta_3(2 \tau)^2} \overline{\theta_2 (2 \tau)^4} \right) \left( \theta_2 (2 \tau)^2 \theta_3 (4 \tau)^2 + \theta_3 (2 \tau)^2 \theta_2 (4 \tau)^2 \right) + \right. \\
    & \left. + 2 \left( \overline{\theta_3 (2 \tau)^4} \overline{\theta_2 (2 \tau)^2} + \overline{\theta_3 (2 \tau)^2} \overline{\theta_2 (2 \tau)^4}\right) \left( \theta_3 (2 \tau)^2 \theta_3 (4 \tau) \theta_2 (4 \tau) + \theta_2 (2 \tau)^2 \theta_3 (4 \tau) \theta_2 (4 \tau) \right) \right\rbrace.
\end{split}
\end{equation*}
}}
Upon an $S$-transformation, we obtain the twisted function $Z^{\mathcal{L}^{(34)(56)}} (\tau) = Z_{\mathcal{L}^{(34)(56)}} \left( - 1 / \tau \right)$ in the form:
{\footnotesize{
\begin{equation*}
    \begin{split}
    Z^{\mathcal{L}^{(34)(56)}} (\tau) & =  \frac{\langle \CL^{(34)(56)} \rangle }{32 \overline{\eta (\tau)^6}\eta (\tau)^2 \eta (\tau/2)} \left\lbrace  \left(  \overline{\theta_3 \left( \frac{\tau}{2} \right)}^6  +  \overline{\theta_4 \left( \frac{\tau}{2}
\right)}^6 \right) \left( \theta_3 \left( \frac{\tau}{2} \right)^2  \theta_3 \left( \frac{\tau}{4} \right)^2 +  \theta_4 \left( \frac{\tau}{2} \right)^2  \theta_4 \left( \frac{\tau}{4} \right)^2 \right)  \right. \\
& \left. + \left( \overline{\theta_4 \left( \frac{\tau}{2} \right)}^2 \overline{\theta_3 \left( \frac{\tau}{2} \right)}^4 + \overline{\theta_3 \left( \frac{\tau}{2} \right)}^2 \overline{\theta_4 \left( \frac{\tau}{2} \right)}^4  \right) \left[ \left( \theta_4 \left( \frac{\tau}{2} \right)^2 \theta_3 \left( \frac{\tau}{4}\right)^2 + \theta_3 \left( \frac{\tau}{2} \right)^2 \theta_4 \left( \frac{\tau}{4}\right)^2  \right) + \right. \right. \\
& \left. \left. + 2 \left( \theta_3 \left( \frac{\tau}{4} \right) \theta_4 \left( \frac{\tau}{4} \right) \right) \left( \theta_3 \left( \frac{\tau}{2} \right)^2 + \theta_4 \left( \frac{\tau}{2} \right)^2 \right) \right] \right\rbrace, \\
    \end{split}
\end{equation*}
}}
whose expansion in $q$ and $\bar{q}$ and is:
{\footnotesize{
\begin{equation*}
    \begin{split}
   Z^{\mathcal{L}^{(34)(56)}} (\tau) & =  \frac{\langle\mathcal{L}^{(34)(56)} \rangle}{q^{1/8} \overline{q}^{1/4}} \left[  \left(\frac{1}{2}+5 q^{1/2}+\frac{47 q}{2}+75 q^{3/2}+\frac{403 q^2}{2}+501 q^{5/2}+1158 q^3+2502 q^{7/2}+\frac{10309}{2} q^4+ \right.\right. \\
   & \left. +10228 q^{9/2}+ O\left(q^{5}\right)\right)+\overline{q}^{1/2} \left(6+32 q^{1/4}+128 q^{3/8}+60 q^{1/2} +192 q^{3/4}+512 q^{7/8}+282 q+ \right. \\
   & \left. + 672 q^{5/4}+1792 q^{11/8}+900 q^{3/2}+1984 q^{7/4}+5120 q^{15/8}+ O \left(q^2 \right) \right) + \overline{q} \left(33+128 q^{1/4}+512 q^{3/8}+ \right.  \\
   & \left. \left. +330 q^{1/2}+768 q^{3/4}+2048 q^{7/8}+1551 q+2688 q^{5/4}+7168 q^{11/8}+O \left( q^{3/2} \right) \right) + O \left( \bar{q}^{3/2} \right) \right] \\
    \end{split}
\end{equation*}
}}
The minimal integer value required for the quantum dimension of $\mathcal{L}^{(34)(56)}$ in order to have just positive integer coefficients in this expansion is $\langle\mathcal{L}^{(34)(56)} \rangle=2$. The three additional simple defects, obtained by fusion with $\CL_{t_2 t_3}, \CL_{t_2 t_5}$ and $\CL_{t_3t_5}$, act on the representations $\Omega_i$ through the operators $ \pm 2 P^i_{(34)(56)}  $, where the minus sign appears for $i=2,3$, or $i=2,5$, or $i=3,5$, respectively. In all such cases, the $q$-expansion of the resulting $\CL$-twisted partition functions have positive integer coefficients, as expected. 

Because the coefficients $\alpha_i$ are real and $(\rho_L,\rho_R)$ has order two, we find that $\CL^\dual=\CL$ is unoriented, and
\be \CL^2=\CI+\CL_{t_3t_4}+\CL_{t_5t_6}+\CL_{t_3t_4t_5t_6}\ .
\ee This implies that $\CL^{(34)(56)}$ coincides with the duality defect $\CN_{34,56}\equiv \CL^{(34)(56)}$ related to the fact that the theory is self-orbifold with respect to the $\ZZ_2\times \ZZ_2$ subgroup of $\ZZ_2^8:M_{20}$ generated by $t_3t_4$ and $t_5t_6$.    Similarly, for every choice of pairwise distinct $i,j,k,l\in \{2,\ldots,6\}$, we have duality defects $\CN_{ij,kl}\equiv \CN_{kl,ij}$ of dimension $2$ such that
\be \CN_{ij,kl}^2=\CI+\CL_{t_it_j}+\CL_{t_kt_l}+\CL_{t_it_jt_kt_l}\ .
\ee Notice that, in a suitable description of the model $\calC_{GTVW}$ as a torus orbifold $T^4/\ZZ_2$, the symmetries $t_it_j$, $t_kt_l$ and $t_it_jt_kt_l$ can be identified, respectively, with $\eta_\frac{\lambda}{2}$, $Q\eta_\frac{\lambda}{2}$, and $Q$, where $\eta_\frac{\lambda}{2}$ is induced by the half-periods along one of the directions of $T^4$, and $Q$ is the quantum symmetry of the orbifold, see sections \ref{s:contDefects}. This suggests that $\CN_{ij,kl}$ can be identified with the topological defect $T_{\frac{\lambda}{4}}$ in eq.\eqref{lambdafour}.

\bigskip

\noindent{\bf Partition $1+5$.} Finally, let us consider the class of defects that act by the automorphism
\be \rho_L=(\one\,\one\,\Omega\,\Omega^{\dag}\,\one\,\one)(23645)\ ,\qquad \rho_R=1
\ee  of order $5$, preserving the $\CN=(4,4)$ supercurrents (in our conventions, the permutation acts after the $SU(2)^6$ transformation). The only possible for a simple defect $\CL^{(23645)}$ in this class is as follows (we set $\CL\equiv \CL^{(23645)}$ in this part)
\be \hat\CL=\langle \CL\rangle P_{(23645)}\ ,
\ee where $P_{(23645)}$ acts by the automorphism $(\rho_L,\rho_R)$ on the four NS-NS representations
$$[000000;000000],\ [111111;000000],\ [000000;111111],\ [111111;111111] $$
that satisfy
\be\label{nonzerocond5} a_2=a_3=a_4=a_5=a_6\ ,\qquad b_2=b_3=b_4=b_5=b_6
\ee
and annihilates any other representation in the NS-NS sector. The characters contributing to torus NS-NS partition function are: 
\be Z_{\CL^{(23645)}}(\tau)=\langle \CL\rangle\sum_{a,b=0}^1 \ch_{a}(\tau,0)\ch_{a}(5\tau,0)\overline{\ch_{b}(\tau,0)}^6\ ,
\ee
which expressions in terms of the $su(2)$ characters is:
{\footnotesize{
\begin{equation*}
 \begin{split}
     Z_{\CL^{(23645)}}(\tau)= & \langle \CL\rangle \frac{1}{\overline{\eta (\tau)^6}} \frac{1}{\eta (\tau) \eta (5 \tau)} \left( \overline{\theta_3 (2 \tau)}^6 + \overline{\theta_2 (2 \tau)}^6\right) \left( \theta_3 (2 \tau) \theta_3 (10 \tau) + \theta_2 (2 \tau) \theta_2 (10 \tau) \right)
 \end{split}   
\end{equation*}
}}
The twined partition function, obtained from the previous one  by modular transformation, takes the form:
{\footnotesize{
\begin{equation*}
    Z^{\CL^{(23645)}}(\tau)= \frac{\langle \CL\rangle}{16 \overline{\eta (\tau)}^6 \eta (\tau) \eta (\tau /5)} \left( \overline{\theta_3 \left( \frac{\tau}{2} \right)}^6 + \overline{\theta_4 \left( \frac{\tau}{2} \right)}^6 \right) \left( \theta_3 \left( \frac{\tau}{2} \right) \theta_3 \left(  \frac{\tau}{10}\right) + \theta_4 \left( \frac{\tau}{2} \right) \theta_4 \left(  \frac{\tau}{10}\right) \right),
\end{equation*}
}}
with corresponding series expansion given by:
{\footnotesize{
\begin{equation*}
    \begin{split}
        & Z^{\CL^{(23645)}} (\tau) =  \frac{\langle \CL\rangle}{q^{1/20} \overline{q}^{1/4}}  \left[ \left(\frac{1}{4}+\frac{3 }{4} q^{1/5}+q^{3/10}+q^{2/5}+q^{1/2}+\frac{7 }{4} q^{3/5}+3 q^{7/10}+\frac{13 }{4} q^{4/5}+4 q^{9/10}+\frac{11 }{2}q+ \right. \right. \\
        & \left. \left.+7 q^{11/10} + \frac{19 }{2} q^{6/5}+11 q^{13/10}+\frac{55 q^{7/5}}{4} q^{7/5}+18 q^{3/2}+\frac{83 }{4} q^{8/5}+26 q^{17/10}+\frac{129 }{4} q^{9/5}+39 q^{19/10}+O \left( q^2 \right) \right) + \right. \\
        &  \left. +\overline{q}^{1/2} \left(15+45 q^{1/5}+60 q^{3/10}+60 q^{2/5}+60 q^{1/2}+105 q^{3/5}+180 q^{7/10}+195 q^{4/5}+240 q^{9/10}+O \left( q \right) \right) + \right. \\
        & \left. + \overline{q} \left(\frac{129}{2}+\frac{387 }{2} q^{1/5}+258 q^{3/10}+258 q^{2/5}+ 258 q^{1/2}+\frac{903 }{2} q^{3/5}+774 q^{7/10}+\frac{1677 }{2} q^{4/5}+O \left( q^{9/10} \right) \right) + O(\overline{q}^{3/2}) \right] \\
        \end{split}
        \end{equation*}
}}
From such expansion we can extract the condition for the minimal value of the quantum dimension at $\langle \CL\rangle =4$.

The product $\CL\CL^\dual$ must act trivially on the whole chiral and antichiral algebra and annihilate any representation that does not satisfy \eqref{nonzerocond5}. Because the defect with smallest quantum dimension satisfying these properties has dimension $16$, this implies $\langle \CL\rangle=4$ and
\be \CL\CL^\dual=\CI +\sum_{2\le i<j\le 6} \CL_{t_it_j}+\sum_{i=2}^6 \CL_{t_it_2t_3t_4t_5t_6}\ .
\ee

The defects $\CL^{(26534)}$, $\CL^{(24356)}$, $\CL^{(25463)}\equiv \CL^\dual$ can be obtained by conjugation $\CL_g \CL^{(23645)}\CL_g^\dual$ with suitable invertible defects $\CL_g$, with $g\in \ZZ_2^8:M_{20}$, so that all such defects have dimension $4$, and are the unique simple defects acting with the given automorphism on the chiral algebra.

Let $p_2\in \ZZ_2^8:M_{20}$ be the order $2$ symmetry acting by $p_2=(\one\,\one\,\one\,\one\,\one\,\one\, ;\  \one\,\one\,\one\,\one\,\one\,\one)(35)(46)$, so that $\CL^\dual=\CL_{p_2} \CL\CL_{p_2}$. Then $(\CL_{p_2}\CL)^\dual=\CL_{p_2}\CL$ is unoriented, has dimension $4$, and
\be (\CL_{p_2}\CL)^2=\CI +\sum_{2\le i<j\le 6} \CL_{t_it_j}+\sum_{i=2}^6 \CL_{t_it_2t_3t_4t_5t_6}\ .
\ee This means that $\CN_{23456}:=\CL_{p_2}\CL$ is the duality defect related to the fact that the theory is self-orbifold under the $\ZZ_2^4$ group of symmetries generated by $t_it_j$, $2\le i<j\le 6$.

\section{The $\CN=2$ $c=1$ superconformal algebra as a free boson.}\label{a:Neq2freeboson}

The $\CN=2$ superconformal algebra at central charge $c_k=\frac{3k}{k+2}$, $k\in \NN$, can be described in terms of a coset $\frac{\widehat{su}(2)_k\oplus \widehat{u}(1)_{4}}{\widehat{u}(1)_{2k+4}}$, which provides the bosonic subalgebra of $\CN=2$. In the particular case of $k=1$, the bosonic subalgebra of the $\CN=2$ algebra, simplifies and becomes essentially the $\hat{u}(1)_{12}$ algebra \cite{Gray:2008je,Kiritsis:1987np,Waterson:1986ru}.
Here, $\hat{u}(1)_{2l}$ denotes\footnote{Our normalization differs by a factor $2$ from the conventions in \cite{DiFrancesco}: the algebra $\hat{u}(1)_{2l}$ in this paper is denoted as $\hat{u}(1)_{l}$ in \cite{DiFrancesco}. } the $c=1$ chiral algebra generated by a single chiral free boson $i\partial X_L(z)=\sum_n \alpha_nz^{-n-1}$ (whose modes generate the Heisenberg algebra), together with the holomorphic vertex operators $V_{n\sqrt{2l}}(z)\sim :e^{in\sqrt{2l}X_L(z)}:$, $n\in \ZZ$, of $\alpha_0$-eigenvalue $n\sqrt{2l}$ and conformal weight $n^2l$. In other words, this is the lattice VOA associated with the lattice $\sqrt{2l}\ZZ$. The irreducible representations $M^l_{[x]}$ of $\hat{u}(1)_{2l}$ are labeled by $[x]\in \ZZ/2l\ZZ$, and the characters are given by
\be K^{2l}_{[x]}(\tau,z)=\frac{\sum_{Q\in \frac{x}{2l}+\ZZ} q^{lQ^2}y^Q}{\eta(\tau)}
\ , \qquad y=e^{2\pi i z}\ ,
\ee where $Q$ is the $\alpha_0$-eigenvalue divided by $\sqrt{2l}$ (this strange normalization is justified below). In particular, the ground state of $M^l_{[x]}$ has conformal weight $h= \frac{x^2}{4l}$  and charge $Q=\frac{x}{2l}$, where $x\in\ZZ$ is a representative of $[x]\in \ZZ/2l\ZZ$ in the range $-l+1\le x\le l$.

The $\CN=2$ superconformal algebra at $c=1$ is obtained by adjoining the bosonic algebra $\hat{u}(1)_{12}$ (i.e., $l=6$) with its module $M^6_{[6]}$, which contains the two supercurrents and the other fermionic fields. The charge $Q$ is the $U(1)$ R-charge with the standard normalization such that the supercurrents have charge $\pm 1/2$. In general, the $\CN=2$ modules are given by sums $M^6_{[x]}\oplus M^6_{[x+6]}$ of $\hat{u}(1)_{12}$ modules, and the representations are NS or Ramond depending on whether $[x]$ is even or odd, respectively. 
In other words, the $\CN=2$ algebra at $c=1$ can be identified with the lattice SVOA related with the \emph{odd} lattice $\sqrt{3}\ZZ$. The representations of this superalgebra are given by $R_{[x]}=M^6_{[x]}\oplus M^6_{[x+6]}$ where now $x\in \ZZ/6\ZZ$, and correspond to the lattice cosets  $\frac{x}{2\sqrt{3}}+\sqrt{3}\ZZ\subset \frac{1}{2\sqrt{3}}\ZZ$.

Gepner models are usually described in terms of the coset algebra $(\hat{su}(2)_1\oplus \hat{u}(1)_4)/\hat{u}(1)_{6}$, whose representations are labeled as $[l,m,s]\equiv [1-l,m+3,s+2]$, $l=0,1$, $m\in\ZZ/6\ZZ$, $s\in \ZZ/4\ZZ$ with $l+m+s\equiv 0\mod 2$. As mentioned above, the algebra $(\hat{su}(2)_1\oplus \hat{u}(1)_4)/\hat{u}(1)_{6}$ is isomorphic to $\hat{u}(1)_{12}$, and the respective representations can be identified as follows:
\begin{align*}
    &[0,m,0] \quad \equiv \quad M^6_{[m]}\qquad m\in\{0,\pm 2\}\ ,\\
    &[0,m,2] \quad \equiv \quad M^6_{[m+6]}\qquad m\in\{0,\pm 2\}\\
    &[0,m,1] \quad \equiv \quad M^6_{[m]}\qquad m\in\{3,\pm 1\}\ ,\\
    &[0,m,-1] \quad \equiv \quad M^6_{[m+6]}\qquad m\in\{3,\pm 1\}\ .
\end{align*}
In terms of representations of the $\CN=2$ superconformal algebra, one has the following identifications
\begin{align*}
    &[0,m,0]\oplus [0,m,2]\quad \equiv \quad R_{[m]}\qquad m\in\{0,\pm 2\}\ ,\\
&[0,m,1]\oplus [0,m,-1]\quad \equiv \quad R_{[m]}\qquad m\in\{3,\pm 1\}\ .
\end{align*}

\newpage
\printbibliography

@article{Cardy:1986gw,
    author = "J.L. Cardy",
    title = "{Effect of Boundary Conditions on the Operator Content of Two-Dimensional Conformally Invariant Theories}",
    doi = "10.1016/0550-3213(86)90596-1",
    journal = "Nucl. Phys. B",
    volume = "275",
    pages = "200-218",
    year = "1986",
}

@article{Zuber:1986ng,
    author = "J-B. Zuber",
    title = "{Discrete Symmetries of Conformal Theories}",
    doi = "10.1016/0370-2693(86)90936-6",
    journal = "Nucl. Phys. B",
    volume = "176",
    pages = "127",
    year = "1986",
}

@article{Oshikawa:1996ww,
    author = "M. Oshikawa and I. Affleck",
    title = "{Defect lines in the Ising model and boundary states on orbifolds}",
    eprint = "hep-th/9606177",
    archivePrefix = "arXiv",
    primaryClass = "hep-th",
    doi = "10.1103/PhysRevLett.77.2604",
    journal = "Phys. Rev. Lett.",
    volume = "77",
    pages = "2604-2607",
    year = "1996",
}

@article{Oshikawa:1996dj,
    author = "M. Oshikawa and I. Affleck",
    title = "{Boundary conformal field theory approach to the critical two-dimensional Ising model with a defect line}",
    eprint = "cond-mat/9612187",
    archivePrefix = "arXiv",
    primaryClass = "cond-mat",
    doi = "10.1016/S0550-3213(97)00219-8",
    journal = "Nucl. Phys. B",
    volume = "495",
    pages = "533-582",
    year = "1997",
}

@article{Verlinde:1988sn,
    author = "E.P. Verlinde",
    title = "{Fusion Rules and Modular Transformations in 2D Conformal Field Theory}",
    doi = "10.1016/0550-3213(88)90603-7",
    journal = "Nucl. Phys. B",
    volume = "300",
    pages = "360-376",
    year = "1988",
}

@article{Petkova:2000ip,
    author = "V.B. Petkova and J-B. Zuber",
    title = "{Generalized twisted partition functions}",
    eprint = "hep-th/0011021",
    archivePrefix = "arXiv",
    primaryClass = "hep-th",
    doi = "10.1016/S0370-2693(01)00276-3",
    journal = "Phys. Lett. B",
    volume = "504",
    pages = "157-164",
    year = "2001",
}

@article{Fuchs:2002cm,
    author = "J. Fuchs and I. Runkel and C. Schweigert",
    title = "{TFT construction of RCFT correlators 1. Partition functions}",
    eprint = "hep-th/0204148",
    archivePrefix = "arXiv",
    primaryClass = "hep-th",
    doi = "10.1016/S0550-3213(02)00744-7",
    journal = "Phys. Lett. B",
    volume = "646",
    pages = "353-497",
    year = "2002",
}

@article{Fuchs:2003id,
    author = "J. Fuchs and I. Runkel and C. Schweigert",
    title = "{TFT construction of RCFT correlators 2. Unoriented world sheets}",
    eprint = "hep-th/0306164",
    archivePrefix = "arXiv",
    primaryClass = "hep-th",
    doi = "10.1016/j.nuclphysb.2003.11.026",
    journal = "Phys. Lett. B",
    volume = "678",
    pages = "511-637",
    year = "2004",
}

@article{Fuchs:2004dz,
    author = "J. Fuchs and I. Runkel and C. Schweigert",
    title = "{TFT construction of RCFT correlators 3. Simple currents}",
    eprint = "hep-th/0403157",
    archivePrefix = "arXiv",
    primaryClass = "hep-th",
    doi = "10.1016/j.nuclphysb.2004.05.014",
    journal = "Phys. Lett. B",
    volume = "694",
    pages = "277-353",
    year = "2004",
}

@article{Fuchs:2004xi,
    author = "J. Fuchs and I. Runkel and C. Schweigert",
    title = "{TFT construction of RCFT correlators IV: Structure constants and correlation functions}",
    eprint = "hep-th/0412290",
    archivePrefix = "arXiv",
    primaryClass = "hep-th",
    doi = "10.1016/j.nuclphysb.2005.03.018",
    journal = "Phys. Lett. B",
    volume = "715",
    pages = "539-638",
    year = "2005",
}

@article{Gaiotto:2014kfa,
    author = "D. Gaiotto and A. Kapustin and N. Seiberg and B. Willett",
    title = "{Generalized Global Symmetries}",
    eprint = "1412.5148",
    archivePrefix = "arXiv",
    primaryClass = "hep-th",
    doi = "10.1007/JHEP02(2015)172",
    journal = "JHEP",
    volume = "02",
    year = "2015",
}

@article{Bhardwaj:2017xup,
    author = "L. Bhardwaj and Y. Tachikawa",
    title = "{On finite symmetries and their gauging in two dimensions}",
    eprint = "1704.02330",
    archivePrefix = "arXiv",
    primaryClass = "hep-th",
    doi = "10.1007/JHEP03(2018)189",
    journal = "JHEP",
    volume = "03",
    pages = "189",
    year = "2018",
}

@article{Chang_2019,
   author={C-M. Chang and Y-H. Lin and S-H. Shao and Y. Wang and X. Yin},
   title={Topological defect lines and renormalization group flows in two dimensions},
    eprint = "1802.04445",
    archivePrefix = "arXiv",
    primaryClass = "hep-th",
    doi = "10.1007/JHEP01(2019)026",
    journal = "JHEP",
    volume = "01",
    year = "2019",
}

@article{Thorngren:2019iar,
   author={R. Thorngren  and Y. Wang },
   title={Fusion Category Symmetry I: Anomaly In-Flow and Gapped Phases},
    eprint = "1912.02817",
    archivePrefix = "arXiv",
    primaryClass = "hep-th",
    year = "2019",
}

@article{Thorngren:2021yso,
   author={R. Thorngren and Y. Wang },
   title={Fusion Category Symmetry II: Categoriosities at $c$ = 1 and Beyond},
    eprint = "2106.12577",
    archivePrefix = "arXiv",
    primaryClass = "hep-th",
    year = "2021",
}

@article{Carqueville:2023jhb,
   author={N. Carqueville and M. Del Zotto and I. Runkel },
   title={Topological defects},
    eprint = "2311.02449",
    archivePrefix = "arXiv",
    primaryClass = "math-ph",
    year = "2023",
}

@article{Kramers:1941kn,
   author={H.A. Kramers and G.H. Wannier},
   title={Statistics of the two-dimensional ferromagnet. Part 1.},
    doi = "10.1103/PhysRev.60.252",
    journal = "Phys. Rev.",
    volume = "06",
    pages ="252-262",
    year = "1941",
}

@article{Frohlich:2004ef,
    author = "J. Froehlich and J. Fuchs and I. Runkel and C. Schweigert",
    title = "{Kramers-Wannier duality from conformal defects}",
    eprint = "cond-mat/0404051",
    archivePrefix = "arXiv",
    primaryClass = "cond-mat",
    doi = "10.1103/PhysRevLett.93.070601",
    journal = "Phys. Rev. Lett.",
    volume = "93,070601",
    year = "2004",
}

@article{Frohlich:2006ch,
    author = "J. Froehlich and J. Fuchs and I. Runkel and C. Schweigert, Christoph",
    title = "{Duality and defects in rational conformal field theory}",
    eprint = "hep-th/0607247",
    archivePrefix = "arXiv",
    primaryClass = "hep-th",
    doi = "10.1016/j.nuclphysb.2006.11.017",
    journal = "Nucl. Phys. B",
    volume = "763",
    pages ="354-430",
    year = "2007",
}

@inproceedings{Frohlich:2009gb,
    author = "Froehlich, Jurg and Fuchs, Jurgen and Runkel, Ingo and Schweigert, Christoph",
    title = "{Defect lines, dualities, and generalised orbifolds}",
    booktitle = "{16th International Congress on Mathematical Physics}",
    eprint = "0909.5013",
    archivePrefix = "arXiv",
    primaryClass = "math-ph",
    reportNumber = "KCL-MTH-09-10, ZMP-HH-09-20",
    doi = "10.1142/9789814304634_0056",
    month = "9",
    year = "2009"
}

@article{Komargodski:2020mxz,
    author = "Z. Komargodski and K. Ohmori and K. Roumpedakis and S. Seifnashri ",
    title = "{Symmetries and strings of adjoint QCD$_{2}$}",
    eprint = "2008.07567",
    archivePrefix = "arXiv",
    primaryClass = "hep-th",
    doi = "10.1007/JHEP03(2021)103",
    journal = "JHEP",
    volume = "03,103",
    year = "2021",
}

@article{Aspinwall:1996mn,
    author = "P.S. Aspinwall",
    title = "{K3 surfaces and string duality}",
    eprint = "hep-th/9611137",
    archivePrefix = "arXiv",
    primaryClass = "hep-th",
    year = "1996",
}

@article{Nahm:1999ps,
    author = "W. Nahm and K. Wendland",
    title = "{A Hiker's guide to K3: Aspects of N=(4,4) superconformal field theory with central charge c = 6}",
    eprint = "hep-th/9912067",
    archivePrefix = "arXiv",
    primaryClass = "hep-th",
    journal = "Commun. Math. Phys.",
    volume = "216",
    pages="85-138",
    year = "2001",
}

@article{Eguchi:1987sm,
    author = "T. Eguchi and A. Taormina",
    title = "{Unitary Representations of $N=4$ Superconformal Algebra}",
    doi="10.1016/0370-2693(87)91679-0",
    journal = "Phys. Lett. B",
    volume = "196,75",
    year = "1987",
}

@article{Eguchi:1987wf,
    author = "T. Eguchi and A. Taormina",
    title = "{Character Formulas for the $N=4$ Superconformal Algebra}",
    doi="10.1016/0370-2693(88)90778-2",
    journal = "Phys. Lett. B",
    volume = "200, 315",
    year = "1988",
}

@article{Eguchi:1988af,
    author = "T. Eguchi and A. Taormina",
    title = "{On the Unitary Representations of $N=2$ and $N=4$ Superconformal Algebras}",
    doi="10.1016/0370-2693(88)90360-7",
    journal = "Phys. Lett. B",
    volume = "210",
    pages="125-132",
    year = "1988",
}

@article{Brunner:2013xna,
    author = "I. Brunner and N. Carqueville and D. Plencner" ,
    title ="{A quick guide to defect orbifolds}" ,
    eprint = "1310.0062",
    archivePrefix = "arXiv",
    primaryClass = "hep-th",
    doi="10.1090/pspum/088/01456",
    journal ="Proc. Symp. Pure Math." ,
    volume="88",
    pages= "231-242",
    year = "2014",
}

@article{Brunner:2014lua,
    author = "I. Brunner and N. Carqueville and D. Plencner" ,
    title ="{Discrete torsion defects}" ,
    eprint = "1404.7497",
    archivePrefix = "arXiv",
    primaryClass = "hep-th",
    doi="10.1007/s00220-015-2297-9",
    journal ="Commun. Math. Phys." ,
    volume="337",
    number="1",
    pages= "429-453",
    year = "2015",
}

@article{Ferrara:1995ih,
    author = "S. Ferrara and R. Kallosh and A. Strominger" ,
    title ="{N=2 extremal black holes}" ,
    eprint = "hep-th/9508072",
    archivePrefix = "arXiv",
    primaryClass = "hep-th",
    doi="10.1103/PhysRevD.52.R5412",
    journal ="Phys. Rev. D" ,
    volume="52",
    pages= "R5412-R5416",
    year = "1995",
}

@article{Andrianopoli:1998qg,
    author = "L. Andrianopoli and R. D'Auria and S. Ferrara and M.A. Lledo " ,
    title ="{Horizon geometry, duality and fixed scalars in six-dimensions}" ,
    eprint = "hep-th/9802147",
    archivePrefix = "arXiv",
    primaryClass = "hep-th",
    doi="10.1016/S0550-3213(98)00332-0",
    journal ="Nucl. Phys. B" ,
    volume="528",
    pages= "218-228",
    year = "1998",
}

@article{Moore:1998pn,
    author = "G.W. Moore" ,
    title ="{Arithmetic and attractors}" ,
    eprint = "hep-th/9807087",
    archivePrefix = "arXiv",
    primaryClass = "hep-th",
    year = "1998",
}

@article{Dijkgraaf:1998gf,
    author = "R. Dijkgraaf" ,
    title ="{Instanton strings and hyperKahler geometry}" ,
    eprint = "hep-th/9810210",
    archivePrefix = "arXiv",
    primaryClass = "hep-th",
    doi="10.1016/S0550-3213(98)00869-4",
    journal ="Nucl. Phys. B" ,
    volume="543",
    pages= "545-571",
    year = "1998",
}

@article{HohnMason2016,
    author = "G. Hoehn and G. Mason" ,
    title ="{The 290 fixed-point sublattices of the {L}eech lattice}" ,
    eprint = "1505.06420",
    archivePrefix = "arXiv",
    primaryClass = "math.GR",
    doi="10.1016/j.jalgebra.2015.08.028",
    journal ="J. Algebra" ,
    volume="448",
    pages= "618-637",
    year = "2016",
}

@article{Gaberdiel:2013psa,
    author = "M.R. Gaberdiel and A. Taormina and R. Volpato and K. Wendland" ,
    title ="{A K3 sigma model with $\mathbb{Z}^8_2$ : $\mathbb{M}_{20}$ symmetry}" ,
    eprint = "1309.4127",
    archivePrefix = "arXiv",
    primaryClass = "hep-th",
    doi="10.1007/JHEP02(2014)022",
    journal ="JHEP" ,
    volume="02,022",
    year = "2014",
}

@article{Chang:2022hud,
    author = "C.M. Chang and J. Cheng and F. Xu" ,
    title ="{Topological defect lines in two dimensional fermionic CFTs}" ,
    eprint = "2208.02757",
    archivePrefix = "arXiv",
    primaryClass = "hep-th",
    doi="10.21468/SciPostPhys.15.5.216",
    journal ="SciPost Phys." ,
    volume="15",
    number= "5, 216",
    year = "2023",
}

@article{Runkel:2022fzi,
    author = "I. Runkel and L. Szegedy and G.M.T. Watts" ,
    title ="{Parity and Spin CFT with boundaries and defects}" ,
    eprint = "2210.01057",
    archivePrefix = "arXiv",
    primaryClass = "hep-th",
    year = "2022",
}

@article{Fuchs:2007tx,
    author = "J. Fuchs and M.R. Gaberdiel and I. Runkel and C. Schweigert" ,
    title ="{Topological defects for the free boson CFT}" ,
    eprint = "0705.3129",
    archivePrefix = "arXiv",
    primaryClass = "hep-th",
    doi="10.1088/1751-8113/40/37/016",
    journal ="J. Phys. A" ,
    volume="40, 11403",
    year = "2007",
}

@article{Bachas:2007td,
    author = "C. Bachas and I. Brunner" ,
    title ="{Fusion of conformal interfaces}" ,
    eprint = "0712.0076",
    archivePrefix = "arXiv",
    primaryClass = "hep-th",
    doi="10.1088/1126-6708/2008/02/085",
    journal ="JHEP" ,
    volume="02, 085",
    year = "2008",
}

@article{Chang:2020imq,
    author = "C.M. Chang and Y.H. Lin" ,
    title ="{Lorentzian dynamics and factorization beyond rationality}" ,
    eprint = "2012.01429",
    archivePrefix = "arXiv",
    primaryClass = "hep-th",
    doi="10.1007/JHEP10(2021)125",
    journal ="JHEP" ,
    volume="10, 125",
    year = "2021",
}

@article{Bachas:2012bj,
    author = "C. Bachas and I. Brunner and D. Roggenkamp" ,
    title ="{A worldsheet extension of $O(d,d:Z)$}" ,
    eprint = "1205.4647",
    archivePrefix = "arXiv",
    primaryClass = "hep-th",
    doi="10.1007/JHEP10(2012)039",
    journal ="JHEP" ,
    volume="10, 039",
    year = "2012",
}

@article{Gaberdiel:2011fg,
    author = "M.R. Gaberdiel and S. Hohenegger and R. Volpato" ,
    title ="{Symmetries of K3 sigma models}" ,
    eprint = "1106.4315",
    archivePrefix = "arXiv",
    primaryClass = "hep-th",
    doi="10.4310/CNTP.2012.v6.n1.a1",
    journal ="Commun. Num. Theor. Phys." ,
    volume="6",
    pages="1-50",
    year = "2012",
}

@article{Eguchi:2010ej,
    author = "T. Eguchi and H. Ooguri and Y. Tachikawa" ,
    title ="{Notes on the K3 Surface and the Mathieu group $M_{24}$}" ,
    eprint = "1004.0956",
    archivePrefix = "arXiv",
    primaryClass = "hep-th",
    doi="10.1080/10586458.2011.544585",
    journal ="Experimental Mathematics" ,
    volume="20",
    pages="91-96",
    year = "2011",
}

@article{Brunner:2007qu,
    author = "I. Brunner and D. Roggenkamp" ,
    title ="{B-type defects in Landau-Ginzburg models}" ,
    eprint = "0707.0922",
    archivePrefix = "arXiv",
    primaryClass = "hep-th",
    doi="10.1088/1126-6708/2007/08/093",
    journal ="JHEP" ,
    volume="08",
    pages="093",
    year = "2007",
}

@article{Brunner:2009zt,
    author = "I. Brunner and D. Roggenkamp and S. Rossi" ,
    title ="{Defect Perturbations in Landau-Ginzburg Models}" ,
    eprint = "0909.0696",
    archivePrefix = "arXiv",
    primaryClass = "hep-th",
    doi="10.1007/JHEP03(2010)015",
    journal ="JHEP" ,
    volume="03",
    pages="015",
    year = "2010",
}

@article{Carqueville:2012st,
    author = "N. Carqueville and D. Murfet" ,
    title ="{Adjunctions and defects in Landau-Ginzburg models}" ,
    eprint = "1208.1481",
    archivePrefix = "arXiv",
    primaryClass = "math.AG",
    doi="10.1016/j.aim.2015.03.033",
    journal ="Adv. Math." ,
    volume="289",
    pages="480-566",
    year = "2016",
}

@article{Carqueville:2012dk,
    author = "N. Carqueville and I. Runkel" ,
    title ="{Orbifold completion of defect bicategories}" ,
    eprint = "1210.6363",
    archivePrefix = "arXiv",
    primaryClass = "math.AG",
    doi="10.4171/qt/76",
    journal ="Quantum Topol." ,
    volume="7, 02",
    pages="203-279",
    year = "2016",
}

@article{Behr:2020gqw,
    author = "N. Behr and S. Fredenhagen" ,
    title ="{Fusion of interfaces in Landau-Ginzburg models: a functorial approach}" ,
    eprint = "2012.14225",
    archivePrefix = "arXiv",
    primaryClass = "hep-th",
    doi="10.1007/JHEP04(2021)235",
    journal ="JHEP" ,
    volume="04",
    pages="235",
    year = "2021",
}

@article{Brunner:2020miu,
    author = "I. Brunner and I. Mayer and C. Schmidt-Colinet" ,
    title ="{Topological defects and SUSY RG flow}" ,
    eprint = "2007.02353",
    archivePrefix = "arXiv",
    primaryClass = "hep-th",
    doi="10.1007/JHEP03(2021)098",
    journal ="JHEP" ,
    volume="03",
    pages="098",
    year = "2021",
}

@article{Becker:2017zai,
    author = "M. Becker and Y. Cabrera and D. Robbins" ,
    title ="{Conformal interfaces between free boson orbifold theories}" ,
    eprint = "1706.03802",
    archivePrefix = "arXiv",
    primaryClass = "hep-th",
    doi="10.1007/JHEP09(2017)148",
    journal ="JHEP" ,
    volume="09",
    pages="148",
    year = "2017",
}

@article{Harvey:2020jvu,
    author = "J.A. Harvey and G.W. Moore" ,
    title ="{Moonshine, superconformal symmetry, and quantum error correction}" ,
    eprint = "2003.13700",
    archivePrefix = "arXiv",
    primaryClass = "hep-th",
    doi="10.1007/JHEP05(2020)146",
    journal ="JHEP" ,
    volume="05",
    pages="146",
    year = "2020",
}

@article{Dijkgraaf:1989pz,
    author = "R. Dijkgraaf and E. Witten" ,
    title ="{Topological gauge theories and group cohomology}" ,
    journal ="Commun. Math. Phys." ,
    volume="129",
    pages="393",
    year = "1990",
}

@article{Bantay:1990yr,
    author = "P. Bantay" ,
    title ="{Orbifolds and Hopf algebras}" ,
    journal ="Phys. Lett. B" ,
    volume="245",
    pages="477",
    year = "1990",
}

@article{Roche:1990hs,
    author = "R. Dijkgraaf and V. Pasquier and P. Roche" ,
    title ="{Quasi-Hopf algebras, group cohomology and orbifold models}" ,
    journal ="Nucl. Phys. Proc. Suppl." ,
    volume="18B",
    pages="60",
    year = "1990",
}

@article{Tachikawa:2017gyf,
    author = "Y. Tachikawa" ,
    title ="{On gauging finite subgroups}" ,
    eprint = "1712.09542",
    archivePrefix = "arXiv",
    primaryClass = "hep-th",
    doi="10.21468/SciPostPhys.8.1.015",
    journal ="SciPost Phys." ,
    volume="8, 01",
    pages="015",
    year = "2020",
}

@article{Gaberdiel:2012um,
    author = "M.R. Gaberdiel and R. Volpato" ,
    title ="{Mathieu Moonshine and Orbifold K3s}" ,
    eprint = "1206.5143",
    archivePrefix = "arXiv",
    primaryClass = "hep-th",
    doi="10.1007/978-3-662-43831-2\_5",
    journal ="Contrib. Math. Comput. Sci." ,
    volume="8",
    pages="109-141",
    year = "2014",
}

@article{deWildPropitius:1995cf,
    author = "M.D.F. de Wild Propitius" ,
    title ="{Topological interactions in broken gauge theories}" ,
    eprint = "hep-th/9511195",
    archivePrefix = "arXiv",
    primaryClass = "hep-th",
    year = "1995",
}

@article{Coste:2000tq,
    author = "A. Coste and T. Gannon and P. Ruelle" ,
    title ="{Finite group modular data}" ,
    eprint = "hep-th/0001158",
    archivePrefix = "arXiv",
    primaryClass = "hep-th",
    journal ="Nucl. Phys. B" ,
    volume="581",
    pages="679",
    year = "2000",
}

@article{Cheng:2010pq,
    author = "M.C.N. Cheng" ,
    title ="{K3 Surfaces, N=4 dyons, and the Mathieu group $M_{24}$}" ,
    eprint = "1005.5415",
    archivePrefix = "arXiv",
    primaryClass = "hep-th",
    doi="10.48550/arXiv.1005.5415",
    journal ="Commun. Number Theory Phys." ,
    volume="4, 623",
    year = "2010",
}

@article{Gaberdiel:2010ch,
    author = "M.R. Gaberdiel and S. Hohenegger and R. Volpato" ,
    title ="{Mathieu twining characters for K3}" ,
    eprint = "1006.0221",
    archivePrefix = "arXiv",
    primaryClass = "hep-th",
    doi="10.1007/jhep09(2010)058",
    journal ="JHEP" ,
    volume="2010,058",
    year = "2010",
}

@article{Gaberdiel:2010ca,
    author = "M.R. Gaberdiel and S. Hohenegger and R. Volpato" ,
    title ="{Mathieu Moonshine in the elliptic genus of K3}" ,
    eprint = "1008.3778",
    archivePrefix = "arXiv",
    primaryClass = "hep-th",
    doi="10.1007/jhep10(2010)062",
    journal ="JHEP" ,
    volume="2010, 062",
    year = "2010",
}

@article{Eguchi:2010fg,
    author = "T. Eguchi and K. Hikami" ,
    title ="{Note on twisted elliptic genus of K3 surface}" ,
    eprint = "1008.4924",
    archivePrefix = "arXiv",
    primaryClass = "hep-th",
    doi="10.1016/j.physletb.2010.10.017",
    journal ="Phys. Lett. B" ,
    volume="694,4",
    pages="446-455",
    year = "2011",
}

@article{Paquette:2017gmb,
    author = "N.M. Paquette and R. Volpato and M. Zimet " ,
    title ="{No More Walls! A Tale of Modularity, Symmetry, and Wall Crossing for 1/4 BPS Dyons}" ,
    eprint = "1702.05095",
    archivePrefix = "arXiv",
    primaryClass = "hep-th",
    doi="10.1007/JHEP05(2017)047",
    journal ="JHEP" ,
    volume="05",
    pages="047",
    year = "2017",
}

@article{Volpato:2014zla,
    author = "R. Volpato" ,
    title ="{On symmetries of $\mathcal{N}=(4,4)$ sigma models on $T^4$}" ,
    eprint = "1403.2410",
    archivePrefix = "arXiv",
    primaryClass = "hep-th",
    doi="10.1007/JHEP08(2014)094",
    journal ="JHEP" ,
    volume="08",
    pages="094",
    year = "2014",
}

@article{Atlas,
    author = "J.H. Conway and R.T. Curtis and S.P. Norton and R.A. Parker and R.A. Wilson" ,
    title ="{Atlas of finite groups}" ,
    journal ="Oxford University Press" ,
    year = "1985",
}

@article{nikulin,
    author = "V.V. Nikulin" ,
    title ="{Integer symmetric bilinear forms and some of their geometric applications}" ,
    journal ="Izv. Akad. Nauk SSSR Ser.Mat." ,
    volume="43, 111",
    year = "1979",
}

@article{ConwaySloane,
    author = "J.H. Conway and N.J.A. Sloane" ,
    title ="{Sphere packings, lattices and groups}" ,
    journal ="Grundlehren der Mathematischen Wissenschaften" ,
    volume="290",
    year = "1999",
}

@manual{GAP4,
	key          = "GAP",
	organization = "The GAP~Group",
	title        = "{GAP -- Groups, Algorithms, and Programming,
	Version 4.12.2}",
	year         = 2022,
	url          = "https://www.gap-system.org",
}

@article{DiFrancesco,
    author = "P. DiFrancesco and P. Mathieu and D. Senechal" ,
    title ="{Conformal Field Theory}" ,
    doi="10.1007/978-1-4612-2256-9",
    journal ="Springer-Verlag" ,
    volume="ISBN 978-0-387-94785-3, 978-1-4612-7475-9",
    year = "1997",
}

@article{Duncan:2015xoa,
    author = "J.F.R. Duncan and S. Mack-Crane" ,
    title ="{Derived Equivalences of K3 Surfaces and Twined Elliptic Genera}" ,
JOURNAL = {Res. Math. Sci.},
  FJOURNAL = {Research in the Mathematical Sciences},
    VOLUME = {3},
      YEAR = {2016},
    eprint = "1506.06198",
    archivePrefix = "arXiv",
    primaryClass = "math.RT",
    year = "2015",
DOI = {10.1186/s40687-015-0050-9},
}

@article{Cheng:2016org,
    author = "M.C.N. Cheng and S.M. Harrison and R. Volpato and M. Zimet" ,
    title ="{K3 String Theory, Lattices and Moonshine}" ,
    eprint = "1612.04404",
    archivePrefix = "arXiv",
    primaryClass = "hep-th",
   JOURNAL = {Res. Math. Sci.},
  FJOURNAL = {Research in the Mathematical Sciences},
    VOLUME = {5},
      YEAR = {2018},
    NUMBER = {3},
       DOI = {10.1007/s40687-018-0150-4},
}

@article{Waterson:1986ru,
    author = "G. Waterson" ,
    title ="{Bosonic Construction of an $N=2$ Extended Superconformal Theory in Two-dimensions}" ,
    doi="10.1016/0370-2693(86)91002-6",
    journal ="Phys. Lett. B" ,
    volume="171",
    pages="77-80",
    year = "1986",
}

@article{Kiritsis:1987np,
    author = "E.B. Kiritsis" ,
    title ="{The Structure of $N=2$ Superconformally Invariant 'Minimal' Theories: Operator Algebra and Correlation Functions}" ,
    doi="10.1103/PhysRevD.36.3048",
    journal ="Phys. Rev. D" ,
    volume="36",
    pages="3048",
    year = "1987",
}

@article{Gray:2008je,
    author = "O. Gray" ,
    title ="{On the complete classification of the unitary N=2 minimal superconformal field theories}" ,
    eprint = "0812.1318",
    archivePrefix = "arXiv",
    primaryClass = "hep-th",
    doi="10.1007/s00220-012-1478-z",
    journal ="Commun. Math. Phys." ,
    volume="312",
    pages="611-654",
    year = "2012",
}

@article{Strominger:1996sh,
    author = "A. Strominger and C. Vafa" ,
    title ="{Microscopic origin of the Bekenstein-Hawking entropy}" ,
    eprint = "hep-th/9601029",
    archivePrefix = "arXiv",
    primaryClass = "hep-th",
    doi="10.1016/0370-2693(96)00345-0",
    journal ="Phys. Lett. B" ,
    volume="379",
    pages="99-104",
    year = "1996",
}

@article{Dijkgraaf:1996it,
    author = "R. Dijkgraaf and E.P. Verlinde and H.L. Verlinde" ,
    title ="{Counting dyons in N=4 string theory}" ,
    eprint = "hep-th/9607026",
    archivePrefix = "arXiv",
    primaryClass = "hep-th",
    doi="10.1016/S0550-3213(96)00640-2",
    journal ="Phys. Lett. B" ,
    volume="484",
    pages="543-561",
    year = "1997",
}

@article{Jatkar:2005bh,
    author = "D.P. Jatkar and A. Sen" ,
    title ="{Dyon spectrum in CHL models}" ,
    eprint = "hep-th/0510147",
    archivePrefix = "arXiv",
    primaryClass = "hep-th",
    doi="10.1088/1126-6708/2006/04/018",
    journal ="JHEP" ,
    volume="04",
    pages="018",
    year = "2006",
}

@article{David:2006ji,
    author = "J.R. David and D.P. Jatkar and A. Sen" ,
    title ="{Product representation of Dyon partition function in CHL models}" ,
    eprint = "hep-th/0602254",
    archivePrefix = "arXiv",
    primaryClass = "hep-th",
    doi="10.1088/1126-6708/2006/06/064",
    journal ="JHEP" ,
    volume="06",
    pages="064",
    year = "2006",
}

@article{David:2006yn,
    author = "J.R. David and A. Sen" ,
    title ="{CHL Dyons and Statistical Entropy Function from D1-D5 System}" ,
    eprint = "hep-th/0605210",
    archivePrefix = "arXiv",
    primaryClass = "hep-th",
    doi="10.1088/1126-6708/2006/11/072",
    journal ="JHEP" ,
    volume="11",
    pages="072",
    year = "2006",
}

@article{David:2006ud,
    author = "J.R. David and D.P. Jatkar and A. Sen" ,
    title ="{Dyon spectrum in generic N=4 supersymmetric Z(N) orbifolds}" ,
    eprint = "hep-th/0609109",
    archivePrefix = "arXiv",
    primaryClass = "hep-th",
    doi="10.1088/1126-6708/2007/01/016",
    journal ="JHEP" ,
    volume="01",
    pages="016",
    year = "2007",
}

@article{David:2006ru,
    author = "J.R. David and D.P. Jatkar and A. Sen" ,
    title ="{Dyon Spectrum in N=4 Supersymmetric Type II String Theories}" ,
    eprint = "hep-th/0607155",
    archivePrefix = "arXiv",
    primaryClass = "hep-th",
    doi="10.1088/1126-6708/2006/11/073",
    journal ="JHEP" ,
    volume="11",
    pages="073",
    year = "2006",
}

@article{Dabholkar:2008zy,
    author = "A. Dabholkar and J. Gomes and S. Murthy" ,
    title ="{Counting all dyons in N =4 string theory}" ,
    eprint = "arXiv:0803.2692",
    archivePrefix = "arXiv",
    primaryClass = "hep-th",
    doi="10.1007/JHEP05(2011)059",
    journal ="JHEP" ,
    volume="05",
    pages="059",
    year = "2011",
}

@article{Dabholkar:2006xa,
    author = "A. Dabholkar and S. Nampuri" ,
    title ="{Spectrum of dyons and black holes in CHL orbifolds using Borcherds lift}" ,
    eprint = "hep-th/06030662",
    archivePrefix = "arXiv",
    primaryClass = "hep-th",
    doi="10.1088/1126-6708/2007/11/077",
    journal ="JHEP" ,
    volume="11",
    pages="077",
    year = "2007",
}

@article{Dabholkar:2007vk,
    author = "A. Dabholkar and D. Gaiotto and S. Nampuri" ,
    title ="{Comments on the spectrum of CHL dyons}" ,
    eprint = "hep-th/0702150",
    archivePrefix = "arXiv",
    primaryClass = "hep-th",
    doi="10.1088/1126-6708/2008/01/023",
    journal ="JHEP" ,
    volume="01",
    pages="023",
    year = "2008",
}

@article{Shih:2005uc,
    author = "D. Shih and A. Strominger and X. Yin" ,
    title ="{Recounting Dyons in N=4 string theory}" ,
    eprint = "hep-th/0505094",
    archivePrefix = "arXiv",
    primaryClass = "hep-th",
    doi="10.1088/1126-6708/2006/10/087",
    journal ="JHEP" ,
    volume="10",
    pages="087",
    year = "2006",
}

@article{Angius:2024xxx,
    author = "R. Angius and S. Giaccari and S.Harrison and R. Volpato",
    year="in preparation",
}

@article{Cheng:2014zpa,
    author = "M.C.N. Cheng and S. Harrison" ,
    title ="{Umbral Moonshine and K3 Surfaces}" ,
    eprint = "arXiv:1406.0619",
    archivePrefix = "arXiv",
    primaryClass = "hep-th",
    doi="10.1007/s00220-015-2398-5",
    journal ="Commun. Math. Phys." ,
    volume="339, 01",
    pages="221-261",
    year = "2015",
}

@article{Choi:2023xjw,
    author = "Choi, Yichul and Rayhaun, Brandon C. and Sanghavi, Yaman and Shao, Shu-Heng",
    title = "{Remarks on boundaries, anomalies, and noninvertible symmetries}",
    eprint = "2305.09713",
    archivePrefix = "arXiv",
    primaryClass = "hep-th",
    doi = "10.1103/PhysRevD.108.125005",
    journal = "Phys. Rev. D",
    volume = "108",
    number = "12",
    pages = "125005",
    year = "2023"
}

@article{Fuchs:2001qc,
    author = "Fuchs, Jurgen and Schweigert, Christoph",
    title = "{Category theory for conformal boundary conditions}",
    eprint = "math/0106050",
    archivePrefix = "arXiv",
    reportNumber = "PAR-LPTHE-01-24",
    journal = "Fields Inst. Commun.",
    volume = "39",
    pages = "25",
    year = "2003"
}

@article{Brunner:2007ur,
    author = "Brunner, Ilka and Roggenkamp, Daniel",
    title = "{Defects and bulk perturbations of boundary Landau-Ginzburg orbifolds}",
    eprint = "0712.0188",
    archivePrefix = "arXiv",
    primaryClass = "hep-th",
    doi = "10.1088/1126-6708/2008/04/001",
    journal = "JHEP",
    volume = "04",
    pages = "001",
    year = "2008"
}

@article{Fredenhagen:2009tn,
    author = "Fredenhagen, Stefan and Gaberdiel, Matthias R. and Schmidt-Colinet, Cornelius",
    title = "{Bulk flows in Virasoro minimal models with boundaries}",
    eprint = "0907.2560",
    archivePrefix = "arXiv",
    primaryClass = "hep-th",
    reportNumber = "AEI-2009-063",
    doi = "10.1088/1751-8113/42/49/495403",
    journal = "J. Phys. A",
    volume = "42",
    number = "49",
    pages = "495403",
    year = "2009"
}

@article{Kojita:2016jwe,
    author = "Kojita, Toshiko and Maccaferri, Carlo and Masuda, Toru and Schnabl, Martin",
    title = "{Topological defects in open string field theory}",
    eprint = "1612.01997",
    archivePrefix = "arXiv",
    primaryClass = "hep-th",
    doi = "10.1007/JHEP04(2018)057",
    journal = "JHEP",
    volume = "04",
    pages = "057",
    year = "2018"
}

@article{Konechny:2019wff,
    author = "Konechny, Anatoly",
    title = "{Open topological defects and boundary RG flows}",
    eprint = "1911.06041",
    archivePrefix = "arXiv",
    primaryClass = "hep-th",
    doi = "10.1088/1751-8121/ab7c8b",
    journal = "J. Phys. A",
    volume = "53",
    number = "15",
    pages = "155401",
    year = "2020"
}

@article{Konechny:2020jym,
    author = "Konechny, Anatoly",
    title = "{Properties of RG interfaces for 2D boundary flows}",
    eprint = "2012.12361",
    archivePrefix = "arXiv",
    primaryClass = "hep-th",
    doi = "10.1007/JHEP05(2021)178",
    journal = "JHEP",
    volume = "05",
    pages = "178",
    year = "2021"
}

@article{Fukusumi:2021zme,
    author = "Fukusumi, Yoshiki and Tachikawa, Yuji and Zheng, Yunqin",
    title = "{Fermionization and boundary states in 1+1 dimensions}",
    eprint = "2103.00746",
    archivePrefix = "arXiv",
    primaryClass = "hep-th",
    doi = "10.21468/SciPostPhys.11.4.082",
    journal = "SciPost Phys.",
    volume = "11",
    number = "4",
    pages = "082",
    year = "2021"
}

@article{Collier:2021ngi,
    author = "Collier, Scott and Mazac, Dalimil and Wang, Yifan",
    title = "{Bootstrapping boundaries and branes}",
    eprint = "2112.00750",
    archivePrefix = "arXiv",
    primaryClass = "hep-th",
    doi = "10.1007/JHEP02(2023)019",
    journal = "JHEP",
    volume = "02",
    pages = "019",
    year = "2023"
}

@article{Cordova:2023qei,
    author = "Cordova, Clay and Rizi, Giovanni",
    title = "{Non-Invertible Symmetry in Calabi-Yau Conformal Field Theories}",
    eprint = "2312.17308",
    archivePrefix = "arXiv",
    primaryClass = "hep-th",
    month = "12",
    year = "2023"
}

\end{document}